\documentclass[11pt]{amsart}

\pdfoutput = 1

\usepackage[english]{babel}
\usepackage[autostyle]{csquotes}
\usepackage{amsmath,amssymb,amsthm}
\usepackage{amsfonts}
\usepackage{amsxtra}
\usepackage{todonotes}

\usepackage{array,dcolumn}
\usepackage{graphicx}
\usepackage[hyperfootnotes=true,
        pdffitwindow=true,
        plainpages=false,
        pdfpagelabels=true,
        pdfpagemode=UseOutlines,
        pdfpagelayout=SinglePage,
                hyperindex,]{hyperref}

\usepackage{lmodern}
\usepackage[T1]{fontenc}
\usepackage[utf8]{inputenc}
\usepackage[activate={true,nocompatibility},spacing,kerning]{microtype}
\usepackage{bm}
\usepackage{bbm}
\usepackage[sc,osf]{mathpazo}
\microtypecontext{spacing=nonfrench}

  \calclayout

\newcommand{\mathsym}[1]{{}}
\newcommand{\unicode}[1]{{}}

\theoremstyle{plain}
\newtheorem{theorem}{Theorem}

\newtheorem{corollary}[theorem]{Corollary}
\newtheorem{proposition}[theorem]{Proposition}

\theoremstyle{definition}

\theoremstyle{remark}
\newtheorem{remark}[theorem]{Remark}

\newcommand{\R}{\mathbb R}
\newcommand{\C}{\mathbb C}

\renewcommand{\geq}{\geqslant}

\newcommand{\MeijerG}[8][\bigg]{G^{{ #2 },{ #3 }}_{{ #4 },{ #5 }} #1( \begin{matrix} #6 \\ #7 \end{matrix}\, #1\vert\, #8 #1)}

\numberwithin{equation}{section}

\numberwithin{theorem}{section}

\numberwithin{figure}{section}

\begin{document}


\title[]{Progress on the study of the Ginibre ensembles II:\\ G{\SMALL in}OE and G{\SMALL in}SE}

\author{Sung-Soo Byun}
\address{Center for Mathematical Challenges, Korea Institute for Advanced Study, 85 Hoegiro, Dongdaemun-gu, Seoul 02455, Republic of Korea}
\email{sungsoobyun@kias.re.kr}

\author{Peter J. Forrester}
\address{School of Mathematics and Statistics, 
University of Melbourne, Victoria 3010, Australia}
\email{pjforr@unimelb.edu.au}

\date{}


\begin{abstract}
This is part II of a review relating to the three classes of random non-Hermitian Gaussian matrices introduced by Ginibre in 1965. While part I restricted attention to the GinUE (Ginibre unitary ensemble) case of complex elements, in this part the cases of real elements (GinOE, denoting Ginibre orthogonal ensemble) and quaternion elements represented as $2 \times 2$ complex blocks (GinSE, denoting Ginibre symplectic ensemble) are considered. The eigenvalues of both GinOE and GinSE form Pfaffian point processes, which are more complicated than the determinantal point processes resulting from GinUE. Nevertheless, many of the obstacles that have slowed progress on the development of traditional aspects of the theory have now been overcome, while new theoretical aspects and new applications have been identified. This permits a comprehensive account of themes addressed too in the complex case: eigenvalue probability density functions and correlation functions, limit formulas for correlation functions, fluctuation formulas, sum rules, gap probabilities and eigenvector statistics, among others. Distinct from the complex case is the need to develop a theory of skew orthogonal polynomials corresponding to the skew inner product associated with the Pfaffian. Another distinct theme is the statistics of real eigenvalues, which is unique to GinOE. These appear in a number of applications of the theory, coming from areas as diverse as diffusion processes and persistence in statistical physics, topologically driven parametric energy level crossings for certain quantum dots, and equilibria counting for a system of random nonlinear differential equations.
\end{abstract}


\maketitle

\tableofcontents

\medskip 

\section{Introduction}\label{S1}
The 1965 paper of Ginibre ``Statistical ensembles of complex, quaternion and real matrices'' both isolated a distinguished class of non-Hermitian random matrices, and presented methods for their analysis. As the title
suggest, the elements of the random matrices --- which are all independent and chosen as mean zero, unit standard
deviation random variables --- can be  complex, quaternion or real numbers. In the quaternion
case, the $2 \times 2$ matrix representation
\begin{equation}\label{1.1}
\begin{bmatrix} z & w \\ - \bar{w} & \bar{z} \end{bmatrix}.
\end{equation}
involving two complex numbers is used. Thus in practical terms an $N \times N$ matrix with quaternion entries becomes a $2N\times
2N$ matrix with complex entries, which moreover has a special block structure. The Dyson index labels the quaternion Ginibre ensemble,
which is denoted GinSE, by $\beta = 4$ with the significance of four as specifying how many independent real numbers occur in (\ref{1.1}).
For the same reasons, the real Ginibre ensemble, denoted GinOE, is labelled $\beta = 1$ and the complex Ginibre ensemble, denoted
GinUE, is labelled $\beta = 2$. Using the Dyson index, the joint element distribution of a member $G$ of any of the three Ginibre
ensembles is seen to be proportional to
 \begin{equation}\label{1.1e}
\exp(-\beta {\rm Tr} \, G^\dagger  G  /2),
\end{equation} 
provided that in the quaternion case the convention that only independent terms occur in the trace --- the block structure implies
that all terms will be repeated twice.
In the notation GinSE,  S stands for symplectic and refers to the bi-unitary symplectic invariance of (\ref{1.1e}) being unchanged by the
mapping $G \mapsto UGV$ for $U,V$ unitary symplectic matrices as is seen from (\ref{1.1e}). For the same reason the O in GinOE
stands for orthogonal, and the U in GinUE stands for unitary.

A recent work by the present authors \cite{BF22a} has reviewed a number of themes and addition topics relating to GinUE. The chosen themes were eigenvalue probability density functions and correlation functions, fluctuation formulas, 
as well as sum rules and asymptotic behaviours of correlation functions, 
and normal matrix models. The additional topics included applications in quantum many body physics and quantum chaos, and
statistical properties of the eigenvectors. A significant structural feature of GinUE is that the eigenvalues
form a determinantal point process. This is not true of either GinOE or GinSE. On the other hand the eigenvalues
of the latter both form a Pfaffian point process. This distinction is one reason why it makes sense to review
GinUE separately to both GinOE and GinSE. Here we take up the task of viewing progress on the study of
GinOE and GinSE.

In the introduction to \cite{BF22a} it was remarked that the first occurrence of any of the Ginibre ensembles in applications was in fact in relation to GinOE. Thus in the 1972 study by May \cite{Ma72a}, the first order linear matrix differential equation
\begin{equation}\label{1.1h}
 {d \over dt} \mathbf x = (- \mathbb I + \alpha G)  \mathbf x,
 \end{equation} 
 where $\alpha$ is a scalar parameter and $G$ a GinOE matrix was encountered. The question of interest
 was in relation to the stability of the solution. This is determined by the maximum of the real part of the spectrum of $G$.
 In fact precise knowledge of this  quantity has only recently become available in the literature \cite{AP14,CESX22}.
 
 The understanding of other features of GinOE have similarly fallen dormant for long periods before progress
 was made. 
 A case in point is the joint eigenvalue probability density function (PDF). 
 In the cases of the GinUE and GinSE this was calculated in Ginibre's original paper as being proportional to
 \begin{equation}\label{1.1f}
\prod_{l=1}^N e^{-  | z_l|^2 } \prod_{1 \le j < k \le N} | z_k - z_j |^2,
\end{equation} 
and
 \begin{equation}\label{1.1g}
\prod_{l=1}^N e^{-  2| z_l|^2 } | z_l - \bar{z}_l |^2  \prod_{1 \le j < k \le N} | z_k - z_j |^2 | z_k - \bar{z}_j|^2, \quad {\rm Im} \, z_l > 0
\end{equation} 
respectively.
The case of GinOE is more complicated than for the GinUE or GinSE, and was not solved in \cite{Gi65}.
The extra complexity is because in the real case there is a non-zero probability of some eigenvalues being real. Consequently the
joint eigenvalue PDF consists of disjoint sectors depending on the number of real eigenvalues, $k$ say
($k$ must be of the same parity as $N$). Only the functional form of the eigenvalue PDF in the sector with
all eigenvalues real could be determined in \cite{Gi65}. The solution for general $k$ had to wait for 
another quarter of a century or so, at the hands of Lehmann and Sommers \cite{LS91}, followed a few years later by an independent
calculation of Edelman \cite{Ed97}. To present this, define the normalisation and the weight by
\begin{equation}\label{1.1ga}
C_N^{\rm g} = {1 \over 2^{N(N+1)/4} \prod_{l=1}^N \Gamma(l/2) }, \qquad
\omega^{\rm g}(z) = e^{-|z|^2}
e^{2 y^2} {\rm erfc(\sqrt{2}y)}
\end{equation}
respectively, where $z=x+iy$. 
Note that with $z=x$ and thus real, the weight simplifies
 $\omega^{\rm g}(x)= e^{-x^2}$.
 The joint eigenvalue PDF for $k$ real eigenvalues $\{\lambda_l\}_{l=1,\dots,k} $
and the $(N-k)/2$ complex eigenvalues $\{ x_j + i y_j \}_{j=1,\dots,(N-k)/2}$
in the upper half plane (note that the remaining $(N-k)/2$ complex eigenvalues
are the complex conjugate of these and so not independent) is then given by
\begin{align}
 C_N^{\rm g} \frac{   2^{(N-k)/2} }{  k! ((N-k)/2)! } & \prod_{s=1}^k (\omega^{\rm g}(\lambda_s))^{1/2}  \prod_{j=1}^{(N-k)/2} \omega^{\rm g}(z_j)   \nonumber
\\
&\times \Big | \Delta(\{\lambda_l\}_{l=1,\dots,k} \cup \{ x_j \pm i y_j \}_{j=1,\dots,(N-k)/2}) \Big |, \label{3.1}
\end{align}
where $\Delta(\{z_p\}_{p=1,\dots,m}) := \prod_{j < l}^m (z_l - z_j)$. 
Here $\lambda_l \in (-\infty, \infty)$ while $(x_j,y_j) \in \mathbb R \times {\mathbb R}_+$, ${\mathbb R}^2_+ := \{ (x,y) \in {\mathbb R}^2 : \, y>0 \}$. 

Of the themes considered in \cite{BF22a}, the one on eigenvalue probability density functions and correlation functions shows the most complete analogy between GinUE, and GinOE and GinSE. This is notwithstanding the already mentioned fact that the eigenvalues of GinUE form a determinantal point process, while those of GinOE and GinSE form a Pfaffian point process. Nor the fact that eigenvalue PDF of in the GOE is not absolutely continuous, but rather according to (\ref{3.1}) divides into sector depending on the number of real eigenvalues. We know in the theory of GinUE that there are elliptic and induced extensions, as well as one giving rise to a spherical ensemble, one coming from truncating a Haar distributed unitary matrices, and a product ensemble of GinUE matrices or truncated Haar unitary matrices. These are distinguished by having an explicit formula for the joint eigenvalue PDF, with the correlation kernel relating to certain special functions, the properties of which allow for a detailed asymptotic analysis. We will see that each of these ensembles has a counterpart in the theory of GinOE and GinSE, which furthermore permit detailed asymptotic analysis making use of the same classes of special functions. To varying degrees of generality, fluctuation formulas, gap probabilities, sum rules, asymptotic expansion of the partition function and eigenvector statistics, discussed in \cite{BF22a} for GinUE, are again amenable to exact analysis.

Topics distinct from those seen in \cite{BF22a} show themselves. The fact that the eigenvalues form a Pfaffian point process is one reason for this. Thus  associated with the Pfaffian structure is an underlying skew inner product, and associated skew polynomials. In GinOE theory, essential use is made of their form as a matrix average (\ref{12.212f}). But for GinSE, a relation with the corresponding GinUE orthogonal polynomials turns out to have a wider scope. Asymptotic analysis is more challenging for GinSE, with a technique based on differential equations found to be powerful. The topic of real eigenvalues is unique to GinOE. The corresponding statistics have features distinct to those exhibited in GinUE studies. They also provide for a number of applications.

Sections \ref{S2} to \ref{S2a} relate to GinOE, and Sections \ref{S3} and \ref{S6} to GinSE.

\bigskip  

\section{Eigenvalue statistics for GinOE and elliptic GinOE}\label{S2}
Where appropriate, our approach will be to summarise the main steps in the GinUE analogues --- these have for the most part been presented in our review \cite{BF22a} --- then to outline the required modification needed in the GinOE case.
\subsection{Eigenvalue PDF for GinOE}\label{S2.1}
Dyson's derivation of the GinUE eigenvalue PDF consisted of the following mains steps (see \cite[\S 2.1]{BF22a}):
\begin{itemize}
    \item[(i)] Decompose a GinUE matrix $G$ using the Schur decomposition $G = U Z U^\dagger$. Here $U$ is a unitary matrix, and $Z$ is an upper triangular matrix with the eigenvalues $\{z_j\}$ of $G$ on the diagonal.
    \item[(ii)] Decompose the measure for the independent elements of $G$, both real and imaginary parts, using the coordinates from (i). It is found that the dependence on $U$, $\tilde{Z}$ (the strictly upper triangular elements of $Z$) and $\{z_j\}$ factorises, and that the Jacobian equals $\prod_{j<k} |z_k - z_j|^2$.
    \item[(iii)] Rewrite the weight for the joint element PDF using the coordinates of (i), $e^{-{\rm Tr} \, G^\dagger G} = e^{-\sum_{j=1}^N |z_j|^2 - \sum_{j<k} |\tilde{Z}_{jk}|^2}$. 
    \item[(iv)] From the decomposition in (ii) and (iii) observe that the dependence on the elements of $\tilde{Z}$ factorises and hence only contributes to the normalisation after integration over these variables to leave the functional form (\ref{1.1f}).
\end{itemize}

Following \cite{Ed97} (see also \cite[\S 15.10]{Fo10}), in the case of $A \in {\rm GinOE}$, conditioned to have $k$ real eigenvalues ($k$ same parity as $N$), the appropriate Schur decomposition in step (i) reads $A = QR Q^T$. Here $Q$ is a real orthogonal matrix, while $R$ is upper block triangular. The first $k$ diagonal elements of $R$ are the scalars $\{ \lambda_j\}_{j=1}^k$ --- the real eigenvalues --- while the next $(N-k)/2$ diagonal elements are the $2 \times 2$ matrices $$ X_j:=\begin{bmatrix} x_j & b_j \\
-c_j & x_j \end{bmatrix}, \qquad b_j, c_j >0,$$ where $x_j \pm i y_j$ ($y_j = \sqrt{b_j c_j}$) are the complex eigenvalues.

Step (ii) seeks to decompose the measure for the elements of $A$. A factorisation again results, with the Jacobian equalling $2^{(N-k)/2} | \tilde{\Delta}|$, where $\tilde{\Delta}$ is as in (\ref{3.1}) but with the difference between each pair $x_j \pm i y_j$ omitted, times the additional factor $\prod_{l=k+1}^{(N+k)/2} | b_l - c_l|$. For step (iii) we calculate
$$
e^{-{\rm Tr} \, A A^T/2} =
e^{- \sum_{i < j} r_{ij}^2/2}
e^{- \sum_{j=1}^k \lambda_j^2/2}
e^{-\sum_{j=1}^{(N-k)/2}(x_j^2 + y_j^2+\delta_j^2/2)},
$$
where $\delta_j = b_j - c_j$ and $\{ r_{ij} \}$ are the off diagonal elements. To carry out step (iv) and thus integrate out all variables except the eigenvalues, it is convenient to change variables from $\{b_j,c_j\}$ to $\{y_j,\delta_j\}$ according to
$
db_j dc_j = (2 y_j / \sqrt{\delta_j^2 + 4 y_j^2}) dy_j d \delta_j.
$
Integrating over $\delta_j$ in this is responsible for the factors $e^{y_j^2} {\rm erfc}(\sqrt{2} y_j)$ in (\ref{3.1}), while integrating over the other variables only contributes to the normalisation.
\subsection{Coulomb gas perspective}\label{S2.1a}
The factor $|\Delta|$ in (\ref{3.1}) can be written in exponential form
\begin{multline}\label{SAS}
 \Big | \Delta(\{\lambda_l\}_{l=1,\dots,k} \cup
\{ x_j \pm i y_j \}_{j=1,\dots,(N-k)/2}) \Big | \\
 =
\exp \Big (  \sum_{1 \le j < p \le k} \log | \lambda_p - \lambda_j | +
\sum_{j=1}^k \sum_{s=1}^{(N-k)/2} \log |z_s - \lambda_j|   | \bar{z}_s - \lambda_j| +
\sum_{a,b = 1}^{(N-k)/2} \log |z_a - \bar{z}_b| \Big ) \\
\times
\exp \Big (  \sum_{1 \le a < b \le (N-k)/2} \log | z_b - z_a|  | \bar{z}_b - \bar{z}_a | \Big ).
\end{multline}
This permits the interpretation as a Boltzmann factor for a classical two-dimensional Coulomb gas 
\cite{Fo16,GPTW16}. Relevant to this is the fact that
 the solution of the two-dimensional Poisson equation $\nabla^2_{\vec{r}} \phi(\vec{r}, \vec{r}\,') = - 2\pi \delta(\vec{r} - \vec{r}\,')$
with the Neumann boundary condition along the $x$-axis
$
{\partial \over \partial y}  \phi(\vec{r}, \vec{r}\,') |_{y \to 0^+} = 0$
is, with the use of complex coordinates, given by
\begin{equation}\label{pp}
 \phi(\vec{r}, \vec{r}\,') = - \log \Big ( | z - z'| \, |z -\bar{z}\,'| \Big ).
 \end{equation}
 How to approximate (\ref{pp}) in terms of different dielectric constants for $y>0$ and $y<0$ is
discussed in \cite[\S 15.9]{Fo10}; its effect is to  give
rise to an image particle of identical charge at the reflection point $\bar{z}\,'$ of $z\,'$ about the real axis. 

We see that (\ref{SAS}) contains terms corresponding to the sum over pairs of the potential (\ref{pp}) for the complex coordinates $\{z_j\}_{j=1}^{(N-k)/2}$ in the upper half plane, interacting at dimensionless inverse temperature $\beta = 2$. In addition there is an interaction energy between $k$ real coordinates and the complex coordinates. However this is only consistent with (\ref{pp}) if the real eigenvalues are weighted by assigning their charge to equal $1/2$. Assuming this, after noting from  (\ref{pp}) that before weighting the $k$ real coordinates themselves interact via the pair potential $-\log |\lambda - \lambda'|^2$, the correct term in the Boltzmann factor is obtained.

We turn our attention now to the one body terms involving the weight in (\ref{3.1}), which when written out in full read 
\begin{equation}\label{see}
e^{- \sum_{j=1}^k \lambda_j^2/2} e^{- \sum_{j=1}^{(N-k)/2}(x_j^2 + y_j^2)}
\prod_{j=1}^{(N-k)/2} e^{2 y_j^2} {\rm erfc}(\sqrt{2} y_j).
\end{equation}
From a two-dimensional Coulomb gas viewpoint, 
we see that the first two exponential terms  result from a coupling between the charges (particles on the real line having charge $1/2$) 
confined to the semi-disk $|z| < \sqrt{N}$, $y > 0$, 
and with a neutralising
background of uniform density $\rho = 1/\pi$ filling the semi-disk. In the random matrix problem, this implies that the global scaled eigenvalues obey the circular law \cite[Eq.~(2.17)]{BF22a}.
The final term can be interpreted as the coupling of the complex coordinates to a smeared out charge on the real axis; asymptotically each factor decays as $1/y_j$. This is then cancelled by the term in (\ref{SAS}) corresponding to the interaction between the complex coordinate and its image.

\subsection{Generalised partition function and probabilities}
Integrating (\ref{3.1}) over the specified range of the eigenvalues gives the probability that a GinOE matrix has precisely $k$ real eigenvalues, $p_{k,N}^{\rm GinOE}$ say. However to perform the integrations in general, more theory is required --- that of skew orthogonal polynomials --- which is to be developed below. An exception is the case $k=N$ \cite{Ed97}.

\begin{proposition}\label{P2.1r}
We have $p_{N,N}^{\rm GinOE} = 2^{-N(N-1)/4}$.
\end{proposition}

\begin{proof}
According to (\ref{3.1}) with $k=N$,
$$
p_{N,N}^{\rm GinOE} =
{C_N^{\rm g} \over N!}  
\int_{-\infty}^\infty d \lambda_1 \cdots \int_{-\infty}^\infty d \lambda_N \,
e^{-\sum_{j=1}^N \lambda_j^2/2} 
\prod_{1 \le j < k \le N} | \lambda_k - \lambda_j|.
$$
This multiple integral --- the $\beta = 1$ case of what is referred to as Mehta's integral --- has a known evaluation in terms of products of gamma functions (see \cite[Prop.~4.7.1]{Fo10}),
which implies the result.
\end{proof}

Pfaffian structures associated with integrations over (\ref{3.1}) are most readily revealed by considering the so-called generalised partition function
\begin{equation}\label{GP}
Z_{k,(N-k)/2}[u,v] = \Big \langle
\prod_{l=1}^k u(\lambda_l) \prod_{l=1}^{(N-k)/2} v(x_l,y_l) \Big \rangle.
\end{equation}

\begin{proposition}\label{P2.2}
Let $\{p_{l-1}(x) \}_{l=1,\dots,N}$ be a set of monic polynomials, with $p_{l-1}(x)$ of degree
$l-1$. With the weight $\omega^{\rm g}(x)$ as in (\ref{1.1ga}) let
\begin{align*}
\alpha_{j,k}^{\rm g} & =  \int_{-\infty}^\infty dx \, u(x) 
\int_{-\infty}^\infty dy \, u(y) \,
(\omega^{\rm g}(x) \omega^{\rm g}(y))^{1/2} p_{j-1}(x) p_{k-1}(y) {\rm sgn} \, (y - x),  \\
\beta_{j,k}^{\rm g} & =  2 i \int_{{\mathbb R}_+^2} dx dy \, v(x,y) \omega^{\rm g}(z) \Big ( p_{j-1}(x+iy) p_{k-1}(x-iy) - p_{k-1}(x+iy) p_{j-1}(x-iy) \Big ).
\end{align*}
For $k,N$ even and with $C_N^{\rm g}$ as in (\ref{1.1ga}) we have
\begin{equation}\label{12.Z}
Z_{k,(N-k)/2}[u,v] = C_N^{\rm g}
[ \zeta^{k/2} ] {\rm Pf} [ \zeta \alpha_{j,l}^{\rm g} + \beta_{j,l}^{\rm g}  ]_{j,l=1,\dots,N},
\end{equation}
 where $[\zeta^p] f(\zeta)$ denotes the coefficient of $\zeta^p$ in $f(\zeta)$. 
With $Z_N[u,v] := \sum_{k=0}^{N/2}  Z_{2k,(N-2k)/2}[u,v]$, it follows
\begin{equation}\label{12.Ca}
Z_N[u,v] = C_N^{\rm g}
 {\rm Pf} [  \alpha_{j,k}^{\rm g} + \beta_{j,k}^{\rm g}  ]_{j,k=1,\dots,N}.
\end{equation}
\end{proposition}

\begin{proof}
The strategy to deduce (\ref{12.Z})
is to combine two integration methods involving determinants due to de Bruijn \cite{dB55}. Details can be found in \cite[proof of Prop.~15.10.3]{Fo10}.
\end{proof}

\begin{remark} 
The case $u=v$ of (\ref{12.Ca}) is due to Sinclair \cite{Si06}, who also considered the modification required for $N$ odd. Subsequently, the $N$ odd case of GinOE was considered in more detail in \cite{FM09,BS09}. Introducing
$
\nu_k := \int_{-\infty}^\infty e^{-x^2/2}p_{k-1}(x) \, dx,
$
the $N$ odd version of (\ref{12.Ca}) reads
\begin{equation}\label{12.CaO}
Z_N[u,v] = C_N^{\rm g} {\rm Pf} \bigg [ \begin{array}{cc} {}[\alpha_{j,k} + \beta_{j,k}]  & [\nu_j] \\
 {}[- \nu_l] & 0 \end{array}
 \bigg ]_{j,k=1,\dots,N}.
\end{equation} 
However below, for efficiency of presentation, we will always take $N$ to be even as we further develop theory relating to GinOE.
\end{remark}

\smallskip

Define $Z_N(\zeta) := \sum_{k=0}^{N/2}  \zeta^k Z_{2k,(N-2k)/2}[1,1]$, which is
the generating function for the probabilities $p_{2k,N}^{\rm GinOE} = Z_{2k,(N-2k)/2}[1,1]$ of there being exactly $2k$ real eigenvalues. 
Choosing $p_j(x)$ to be even for $j$ even and odd for $j$ odd, the Pfaffian in (\ref{12.Ca}) can then be written as a determinant of half the original size \cite{AK07},
\begin{equation}\label{12.212a}
Z_N(\zeta) = C_N^{\rm g} \det \Big [ \alpha_{2j-1,2k} |_{u=1}  + \beta_{2j-1,2k}|_{v=1} \Big ]_{j,k=1,\dots,N/2}.
\end{equation}
Moreover, further refining the choice of the $p_j$ to be skew orthogonal --- see the following subsection for this notion and their expansion as monomials --- allows (\ref{12.212a}) to be written in the explicit form \cite{KPTTZ15}
\begin{equation}\label{12.212b}
Z_N(\zeta) = \det \bigg [
\delta_{j,k} + {(\zeta - 1) \over \sqrt{2 \pi}}
{\Gamma(j+k-3/2) \over \sqrt{\Gamma(2j-1)\Gamma(2k-1)}}
\bigg ]_{j,k=1,\dots,N/2}.
\end{equation}
As an application, the first term of the conjectured asymptotic formula
\begin{equation}\label{2.10}
\frac1{\sqrt{N}}\log p_{N,0}^{\rm GinOE}=-\frac{1}{\sqrt{2\pi}}\zeta\Big(\frac32\Big)
+\frac C{\sqrt N}+\cdots,
\end{equation}
where  $\zeta(x)$ denotes the Riemann zeta function and
\begin{equation}\label{2.11}
C=\log 2-\frac14+\frac1{4\pi}\sum_{n=2}^\infty\frac1n\Big(-\pi+\sum_{p=1}^{n-1}\frac{1}{p(n-p)}\Big)\approx 0.0627,
\end{equation}
as implied by results of \cite{Fo15} for the probability of a large gap in the real spectrum of bulk scaled GinOE, was rigorously established. A known arithmetic property of $\{ p_{2k,N}^{\rm GUE} \}$, namely that each member of the sequence is of the form $p_{2k,N} = r + s \sqrt{2}$, with $r$ and $s$ rational numbers consisting of powers of $2$ in the denominator \cite{Ed97}, is (after minor manipulation) also evident from (\ref{12.212b}).

Another application of (\ref{12.212b}) is to differentiate with respect to $\zeta$ and set $\zeta =1$. This gives for the expected number of real eigenvalues, $E_N^{\rm r} := \sum_{k=0}^{N/2} 2 k p_{2k,N}$, the explicit formulas \cite{EKS94}
\begin{equation}\label{12.212c}
E_N^{\rm r}= \sqrt{2 \over  \pi} \sum_{k=1}^{N/2} {\Gamma(2k-3/2) \over \Gamma(2k-1)} =
{1 \over 2}+\sqrt{2 \over \pi}
{\Gamma(N+1/2) \over \Gamma(N)} \,
{}_2 F_1\bigg ( {1,-1/2 \atop N} \bigg | {1 \over 2} \bigg ),
\end{equation}
where the validity of the hypergeometric expression can be checked by recurrence.
The latter, which  holds too for $N$ odd, has the utility of implying the large $N$ asymptotic expansion
\begin{equation}\label{12.212d}
E_N^{\rm r} - {1 \over 2} \mathop{\sim}\limits_{N \to \infty}
\sqrt{2N \over \pi} \bigg ( 1 
  -
{3 \over 8 N} - {3 \over 128 N^2} + \cdots \bigg ).
\end{equation}
We will see later (working below Proposition \ref{P2.9}) that the leading order value $\sqrt{2N/\pi}$ is consistent with the bulk density of real eigenvalues equalling the value $ 1 / \sqrt{2 \pi}$, and being supported on the interval $[-\sqrt{N}, \sqrt{N}]$ to leading order. Relating to  this, we already know from the Coulomb gas viewpoint of Section \ref{S2.1a} that the density of complex eigenvalues is $1/\pi$ supported on the disk of radius $\sqrt{N}$.
We also mention that the leading order asymptotic of \eqref{12.212d} has been extended to a class of i.i.d. real random matrices \cite{TV15}. 

The variance $(\sigma_N^{\rm r})^2$ of the distribution of the real eigenvalues can, using (\ref{12.212b}), be expressed in terms of a double summation over gamma functions; however the large $N$ asymptotic form is not easy to then deduce. Later, in Proposition \ref{P2.10c}, an alternative method will be used which shows $(\sigma_N^{\rm r})^2 \sim (2 - \sqrt{2}) E_N^{\rm r}$, and so in particular the variance diverges as $N \to \infty$. This fact, together with a corollary of (\ref{12.212b}) regarding the zeros of $Z_N(\zeta)$, can be used to deduce that upon centring and scaling, the probability distribution for the number of real eigenvalues satisfies a local central limit theorem. 
This is in the spirit of \cite[Prop.~3.2]{BF22a}, although to our knowledge a  local central limit theorem in this context has not appeared previously in the literature. For the corresponding central limit theorem in a more general setting, see \cite{Si17} and \S \ref{S2.6X} below.

\begin{proposition}\label{P2.4b}
We have that $\{ p_{2k,N}^{\rm GinUE} \}$ satisfies the local central limit theorem
\begin{equation}\label{rK1+}
\lim_{N \to \infty} \, \mathop{\sup}\limits_{x \in (-\infty, \infty)}
\Big |  \sigma_{N}^{\rm r}  p_{2k,N} |_{2k= [\sigma_{N}^{\rm r}  x +  E_N^{\rm r}]} - {1 \over \sqrt{2 \pi}} e^{- x^2/2} \Big | = 0.
\end{equation}
\end{proposition}

\begin{proof}
It is established in \cite{KPTTZ15} that the matrix formed by the ratio of gamma functions in (\ref{12.212b}) is positive definite, and hence the zeros of $Z_N(\zeta)$ are all real. Moreover, they are negative real since $Z_N(\zeta)$ is the generating function for the probabilities
$\{ p_{2k,N}^{\rm GinUE} \}$. As remarked above, we know too that for $N \to \infty$ the variance of this probability distribution diverges. Combining these two facts gives, upon appealing to \cite[Th.~2]{Be73}, the stated result.
\end{proof}

\begin{remark}
It has been commented in \S \ref{S2.1a} that the global density of the eigenvalues obeys the circular law in the limit $N \to \infty$. However, the fact that the expected value of real eigenvalues is proportional to $\sqrt{N}$, whereas the corresponding distribution function takes on non-zero values for all (even) values of the number up to $N$ creates, for finite $N$, a so-called Saturn effect whereby the support of the real eigenvalues visibly overshoots the circular law; see \cite[Fig.~1]{BB20}.
\end{remark}

\subsection{Skew orthogonal polynomials}\label{S2.4x}
In (\ref{12.Ca}) set $\gamma_{j,k} := \alpha_{j,k} |_{u=1} + \beta_{j,k} |_{v=1}$. We see that associated with $\gamma_{j,k}$ is the skew inner product
\begin{multline}\label{sO}
\langle f, g \rangle_{s,O}^{\rm g}:=
 \int_{-\infty}^\infty dx \,  
\int_{-\infty}^\infty dy \,  
(\omega^{\rm g}(x) \omega^{\rm g}(y))^{1/2} f(x) f(y) {\rm sgn} \, (y - x),  \\
+ 
 2 i \int_{{\mathbb C}_+} d^2z \,  \omega^{\rm g}(z) \Big ( f(z) g(\bar{z}) - g(z) f(\bar{z}) \Big ).
\end{multline}
Here the subscripts on the inner product indicate that it is (s)kew symmetric and associated with the Gin(O)E.
A Pfaffian is well defined for anti-symmetric matrices of even size only. The analogue of a diagonal form in this setting is the direct sum of $N/2$ two-by-two anti-symmetric matrices.
We would like to choose the polynomials $\{p_{l-1}(x) \}$ in the definition of $[\gamma_{j,k}]$ so that it takes on such a direct sum form. Equivalently we seek a monic polynomial basis that skew-diagonalises the skew inner product (\ref{sO}).
This requires that
\begin{equation}\label{12.212e}
\gamma_{2j,2k}= \gamma_{2j-1,2k-1} = 0, \qquad
\gamma_{2j-1,2k}= - \gamma_{2k,2j-1} = r_{j-1} \delta_{j,k},
\end{equation}
where $\{r_{j-1}\}$ (the skew norms) are the nonzero elements in the block diagonal form $$[\gamma_{j,k}] = \oplus_{j=1}^{N/2} {\small \begin{bmatrix} 0 & r_{j-1} \\
-r_{j-1} & 0 \end{bmatrix}}.$$ If these relations are satisfied, the polynomials are said to be skew-orthogonal.

One observes that (\ref{12.212e}) does not uniquely determine the polynomials. Thus the odd polynomials $p_{2n+1}(z)$ can be replaced by $p_{2n+1}(z) + c q_{2n}(z)$ for any constant $c$
(see e.g.~\cite[\S 6.1.1]{Fo10}).
On the other hand, the existence of $\{q_{m}\}_{m \geq 0}$ follows from a Gram-Schmidt skew-orthogonalisation procedure \cite[Th.~2.4]{AEP22}.
Nevertheless, since this procedure requires to evaluate certain Pfaffians, it is not very useful for the actual computation of the skew orthogonal polynomials.

Crucial for determining the skew orthogonal polynomials relating to the skew inner product (\ref{sO}), and various generalisations associated with ensembles related to the GinUE to be discussed below, are their realisations as $2n \times 2n$ GinOE matrix averages \cite[Eqns.~(4.6)--(4.7)]{AKP10}
\begin{equation}\label{12.212f}
p_{2n}(z) = \langle \det (z \mathbb I_{2n} - G) \rangle, \quad p_{2n+1}(z) = z p_{2n}(z) + \langle \det (z \mathbb I_{2n} - G) {\rm Tr} \, G \rangle.
\end{equation}
The matrix averages in (\ref{12.212f}) are simple to evaluate \cite{FI16}.

\begin{proposition}\label{P2.4}
Let the random matrix $G=[g_{jk}]$ such that the average of distinct pairs is the same as the product of the averages of the  individual entries.
Suppose that the distribution of each  $g_{jk}$ is the same as that for $-g_{jk}$. Then the formulas of (\ref{12.212f}) simplify,
\begin{equation}\label{12.212g}
p_{2n}(z) = z^{2n}, \qquad
p_{2n+1}(z)=z^{2n+1}-\langle {\rm Tr} \, G^2 \rangle z^{2n-1}.
\end{equation}
Also, in the case of GinOE,
$\langle {\rm Tr} \, G^2 \rangle = 2n$ which allows (\ref{12.212g}) to be written 
\begin{equation}\label{12.212g+}
p_{2n}^{\rm g}(z) = z^{2n}, \qquad p_{2n+1}^{\rm g}(z) = - e^{z^2/2}{d \over d z} e^{-z^2/2} p_{2n}(z).
\end{equation}
For the normalisation we have
\begin{equation}\label{12.212h}
r_{n-1}^{\rm g} = 2 \sqrt{2 \pi} \Gamma(2n-1).
\end{equation}
\end{proposition}

\begin{proof}
The expansion of $\det (z \mathbb I_{2n} - G)$ gives a term $z^{2n}$ plus terms involving lower order powers of $z$, the coefficients of which are linear with respect to any single matrix element. Averaging over the matrix element using the assumed invariance of the distribution by negation must give zero. In relation to $\langle \det (z \mathbb I_{2n} - G) {\rm Tr} \, G \rangle$, the assumed invariance of the distribution by negation implies that a nonzero value will result only for terms in the expansion of $\det (z \mathbb I_{2n} - G)$ which contain single powers of the matrix elements, these terms  being in total $z^{2n-1}
{\rm Tr} \, G$. The value of $\langle {\rm Tr} \, G^2 \rangle$ for $G$ a member of $2n \times 2n$ GinOE is immediate from the fact that the elements all have unit variance. For the normalisation $r_n$, from the block diagonal form we have ${\rm Pf} \, [\gamma_{j,k} ] = \prod_{j=0}^{N/2-1} r_j$. Substituting in (\ref{12.Ca}), and noting that as a result of its interpretation as a sum over probabilities we have $Z_N[1,1]=1$, allows (\ref{12.212h}) to be verified.
\end{proof}

\begin{remark} $ $ \\
1.~The skew inner product
$$
\alpha_{j,k}  := 
\int_{-\infty}^\infty dx \, 
\int_{-\infty}^\infty dy \,  \,
(\omega^{\rm g}(x) \omega^{\rm g}(y))^{1/2} p_{j-1}(x) p_{k-1}(y) {\rm sgn} \, (y - x), 
$$
is well known in the theory of the Gaussian orthogonal ensemble (GOE). The corresponding skew orthogonal polynomials are then given in terms of Hermite polynomials (see e.g.~\cite{AFNM00})
$$
p_{2j}^{\rm GOE}(x) = 2^{-2j} H_{2j}(x), \qquad 
p_{2j+1}^{\rm GOE}(x) = - e^{x^2/2}{d \over dx}\Big ( e^{-x^2/2}p_{2j}^{\rm GOE}(x)
\Big );
$$
cf.~(\ref{12.212g}) and (\ref{12.212g+}). \\
2.~Let $\phi(z) = | {1 \over 2} (z + \sqrt{z^2 - 4}) |^{-2s}$, $s > N$ and define the skew inner product $\langle f, g \rangle_{s,O}^{ \phi}$ as in (\ref{sO}) but with $\omega^{\rm g}(z)$ replaced by $\phi(z)$ throughout.  This arose in a study of the so-called Mahler measure of random polynomials \cite{SY19}, and the corresponding skew orthogonal polynomials were required. Without a random matrix underpinning, there is no meaning to (\ref{12.212f}). Nonetheless, the skew orthogonal polynomials have been explicitly determined in terms of a single Chebyshev polynomial (for $p_{2n}(z)$) and a sum of two Chebyshev polynomials (for $p_{2n+1}(z)$).
\end{remark}

\subsection{Correlation functions}\label{S2.5c}
From a statistical mechanics viewpoint, the eigenvalues of GinOE matrices form a two-component system consisting of the real eigenvalues, and of the complex eigenvalues. Here we will show how the real-real and the complex-complex correlation functions can be made explicit, and specify the corresponding Pfaffian point processes.

For this purpose use will be made of the definition (\ref{GP}) of the generalised partition function $Z_N[u,v]$, containing arbitrary functions $u(x)$ and $v(x,y)$. Generally, for integrations involving arbitrary functions, the remaining factors of the integrand can be extracted by functional differentiation,
\begin{equation}\label{GP1}
{\delta \over \delta a(x)}
\int_{-\infty}^\infty a(y) f(y) \, dy = f(x).
\end{equation}
All correlations can be obtained from $Z_N[u,v]$ using the operation of functional differentiation. As an explicit example, for the $m$-point correlation function of the real eigenvalues, $\rho_{(m),N}^{\rm r}$ say, we have
\begin{equation}\label{GP2}
\rho_{(m),N}^{\rm r}(x_1,\dots,x_m) =
{\delta^m \over \delta u(x_1) \cdots \delta u(x_m)} Z_N[u,v]
\bigg |_{u=v=1}.
\end{equation}
This can be checked from (\ref{GP1}) and the definition of $\rho_{(m),N}^{\rm r}$ as the sum over (\ref{3.1}) for $k=m,\dots,N$, with $\lambda_l = x_l$ ($l=1,\dots,m)$, each term weighted by the combinatorial factor $k!/(k-m)!$,
and the variables $\{\lambda_l \}_{l=m+1,\dots,k}$ each integrated over $\mathbb R$. Starting with (\ref{12.Ca}), and choosing the polynomials therein to have the skew orthogonality property (\ref{12.212e}), a Pfaffian formula for $\rho_{(m),N}^{\rm r}$ can be deduced \cite{FN07,BS09}.

\begin{proposition}\label{P2.5}
Let $\{p_j(x)\}$ be the skew orthogonal polynomials for GinOE as specified in (\ref{12.212g+}), and for $x$ real set $\omega(x) = \omega^{\rm g}(x)$ as implied by (\ref{1.1ga}) (specifically then $\omega(x) = e^{-x^2}$). Use these polynomials and this weight to define
$\Phi_k(x) = \int_{-\infty}^\infty {\rm sgn}(x-y) p_k(y) (\omega(y))^{1/2} \, dy$, which in turn is used to define
\begin{align} \label{12.Sr}
S^{\rm r}_N(x,y)  &= \sum_{k=0}^{N/2 - 1} {(\omega(y))^{1/2} \over r_k} \Big (
\Phi_{2k}(x) p_{2k+1}(y) - \Phi_{2k+1}(x) p_{2k}(y) \Big ), \nonumber 
\\
 D^{\rm r}_N(x,y) &= {1 \over 2} {\partial \over \partial x} S^{\rm r}(x,y), \qquad
\tilde{I}^{\rm r}_N(x, y) =  {\rm sgn}(y-x) - 2 \int_x^y S^{\rm r}(x,z) \, dz.
\end{align}
(An equivalent form of $\tilde{I}^{\rm r}(x, y)$ is to replace the integral therein by $ \int_{-\infty}^\infty {\rm sgn} (y-z) S^{\rm r}(x,z) \, dy.$)
We have
\begin{equation}\label{2.19}
 \rho_{(m),N}^{\rm r}(x_1,\dots,x_m) =
 {\rm Pf} \, [\mathcal K_N^{\rm r}(x_j,x_k)]_{j,k=1,\dots,m}, \quad
\mathcal K_N^{\rm r}(x,y):= 
 \begin{bmatrix}
 D^{\rm r}_N(x,y) & S^{\rm r}_N(x,y) \\
 - S^{\rm r}_N(y,x) & \tilde{I}^{\rm r}_N(x, y) \end{bmatrix}.
\end{equation}
\end{proposition}

Before presenting the proof, it is instructive to outline a derivation of the determinantal expression for the $m$-point eigenvalue correlation function of GinUE \cite[Eq.~(2.9) and Prop.~2.2]{BF22a}, the main steps of which can be generalised to GinOE. These steps are

\begin{itemize}
    \item[(i)] Show that for the monic polynomials $p_j(z) = z^j$,
    $$
 \Big \langle \prod_{l=1}^N u(z_l) \Big \rangle_{\rm GinUE} =
 \det \Big [ \delta_{j,k} + {1 \over \langle p_{j-1}, p_{j-1} \rangle} \int_{\mathbb C^2} e^{-|z|^2}(u(z)-1) z^{j-1} \bar{z}^{k-1} \, d^2 z \Big ]_{j,k=1}^N,
    $$
    where $\langle p_{j-1}, p_{k-1} \rangle :=
    \int_{\mathbb C^2} e^{-|z|^2} z^{j-1} \bar{z}^{k-1} \, d^2 z$. 
    \item[(ii)] With $\mathbb I$ denoting the identity operator, use the operator theoretic identity 
    \begin{equation}\label{6.79a}
 \det (\mathbb I + A B) =
 \det (\mathbb I + B A),
 \end{equation}
valid whenever the determinant is well defined (see \cite{De78}) to rewrite the formula in (i) as
$$
\Big \langle \prod_{l=1}^N u(z_l) \Big \rangle_{\rm GinUE} = \det (\mathbb I + K_u),
$$
where $K_u$ is the integral operator on $\mathbb C$ with kernel $(u(z_2)-1) K_N(z_1,z_2)$, with $K_N(z_1,z_2)$ given by \cite[Eq.~(2.9)]{BF22a}.
\item[(iii)] Use functional differentiation to now extract the $m$-point eigenvalue correlation function from the Fredholm expansion \cite{WW65}
$$
\det (\mathbb I + K_u) = 1 +
\sum_{k=1}^N {1 \over k!}
\int_{\mathbb C} d z_1 \, u(z_1)\cdots \int_{\mathbb C} d z_N \, u(z_N)\det [ K_N(z_j,z_l)]_{j,l=1,\dots,k}.
$$
\end{itemize}

\noindent
{\it Proof of Proposition \ref{P2.5}.}
In relation to steps (i) and (ii),
 we follow \cite[Proof of Prop.~6.3.6]{Fo10}. Starting with (\ref{12.Ca}), we set $v=1$, $u = 1 + \hat{u}$ and choose $\{p_{l-1}(x)\}$ to have the skew orthogonality property (\ref{12.212e}). Also, we introduce the notations $\psi_j(x) = e^{-x^2/2} p_{j-1}(x)$ and $\varepsilon[f](x)=
 \int_{-\infty}^\infty {\rm sgn}(x-y) f(y) \, dy$. Then, with $\gamma_{j,k} := \alpha_{j,k} |_{u=1} + \beta_{j,k} |_{v=1}$, we have
 $$
 \alpha_{jk} + \beta_{jk} = \gamma_{jk} - \int_{-\infty}^\infty \Big ( \hat{u}(x) \psi_j(x) \varepsilon[\psi_k](x) -
 \hat{u}(x) \psi_k(x) \varepsilon[\psi_j](x)- \hat{u}(x) \psi_k(x) \varepsilon[\hat{u}\psi_j](x) \Big ) \, dx.
 $$
 Next we introduce the further notation $G_{2j-1}(x) = \psi_{2j}(x)$, $G_{2j}(x) = -\psi_{2j-1}(x)$, multiply the even rows by $-1$ and use the facts that $[\gamma_{jk}]$ has the block diagonal form as specified below (\ref{12.212e}) and that the square of the Pfaffian is the corresponding determinant,  to deduce
 \begin{multline}\label{6.79}
  (Z_N[1+\hat{u},1])^2  = \prod_{j=1}^N r_{j-1}^2 \,
  \det \bigg [ \delta_{j,k} + {1 \over r_{[(j-1)/2]}} \\
  \times \int_{-\infty}^\infty
  \Big ( \hat{u}(x) G_j(x) \varepsilon[\psi_k](x) -
 \hat{u}(x) \psi_k(x) \varepsilon[G_j](x)- \hat{u}(x) \psi_k(x) \varepsilon[\hat{u}G_j](x) \Big ) \, dx \bigg ].
 \end{multline}
 This accomplishes step (i).
 
For step (ii), the key observation, due to \cite{TW98}, is that the determinant in (\ref{6.79}) can be written as $\det(\mathbb I_N + A B)$ for $A$ and $B$ appropriate matrix operators. Specifically, $A$ is the $N \times 2$ matrix valued integral operator, with row $j$ of the kernel corresponding to the pair
 $$
 {1 \over r_{[(j-1)/2]}} \Big (-\hat{u}(y)\varepsilon[G_j](y)-\hat{u}(y)  \varepsilon[\hat{u}G_j](y),
 \hat{u}(y) G_j(y) \Big ).
 $$
 The operator $B$ multiplies by the $2 \times N$ matrix with first row $\psi_k(y)$, and second row 
 $\varepsilon[\psi_k](y)$.

 We are now ready to apply
 (\ref{6.79a}).
 With $A$ and $B$ as above we see that the operator $\mathbb I + B A$ is the $2 \times 2$ matrix integral operator
 \begin{eqnarray}\label{7.mi}
\lefteqn{
\left [ \begin{array}{rr}
1 - \sum_{j=1}^{N} \Big ( \tilde{\psi}_j \otimes f \varepsilon G_j
+ \tilde{\psi}_j \otimes f \varepsilon (f G_j) \Big ) &
\sum_{j=1}^{N} \tilde{\psi}_j \otimes f G_j \\
- \sum_{j=1}^{N} \Big ( \varepsilon \tilde{\psi}_j \otimes  f \varepsilon G_j +
 \varepsilon \tilde{\psi}_j \otimes  f \varepsilon (f G_j) &
1 + \sum_{j=1}^{N}  \varepsilon \tilde{\psi}_j \otimes  f G_j
\end{array} \right ]} \nonumber \\
&&\quad
=
\left [ \begin{array}{rr}
1 - \sum_{j=1}^{N}  \tilde{\psi}_j \otimes f \varepsilon G_j
&
\sum_{j=1}^{N}\tilde{ \psi}_j \otimes f G_j \\
- \sum_{j=1}^{N}  \varepsilon \tilde{\psi}_j \otimes  f \varepsilon G_j 
- \varepsilon  f  &
1 + \sum_{j=1}^{N}  \varepsilon \tilde{\psi}_j \otimes  f G_j
\end{array} \right ]
\left [ \begin{array}{cc}
1 & 0 \\ \varepsilon f & 1 \end{array} \right ].
\end{eqnarray}
Here $\tilde{\psi}_j := \psi_j/ r_{[(j-1)/2]}$ and
the notation $a \otimes b$ denotes the integral operator with kernel
 $a(x) b(y)$. The determinant of the second matrix in the second expression is equal to $1$, and so it can be replaced by the elementary anti-symmetric matrix obtained by setting $z=0$, $w=1$ in (\ref{1.1}), without changing the value. Doing this we can identify the kernel of the matrix integral operator as being given by an anti-symmetric matrix so the square root of the determinant is a Pfaffian. Hence
\begin{equation}\label{2.27m}
Z_N[1+\hat{u},1] = {\rm Pf} \Bigg (
\begin{bmatrix} 0 & \mathbb I \\ -\mathbb I & 0 \end{bmatrix} +
\begin{bmatrix} \mathbf D^{\rm r}(x,\cdot)[\hat{u}] & \mathbf S^{\rm r}(x,\cdot)[\hat{u}] \\
- \mathbf S^{\rm r}(\cdot,x)[\hat{u}] &
\tilde{\mathbf I}^{\rm r} (x,\cdot)[\hat{u}]
\end{bmatrix} \Bigg ),
\end{equation}
where here the Pfaffian is defined in terms of the product of the eigenvalues.
In the second matrix the operators act by integration over the non-specified variable, weighted by $\hat{u}$, and otherwise have kernels given by the corresponding entries in (\ref{2.19}).

We are now up to stage (iii) of the outlined strategy.
Analogous to the theory of Fredholm determinants, this quantity --- a Fredholm Pfaffian --- can be expanded in a series of Pfaffians of scalar matrices to read \cite{Ra00}
\begin{equation}\label{2.23}
Z_N[1+\hat{u},1] = 1 +
\sum_{m=1}^N {1 \over m!}
\int_{-\infty}^\infty dx_1 \, \hat{u}(x_1) \cdots \int_{-\infty}^\infty dx_m\,\hat{u}(x_m)
{\rm Pf} \, [\mathcal K_N^{\rm r}(x_j,x_k)]_{j,k=1,\dots,m}.
    \end{equation}
    Here the $2 \times 2$ anti-symmetric matrix kernel $\mathcal K_N^{\rm r}$ is determined by the kernel of the second matrix integral operator in (\ref{2.27m}).
We note that the formula (\ref{GP2}) remains valid with each $u(x)$ in the functional derivative replaced by $\hat{u}$, with the latter now set equal to zero after the derivatives have been computed. Applying this modified formula to (\ref{2.23}), we read off the sought Pfaffian formula (\ref{2.19}) for the correlations. \hfill $\square$

\medskip
With the polynomials $\{p_{j-1}(x)\}$ given according to (\ref{12.212g+}), the matrix elements in (\ref{12.Sr}) can be made explicit, with knowledge of $S^{\rm r}(x,y)$ sufficing. 

\begin{proposition}\label{P2.9}
We have
\begin{align}\label{13.222}
S^{\rm r}_N(x,y) & =  
  {e^{-(x^2+y^2)/2} \over \sqrt{2 \pi} }
\sum_{k=0}^{N-2} {(xy)^k \over k!}  +
{ e^{- y^2/2} \over 2 \sqrt{2 \pi}  }
{ y^{N-1} \Phi_{N-2}(x) \over (N-2)!} \nonumber\\
& = {e^{-(x^2+y^2)/2} \over \sqrt{2 \pi} } \bigg (
e^{xy} {\Gamma(N-1;xy) \over \Gamma(N-1) } +
{1 \over 2} (\sqrt{2} y)^{N-1} e^{x^2/2} {\rm sgn}(x) {\gamma(N/2 - 1/2;x^2/2) \over \Gamma(N-1) } 
 \bigg ).
\end{align}
\end{proposition}

\begin{proof}
The key to obtaining the first equality is to first note from  the derivative form of $p_{2n+1}(z)$ given by
(\ref{12.212g+}) that
$$
\Phi_{2k+1}(x) = - 2 e^{-x^2/2}x^{2k},
$$
which implies the summation identity
\begin{equation}\label{Ae1}
    -\sum_{k=0}^{N/2-1} {1 \over r_k} \Phi_{2k+1}(x) p_{2k}(y)
= {1 \over \sqrt{2 \pi}} e^{-x^2/2}\sum_{k=0}^{N/2-1} {(xy)^{2k} \over (2k)!}.
\end{equation}
Also needed is the further derivative formula
\begin{equation}\label{Ae2}
{2(k+1) \over r_{k+1}} p_{2k+2}(x) - {1 \over r_k}p_{2k}(x) = -{1 \over 2 \sqrt{2 \pi}} {1 \over (2k+1)!}
e^{x^2/2}{d \over dx} e^{-x^2/2} x^{2k+1},
\end{equation}
which implies the further summation identity
\begin{equation}\label{Ae3}
\sum_{k=0}^{N/2-1} {1 \over r_k} \Phi_{2k}(x) p_{2k+1}(y)
= {1 \over \sqrt{2 \pi}} 
e^{-x^2/2}\sum_{k=0}^{N/2-2} {(xy)^{2k+1} \over (2k+1)!}+{1 \over r_{N/2-1}} \Phi_{N-2}(x) y^{N-1}.
\end{equation}
Adding (\ref{Ae1}) and (\ref{Ae3}) establishes the first equality.

In relation to the second equality, by identifying terms with the first equality, the task becomes to show
$$
\Phi_{N-2}(x) = 2^{(N-1)/2}
{\rm sgn}(x) \gamma(N/2 - 1/2;x^2/2).
$$
Since from the definition in Proposition \ref{P2.5},
$$
\Phi_{N-2}(x) = 2 \int_0^x e^{-y^2/2} y^{N-2} \, dy,
$$
this is readily verified.
\end{proof}

Setting $x=y=\sqrt{N} \tilde{x}$ in (\ref{13.222}) corresponds to a global scaling of the density of real eigenvalues. One can verify that the large $N$ limit is determined entirely by the first term, with the simple result
\begin{equation}\label{11.xs}
\lim_{N \to \infty} \rho_{(1),N}^{\rm r, b}(\sqrt{N} \tilde{x}) = {1 \over \sqrt{2 \pi}} \chi_{|\tilde{x}| < 1}.
\end{equation}

Knowledge of (\ref{11.xs}) indicates that local limit theorems associated with 
(\ref{13.222}) will depend on the choice of origin. In keeping with this we find that 
with the latter chosen on the real axis at $\nu \sqrt{N}$ with $|\nu| < 1$ the final term does not contribute, while the ratio of the incomplete gamma function in the first term equals unity, giving for the bulk limit
\begin{equation}\label{11.SXY+}
S^{\rm r, b}_\infty (x,y):=\lim_{N \to \infty} S^{\rm r}_N(\nu \sqrt{N}+x,\sqrt{N}+y)  = {1 \over \sqrt{2 \pi} }
e^{ - (x - y)^2/2}.
\end{equation}
Setting $x=y$  in this gives for bulk density the value
$ \rho_{(1),\infty}^{\rm r, b}(x)= 1 / \sqrt{2 \pi} $, as already anticipated below (\ref{12.212d}). The result (\ref{11.SXY+}) gives for the bulk limiting kernel in
(\ref{2.19}) the functional form
\begin{equation}\label{11.SXYp}
K_\infty^{\rm r, b}(x,y) =
\begin{bmatrix}{1 \over 2\sqrt{2\pi}} (y-x) e^{-(x-y)^2/2} &
{1 \over \sqrt{2\pi}}  e^{-(x-y)^2/2} \\
- {1 \over  \sqrt{2\pi}}  e^{-(x-y)^2/2} & {\rm sgn}(x-y)
{\rm erfc}(|x-y|/\sqrt{2})
\end{bmatrix}.
\end{equation}

The result (\ref{11.SXY+}) breaks down for $\nu = \pm 1$, as $\pm \sqrt{N}$ are the (leading order) boundaries of the support of the real eigenvalues. Then the
uniform asymptotic expansion 
\cite{Tr50}
 \begin{equation}\label{Tr}
 {\gamma(M-j+1;M) \over \Gamma(M-j+1)} \mathop{\sim}\limits_{M \to \infty} {1 \over 2} \Big ( 1 + {\rm erf} \Big ( {j \over \sqrt{2M}} \Big ) \Big ),
  \end{equation} 
  gives for the edge scaling limit 
  \begin{equation}\label{11.SXY}
\lim_{N \to \infty} S^{\rm r}_N(\sqrt{N} + x,\sqrt{N} + y) = {1 \over \sqrt{2 \pi} }
\bigg ( {1 \over 2} e^{- (x - y)^2/2} \Big ( 1- {\rm erf} {x + y \over \sqrt{2}} \Big ) +
{e^{ - y^2} \over 2 \sqrt{2} } (1 + {\rm erf} \, x )) \bigg ).
\end{equation}
In particular, setting $x=y$ gives for the edge scaled density
 \begin{equation}\label{11.SXYa}
 \rho_{(1),\infty}^{\rm r, e}(x) =
 {1 \over 2 \sqrt{2 \pi}} \bigg (
 1 - {\rm erf} (\sqrt{2} x) +
 {e^{ - x^2} \over  \sqrt{2} } (1 + {\rm erf} \, x )) \bigg ).
 \end{equation}
For $x,y \to - \infty$
(\ref{11.SXY}) reclaims the bulk limiting form
 (\ref{11.SXY+}).

For finite $N$, the density of real eigenvalues $\rho_{(1),N}^{\rm r}(x) = S^{\rm r}_N(x,x)$ was first calculated in \cite{EKS94}. The starting point was the formula
\begin{equation}\label{eks1}
\rho_{(1),N}^{\rm r}(x) = {e^{-x^2/2}\over 2^{N/2} \Gamma(N/2)} \Big \langle 
| \det (G - x \mathbb I_{N-1}) |
\Big \rangle,
\end{equation}
where the average is over $(N-1) \times (N-1)$ GinOE matrices. This can be deduced from (\ref{3.1}), although in \cite{EKS94} a method called eigenvalue deflation was used, whereby a Householder transformation is applied to reduce all entries in the final column, except the last entry, to zero. The formula $E_N^{\rm r} = \int_{-\infty}^\infty \rho_{(1),N}^{\rm r}(x) \, dx$ can then be used to deduce (\ref{12.212c}).

\begin{remark}\label{R2.10} $ $ \\
1.~Up to proportionality the term involving the summation in the first equality of (\ref{13.222}) is recognised as the correlation kernel $K_N(w,z)$ for GinUE \cite[Eq.~(2.10)]{BF22a}, with $w,z$ real and $N \mapsto N-1$. Using this notation, the results of Proposition \ref{P2.4} for the skew orthogonal polynomials substituted into the definition of $S_N^{\rm r}$ in (\ref{12.Sr}) show that it is furthermore true that \cite[\S 4.2]{FI16}
\begin{equation}\label{2.37a}
S_N^{\rm r, g}(x,y) = {1 \over 2\sqrt{2 \pi}}
\int_{-\infty}^\infty (y - s) {\rm sgn}(x-s) K_{N-1}(y,s) \, ds.
\end{equation}
\noindent
2.~The entry $D_N^{\rm r}(x,y)$ of $K_N^{\rm r}$ in (\ref{2.19})
can be multiplied by any scalar $c \ne 0$, provided that the entry $\tilde{I}_N^{\rm r}(x,y)$ is multiplied by the scalar $1/c$ without changing the value of $\rho_{(m),N}^{\rm r}$. This follows from the fact for $B$ a $2m \times 2m$ symmetric matrix, and $A$ a $2m \times 2m$ anti-symmetric matrix, ${\rm Pf} (BAB^T) = \det(B) {\rm Pf}(A)$ (for a derivation see e.g.~\cite[Exercises 6.1 q.1]{Fo10}). \\
3.~The $(k_1,k_2)$-point correlation function for $k_1$ real eigenvalues $x_1,\dots,x_{k_1}$ and $k_2$ complex eigenvalues $z_1,\dots,z_{k_2}$ has the Pfaffian form \cite{FN07,SW08,BS09}
$$
\rho_{(k_1,k_2),N}(x_1,\dots,x_{k_1};z_1,\dots,z_{k_2}) = {\rm Pf} \,
\begin{bmatrix}
[\mathcal K_N^{\rm r}(x_j,x_l)]_{j,l=1}^{k_1} &
[\mathcal K_N^{\rm r,c}(x_j,z_l)]_{j=1,\dots,k_1 \atop l=1,\dots,k_2} \\
\Big ( - [\mathcal K_N^{\rm r,c}(x_j,z_l)]_{j=1,\dots,k_1 \atop l=1,\dots,k_2} \Big )^T &
[\mathcal K_N^{\rm c}(z_j,z_l)]_{j,l=1}^{k_2}]
\end{bmatrix}
$$
generalising (\ref{2.19}), where each of $\mathcal K_N^{\rm r},\mathcal K_N^{\rm r,c},\mathcal K_N^{\rm c} $ is a $2 \times 2$ correlation kernel. We may use the notation $\mathcal K_N^{\rm c,r}$ as an abbreviation for the bottom left block. The explicit form of $\mathcal K_N^{\rm r}$ is given by Proposition \ref{P2.5}, and its bulk large $N$ limit is given by (\ref{11.SXYp}). The bulk large $N$ limits of the other kernels are \cite[Cor.~9]{BS09}
$$
\mathcal K_\infty^{\rm r,c}(x,w) = {1 \over \sqrt{2 \pi}} ({\rm erfc}(\sqrt{2} v) )^{1/2}
\begin{bmatrix} (w-x) e^{-(x-w)^2/2} & i (\bar{w} - x) e^{-(x-\bar{w})^2/2} \\
- e^{-(x-w)^2/2} &
-i e^{-(x-\bar{w})^2/2} \end{bmatrix},
$$
$$
\mathcal K_\infty^{\rm c}(w,z) = {1 \over \sqrt{2 \pi}} ({\rm erfc}(\sqrt{2} v) {\rm erfc}(\sqrt{2} y) )^{1/2}
\begin{bmatrix} (z-w) e^{-(w-z)^2/2} & i (\bar{z} - w) e^{-(w-\bar{z})^2/2} \\
i ({z} - \bar{w}) e^{-(\bar{w}-{z})^2/2} &
-(\bar{z} - \bar{w} )e^{-(\bar{w}-\bar{z})^2/2} \end{bmatrix},
$$
where $w=u+iv$, $z=x+iy$ with $v,y > 0$.
In particular, we read off from the latter of these the formula for the bulk density of complex eigenvalues
\begin{equation}\label{cE}
\rho_{(1), \infty}^{\rm c}(z) =
\sqrt{2 \over \pi} \, {\rm erfc}(\sqrt{2}y) y e^{2 y ^2},
\end{equation}
which tends to the constant $1/\pi$ as $y \to \infty$, as required by the circular law. \\
4.~For finite $N$ the complex-complex correlation function is
formally the same as (\ref{2.19}), but now with correlation kernel $\mathcal K_N^{\rm c}$ and its elements now $D_N^{\rm c}, I_N^{\rm c}, S_N^{\rm c}$ in the same positions as for $\mathcal K_N^{\rm r}$. For the definition of these matrix elements, require that $\{ p_j(x) \}$ be the skew orthogonal polynomials of Proposition \ref{P2.4}, and then set $q_j(z) = (\omega^{\rm g}(z))^{1/2} p_j(z)$, with $z=x+iy$. Use these to define
\begin{equation}\label{cE1}
S_N^{\rm c}(w,z) =2i
\sum_{j=0}^{N/2-1} {1 \over r_j}
\Big ( q_{2j}(w) q_{2j+1}(\bar{z}) - q_{2j+1}(w) q_{2j}(\bar{z}) \Big ).
\end{equation}
We then have \cite{FN07},\cite{BS09}, \cite[\S 4.5]{Ma11}
\begin{equation}\label{cE2}
\mathcal K_N^{\rm c}(w,z) =
\begin{bmatrix} - i S_N^{\rm c}(w,\bar{z}) & S_N^{\rm c}(w,z) \\
- S_N^{\rm c}(z,w) & i  S_N^{\rm c}(\bar{w},z) \end{bmatrix}.
\end{equation}
From the results of Proposition \ref{P2.4} substituted in (\ref{cE1}) 
shows
\begin{equation}\label{cE2a}
S_N^{\rm c}(w,z) = {i(\bar{z}-w
) \over \sqrt{2 \pi}}
\Big ( {\rm erfc}(\sqrt{2}u)  e^{2u^2}  {\rm erfc}(\sqrt{2}y)  e^{2y^2}  \Big )^{1/2}
K_{N-1}(w,z),
\end{equation}
where $K_{N-1}$ corresponds to the correlation kernel for GinUE 
\cite[Eq.~(2.10)]{BF22a} with $N \mapsto N - 1$ and $y={\rm Im}\, z$, $u={\rm Im}\, w$.
In particular, this implies
\begin{equation}\label{cE3}
\rho_{(1), N}^{\rm c}(z) =
S_N^{\rm c}(z,z)=
\sqrt{2 \over \pi} \, y \, {\rm erfc}(\sqrt{2}y)  e^{2 y ^2} {\Gamma(N-1;x^2+y^2) \over \Gamma(N-1)},
\end{equation}
which was first derived in \cite{Ed97} by (effectively) obtaining a complex analogue of (\ref{eks1}),
\begin{equation}\label{eks1+}
\rho_{(1),N}^{\rm c}(z) = 
2 y \omega^{\rm g}(z) {C_{N-1}^{\rm g} \over C_{N}^{\rm g}}
 \Big \langle 
| \det (G - z \mathbb I_{N-1}) |^2
\Big \rangle,
\end{equation}
with the average over $(N-1) \times (N-1)$ GinUE matrices. We remark that this in turn has the generalisation  \cite{SW08}
\begin{equation}\label{eks2+}
S_N^{\rm c}(w,z)=i(\bar{z}-w) \Big ( \omega^{\rm g}(w) \omega^{\rm g}(z) \Big )^{1/2}
 {C_{N-1}^{\rm g} \over C_{N}^{\rm g}} \Big \langle \det \Big ( (G - w \mathbb I_{N-1})
  (G - \bar{z} \mathbb I_{N-1}) \Big )
\Big \rangle.
\end{equation}
The global scaling limit of (\ref{cE3}) is readily computed to give \cite[Th.~6.3]{Ed97} 
\begin{equation}\label{cE3+}
\lim_{N \to \infty} 
\rho_{(1), N}^{\rm c}(\sqrt{N}z) =
\begin{cases}{1 \over \pi}, & |z|<1 \\
0, & |z| > 1, 
\end{cases}
\end{equation}
in keeping with the circular law \cite[Eq.~(2.17)]{BF22a}. Moreover, in the neighbourhood of the boundary, but away from the real axis, edge scaling reclaims the edge scaled GinUE result \cite[Eq.~(2.19) with $(x_1,y_1)=(x_2,y_2)$]{BF22a}; see \cite{BS09}. \\
\end{remark}

\subsection{Fluctuation formulas for real eigenvalues}\label{S2.6X}
The variance $(\sigma_N^{\rm r})^2$ of the distribution of the real eigenvalues can be expressed in terms of $\rho_{(2),N}^{\rm r}(x,y)$, and its explicit expression as implied by (\ref{2.19}) used to deduce its large $N$ form \cite{FN07}.

\begin{proposition}\label{P2.10c}
In terms of the leading large $N$ form of $E_N^{\rm r}$,
$E_N^{\rm r} \sim \sqrt{2N/\pi}$ as given by (\ref{12.212d}), we have
\begin{equation}\label{eks2}
(\sigma_N^{\rm r})^2 \sim (2 - \sqrt{2}) E_N^{\rm r}.
\end{equation}
\end{proposition}

\begin{proof}
With $\rho_{(2),\infty}^T(z_1,z_2):=  \rho_{(2),\infty}(z_1,z_2)-
\rho_{(1),\infty}(z_1)\rho_{(1),\infty}(z_2)$, we have 
\begin{equation}\label{eks2.a}
(\sigma_N^{\rm r})^2 =
\int_{-\infty}^\infty dx 
\int_{-\infty}^\infty dy \,
 ( \rho_{(2),N}^{{\rm r}, T}(x,y) + 
\rho_{(1),N}(x) \delta(x-y)
 );
 \end{equation}
 cf.~\cite[first formula of Eq.~(3.2)]{BF22a}. From the qualitative knowledge that $\rho_{(2),N}^{{\rm r}, T}(x,y)$ is to leading order translationally invariant for $x,y \in (-\sqrt{N},\sqrt{N})$ with corresponding density $1/\sqrt{2 \pi}=: \rho^{\rm r}$ it follows
 \begin{equation}\label{eks2.a1}
 (\sigma_N^{\rm r})^2 \sim E_N^{\rm r} \Big ( 1 +
 {1 \over \rho^{\rm r}}
 \int_{-\infty}^\infty \rho_{(2),\infty}^{{\rm r}, T}(0,y) \, dy \Big ).
  \end{equation}
 We read off from (\ref{2.19}) that
 \begin{equation}\label{eks2X}
 \rho_{(2),\infty}^{{\rm r}, T}(x,y) = - S^{\rm r}_\infty(x,y) S^{\rm r}_\infty(y,x) + D^{\rm r}_\infty(x,y) \tilde{I}_\infty(x,y),
 \end{equation}
 where explicit form of the functions herein are the matrix elements of (\ref{11.SXYp}).
 This allows the integral to be computed to give
  (\ref{eks2}).
\end{proof}

The asymptotic formula
(\ref{eks2}) can be generalised to the setting of a scaled linear statistic $F:= \sum_{j=1}^N f(x_j/\sqrt{N})$ \cite{Si17,FS23}.

\begin{proposition}\label{P2.10c1}
Subject only to a mild growth condition on $f$, we have
\begin{equation}\label{eks2.b}
{\rm Var} \, F \sim \Big ({1 \over 2} \int_{-1}^1 f^2(x) \, dx \Big ) (2 - \sqrt{2}) E_N^{\rm r}.
\end{equation}
\end{proposition}

\begin{proof}(Sketch)
The variance is given by the RHS of (\ref{eks2.a}) with factors  $f(x/\sqrt{N}) f(y/\sqrt{N})$ also in the integrand. A simple change of variables, and knowledge that correlations decay as Gaussians outside the leading support, gives in place of (\ref{eks2.a1})
$$
{\rm Var} \, F \sim E_N^{\rm r} \Big ({1 \over 2} \int_{-1}^1 f^2(x) \, dx +
{1 \over 2} \int_{-1}^1 dx \, f(x) 
\int_{-1}^1 dy \, f(y)
\sqrt{2 \pi N} \rho_{(2),\infty}^T(\sqrt{N}x, \sqrt{N}y) \Big).
$$
From the explicit functional form (\ref{eks2X}) we observe that for large $N$
$$
\sqrt{2 \pi N} \rho_{(2),\infty}^T(\sqrt{N}x, \sqrt{N}y) \sim (1 - \sqrt{2})\delta(x-y).
$$
This gives the result for $f$ at least piecewise continuous. As noted in \cite{FS23}, following an idea in \cite{Ko15}, taking advantage of the translation invariance of the limiting two-point truncated correlation function allows the integral instead to be computed in terms of Fourier transforms, and so further enlarging the class of $f$ for its applicability.
\end{proof}

In the case of the counting function for a determinantal point process, we know that a diverging variance is sufficient to conclude that the centred and rescaled linear statistic obeys a central limit theorem. Here we see that for general linear statistic the variance diverges (this is in distinction to the case of smooth statistics for GinUE; recall \cite[\S 3.3]{BF22a}), and furthermore the correlations are no longer determinantal but rather Pfaffian. Nonetheless, it has been shown in \cite{Si17,FS23} that it can be proved that the higher order cumulants tend to zero as $N \to \infty$, implying the sought central limit theorem. The conclusion is then that the statistic $(F-\langle F \rangle )/\sqrt{E_N^{\rm r}}$ tends to a zero mean Gaussian, with variance given by the RHS of (\ref{eks2.b}) with the factor $E_N^{\rm r}$ omitted. 

\subsection{Relationship to diffusion processes and gap probabilities}
A Pfaffian point process with correlation kernel (\ref{11.SXYp}), distinct from the real eigenvalues of GinOE, has been known in the literature (at least implicitly) for a long time \cite{MbA01}. This point process is the annihilation process $A + A \to \varnothing$ on the real line, in which particles are initially non-interacting on the real line with density $\rho_0 = 1/\sqrt{2\pi}$, and evolving according to Brownian motion for all $t > 0$. Colliding particles are each annihilated, and the density is rescaled to have the initial value $\rho_0$. As explained in \cite{Fo15}, assembling results in \cite{MbA01} gives that the correlation functions of this process are, in the limit $t \to \infty$ the Pfaffian point process with the same correlation kernel (\ref{11.SXYp}) as the real eigenvalues of GinOE matrices in the bulk scaling limit. This process was reconsidered in \cite{TZ11}, in which the Pfaffian point process structure, and its correlation kernel, were made explicit.
Our interest here is first to make note of some consequences of this co-incidence in relation to certain gap probabilities in the GinOE, and then to proceed to discuss special properties of other GOE gap probabilities.

For the first such consequence, knowledge about the correlations for another diffusion process --- specifically the coalescence process $A + A \to A$ --- is required \cite{bA98}.
As for the annihilation process the particles are initially non-interacting on the real line with density $\rho_0 = \sqrt{2/\pi}$
(not $\sqrt{1/2\pi}$), and evolve according to Brownian motion for all $t > 0$. Now pairs  of colliding particles merge to a single particle, and the density is rescaled to have the initial value $\rho_0$. Results from \cite{MbA01,TZ11,Fo15} give that that this is a Pfaffian point process with correlation kernel equal to (\ref{11.SXYp}), now multiplied by a factor of $2$.
From the theory of thinning a point process (recall \cite[\S 3.4]{BF22a}) it follows that bulk scaled GinOE real eigenvalues are statistically equivalent to the coalescence process with every particle deleted with probability $1/2$. Now denote by $E^{\rm b, r}({\rm even};J)$ the probability that 
a collection of intervals $J$ contain an even number of particles for bulk scaled GinOE real eigenvalues. And denote by 
$E^{\rm c}(k;J)$ the probability that there are $k$ particles in $J$ for the coalescence process. The relationship between the two processes via deletion with probability half in the latter implies
\cite{MbA01,Fo15}
$$
E^{\rm r, b}({\rm even};J) = E^{\rm c}({\rm 0};J) + {1 \over 2}\sum_{k=1}^\infty E^{\rm c}( k;J) =
{1 \over 2} + {1 \over 2} E^{\rm c}({\rm 0};J).
$$
This leads to an explicit formula because $E^{\rm c}({\rm 0};J)$, for $J$ a set of $m$ disjoint intervals, endpoints $x_1<\cdots <x_{2m}$ is the Pfaffian of the antisymmetric matrix with entries $A_{ij} = {\rm erfc}((x_j-x_i)/\sqrt{2})$ \cite{MbA01}. Thus in the simplest case of a single interval
\begin{equation}
E^{\rm r, b}({\rm even};(0,s)) =
{1 \over 2} \Big (1 + {\rm erfc}
{s \over \sqrt{2}} \Big ).
\end{equation}

The second and final of the consequences as noted in \cite{Fo15} relate to the small and large $s$ expansions of $E^{\rm b, r}(0;(0,s))$. As these had earlier been reported in the literature for the same gap probability in case of the annihilation process \cite{DZ96}, the fact that the statistical state of the latter coincides with that of bulk scaled real GinOE eigenvalues gives the sought results.

\begin{proposition}\label{PC1}
For small $s$ 
\begin{equation}\label{Rr1}
 E^{\rm r, b}(0;(0,s)) = 1-\frac{s}{\sqrt{2 \pi }}+\frac{s^3}{12 \sqrt{2 \pi }} -\frac{s^5}{80 \sqrt{2 \pi }}+\frac{s^6}{720 \pi } + \cdots ,
\end{equation}
while for large $s$ 
\begin{equation}\label{Rr2}
E^{\rm r, b}(0;(0,s)) = e^{- {c} s + C + o(1)},
\end{equation}
where 
$
{c} = {1 \over 2 \sqrt{2 \pi}} \zeta(3/2)\approx 0.5211$, and $C$ is as in (\ref{2.11}).
\end{proposition}

\begin{remark} $ $ \\
1.~We comment now on an application of $E^{\rm r,b}$ to a certain topological transition relating to parametric energy levels associated with particular quantum dots, due to Beenakker et al.~\cite{BEDPSW13}. This appears in a setting of energy levels $\cdots > E_3 > E_2 > E_1 > 0$
and $| E_0 | < E_1$. Thus each $E_j$ is positive while $E_0$ may be positive or negative, but no bigger in absolute value to $E_1$. Moreover the energy levels are all functions of a parameter $\phi$. The question of interest is the statistics of the values of $\phi$ such that $E_0(\phi) = 0$. As a theoretical model of the particular scattering problem (Andreev reflection) giving rise to these energy levels, the solution of this last equation is mapped to the real eigenvalues of the real, but non-symmetric, random matrix
\begin{equation}\label{MZ}
M = (\mathbb I_{2N} - R)(\mathbb I_{2N} + R)^{-1} \mathbb Z_{2N},
\end{equation}
where $R$ is a $2N \times 2N$ Haar distributed real orthogonal matrix with determinant $+1$, and $Z_{2N} = \mathbb I_N \otimes \begin{bmatrix} 0 & 1 \\ -1 & 0
\end{bmatrix}$. Note that $M$ has the skew symmetry $M^T = - Z_{2N} M Z_{2N}$, implying all eigenvalues of $M$ are doubly degenerate. Although no exact solution relating to the ensemble $\{M\}$ is known, numerics indicate that the leading term for the number of real eigenvalues is $\sqrt{2N/\pi}$ as for GinOE, and most significantly for present purposes that after rescaling so that the mean density is $1/\sqrt{2 \pi}$, these real eigenvalue have bulk spacing specified by the GinOE distribution $E^{\rm r,b}$. \\
2.~The ensemble of random matrices (\ref{MZ}) is one example of a member of the same universality class for its real eigenvalues (at least according to numerical evidence) as GinOE. On the other hand, there are known ensembles for which the real eigenvalues appear (from numerical evidence) to have different statistics, for example
$$
\begin{bmatrix} A & B \\
- B^T & C \end{bmatrix},
$$
where $A,B,C$ are square random matrices of the same size with $A,C$ real symmetric from the GOE, and $B$ from GinOE \cite{XKLOS22}. We make mention too of ensembles of non-Hermitian random matrices with complex entries which exhibit real eigenvalues; see \cite[\S 8]{Ma12} and \cite{XKLOS22}.
\end{remark}

Choosing $\hat{u}(x) = - \chi_{x\in(0,s)}$ in the generalised partition function $Z_N[1+\hat{u},1]$ of \S \ref{S2.5c} we see that $E^{\rm r, b}(0;(0,s)) = \lim_{N \to \infty} Z_N[1+\hat{u},1]$ and thus from (\ref{2.27m}) we have a Fredholm Pfaffian formula for this probability. From the standard fact that a Pfaffian is the square root of the corresponding determinant (this can be deduced from the Pfaffian/ determinant identity of Remark \ref{R2.10}.2, the fact that the $2 \times 2$ entries of the operator defining the corresponding Fredholm determinant is a function of differences was used in \cite{Fo15} to give an alternative derivation of $c$. 
Moreover, this calculation was also carried out in the more general setting of the thinned point process, in which each real eigenvalue is deleted with probability $1-\xi$; we know from \cite[\S 3.4]{BF22a} that its effect in the Fredholm Pfaffian formula is simply to multiply the kernel by $\xi$. Denoting the corresponding gap probability by $E^{\rm r, b}(0;(0,s);\xi)$, it was found
\begin{equation}\label{cpx}
E^{\rm r, b}(0;(0,s);\xi)
\mathop{\sim}\limits_{s \to \infty} e^{-c(\xi)s + O(1)}, \quad c(\xi) = - {1 \over 4 \pi}
\int_{-\infty}^\infty \log\Big (1 - (2\xi - \xi^2) e^{-u^2/2} \Big ) \, du;
\end{equation}
evaluating the integral as a series for $\xi = 1$ reclaims the value of $c$ in (\ref{Rr2}).
In \cite{FTZ21} the theory of the asymptotics of Fredholm Pfaffians was further developed to derive the value of $C$ independent of a certain continuity assumption implicit in the calculation of \cite{DZ96}.

It was observed in \cite{KPTTZ15} that the asymptotic result (\ref{Rr2}), with the interval of the LHS replaced by $(-s,s)$, and $s$ in the exponent of the RHS replaced by $2s$, is suggestive in relation to the large $N$ form of the probability $p_{N,N}^{\rm GinOE}$ that all the eigenvalues of an $N \times N$ matrix are real. First it is argued that the quantity $E^{\rm r, b}(0;(-s,s))$ can be replaced by the corresponding quantity for finite $N$, provided $N$ is taken as large and $(-s,s)$ is contained in the leading interval of support
$(-\sqrt{N}, \sqrt{N})$. But with $(-s,s)$ chosen as exactly this interval, up to corrections which vanish as $N$ the finite $N$ gap probability is the same quantity as $p_{N,0}^{\rm GinOE}$ for the probability of no real eigenvalues, and thus the prediction (\ref{2.10}).

Further developing the viewpoint of the above paragraph, let $E_N^{\rm r}(0;(a,b))$ denote the probability that the interval $(a,b)$ in the $N \times N$ GinOE is free of eigenvalues. From the definition of edge scaling, we see that
\begin{equation}
\lim_{N \to \infty} 
E_N^{\rm r}(0;(\sqrt{N}+s,\infty)) =
E^{\rm r, e}(0;(s, \infty));
\end{equation}
note that minus the derivative of the RHS with respect to $s$ gives the PDF of the largest eigenvalue for $N \to \infty$. It turns out that the $s \to - \infty$ asymptotics of this quantity is closely related to (\ref{Rr2}) \cite{FTZ21} (see also \cite{BB20}),
\begin{equation}\label{cpx3}
E^{\rm r, e}(0;(s,\infty))
\mathop{\sim}\limits_{s \to - \infty} e^{-c|s|  -{1 \over 2} \log 2 + C + O(s^{-1^+})}, 
\end{equation}
where $c,C$ are as in (\ref{Rr2}). Note that in both (\ref{Rr2}) and (\ref{cpx3}) a gap with no real eigenvalues is being created, which has size $|s|$ relative to the bulk region. This has been generalised to the thinned case in \cite{BB22},
\begin{equation}\label{cpx4}
E^{\rm r, e}(0;(s,\infty);\xi)
\mathop{\sim}\limits_{s \to - \infty} e^{-c(\xi)|s| + C(\xi) + o(1)}, 
\end{equation}
where $c(\xi)$ is as in (\ref{cpx}) and, with ${\rm Li}_s(z) := \sum_{n=1}^\infty z^n/ n^s$,
$$
C(\xi) = {1 \over 2} \log {2 \over 2 - \xi} + {1 \over 4 \pi}
\int_0^{2 \xi - \xi^2} \Big ( ({\rm Li}_{1/2}(x))^2 - {\pi x \over 1 - x} \Big ) \, {dx \over x}.
$$
In relation to the diffusion processes that have featured in this subsection, we remark that for the annihilation process $A + A \to \varnothing$ with half line initial positions, an exact coincidence with the edge state Pfaffian point process of real Ginibre eigenvalues can be exhibited \cite{GPTZ18}.

Knowledge that the edge scaled real Ginibre eigenvalues form a Pfaffian point process with the explicit kernel determined by the formulas of Proposition \ref{P2.5} and the limit formula (\ref{11.SXY}) imply a Fredholm Pfaffian formula for $E^{\rm r, e}(0;(s,\infty);\xi)$ (see also \cite{RS14,PTZ17}). A recent result of Baik and Bothner \cite{BB22} gives that a simpler Fredholm determinant formula holds true. For this denote by $\mathcal S_{(s,\infty)}$ the integral operator of $(s,\infty)$ with kernel
$$
S(x,y) = {1 \over 2 \sqrt{\pi}} e^{-(x+y)^2/4},
$$
and write $\bar{\xi}:= \xi (2 - \xi)$. One then has \cite[Th~1.8]{BB22}
\begin{equation}\label{cpx5}
E^{\rm r, e}(0;(s,\infty);\xi) =
\sqrt{1 - \sqrt{\bar{\xi}} \over 2(2 - \xi)}
\det\Big (\mathbb I + \sqrt{\bar{\xi}} \mathcal S_{(s,\infty)} \Big ) +
\sqrt{1 + \sqrt{\bar{\xi}} \over 2(2 - \xi)}
\det \Big ( \mathbb I - \sqrt{\bar{\xi}} \mathcal S_{(s,\infty)} \Big ).
\end{equation}
Note that with $\xi = 1$ (no thinning) the coefficient in front of the first of these Fredholm determinants vanishes, and then \cite[Th.~1.12]{BB20} 
\begin{equation}\label{cpx6}
E^{\rm r, e}(0;(s,\infty)) =
\det(\mathbb I -  \mathcal S_{(s,\infty)}).
\end{equation}
As emphasised in \cite{BB22}, both (\ref{cpx5}) and (\ref{cpx6}) are structurally identical to known Fredholm determinant formulas for the soft edge scaled largest eigenvalue of GOE matrices; see \cite{FS05} and \cite{Fo06c} respectively. Moreover, it is shown in \cite{BB20,BB22} that associated with these gap probabilities are certain integrable systems, associated with the Zakharov-Shabat inverse scattering problem, which lead to alternative formulas again identical in structure to those known for the soft edge scaled largest eigenvalue of GOE matrices in terms of transcendents from Painlev\'e II
\cite{TW96,Di05a}.

For the Fredholm determinants in (\ref{cpx6}) and (\ref{cpx5}) the methods of Bornemann \cite{Bo08,Bo09} provide for fast, high precision numerical evaluations. This in turn allows for the determination of the values of statistical quantities associated with the corresponding distributions. Specifically, in relation to the distribution of the edge scaled largest eigenvalue (no thinning), it is found  that to five decimals the mean, variance, skewness and kurtosis are equal to \cite[Table 1]{BB20}
$
-1.30319,  3.97536,  -1.76969,  5.14560,
$
respectively.

\begin{remark}
Instead of the largest real eigenvalue, one can ask about the distribution of the maximum of the real part of all GinOE eigenvalues, or the maximum modulus of all the eigenvalues. Both are known. With global scaling $G \mapsto {1 \over \sqrt{N}} G$, the first involves the extreme value scaled variable
$$
1 + \sqrt{\gamma_N \over 4 N} +
{x \over \sqrt{4 N \gamma_N}},\qquad \gamma_N = {1 \over 2}
\Big ( \log N - 5 \log \log N - \log (2 \pi^4) \Big ),
$$
while the second involves the same extreme value scaled variable as for GinUE, \cite[Eq.~(3.14)]{BF22a}. In these scalings, the extreme values obey the rescaled Gumbel law, with cumulative distribution $\exp(-{1 \over 2} e^{-t})$; see \cite{AP14,CESX22} and \cite{RS14} respectively. As already remarked, the result for the largest real part is directly relevant to the stability of the system of the matrix differential equation (\ref{1.1h}).
\end{remark}

\subsection{Elliptic GinOE} \label{Section EGinOE}
The elliptic GinOE interpolates between the asymmetric GinOE and the symmetric Gaussian orthogonal ensemble (GOE). First we recall that a member of the GOE, $S$ say, can be constructed from a GinOE matrix $G$ according to $S = {1 \over 2} (G + G^T)$. In addition to $S$, we introduce the antisymmetric real Gaussian matrix $A$ with joint distribution of its elements proportional to $e^{-{\rm Tr} \, A^2/2}$. For a parameter $\tau$, $0 < \tau < 1$, and a scale factor $b$, following \cite{LS91} a random matrix $X$ is defined according to
\begin{equation}\label{X1}
X = {1 \over \sqrt{b}} (S + \sqrt{c} A ), \qquad c = (1 - \tau)/(1 + \tau).
\end{equation}
Up to the scale factor, when $\tau = 0$, $X$ is a member of GinOE, while for $\tau=1$, $X$ is a member of GOE.

\begin{proposition}
    The joint element distribution of $X$ is
\begin{equation}\label{X2} 
A_{\tau, b} \exp \Big ( - {b \over 2(1 - \tau)} \Big (
{\rm Tr} \, X X^T - \tau {\rm Tr} \, X^2 \Big ) \Big ), \quad
A_{\tau,b} =  (\sqrt{c})^{- N(N-1)/2} (\sqrt{b})^{N^2} (2 \pi) ^{- N^2/2}.
\end{equation}
Define the normalisation and the weight
\begin{equation}\label{X2i} 
C_N^{\rm e}(\tau;b)= {(\sqrt{b})^{N(N+1)/2}
(\sqrt{1 + \tau} )^{N(N-1)/2} \over 2^{N(N+1)/4} \prod_{l=1}^N \Gamma(l/2) }, \quad
\omega^{\rm e}(z) = e^{-b |z|^2} e^{2 b y^2} 
{\rm erfc}\Big (\sqrt{2b \over 1 - \tau} y \Big ),
\end{equation}
where $z=x+iy$.
The probability density that there are $k$ real eigenvalues for the random matrix $X$ is then
\begin{equation}\label{3.3}
C_N^{\rm e}(\tau;b)
{2^{(N-k)/2} \over k! ((N-k)/2)! } \prod_{s=1}^k (\omega^{\rm e}(\lambda_s))^{1/2} \prod_{j=1}^{(N-k)/2} \omega^{\rm e}(z_j) \Big | \Delta(\{\lambda_l\}_{l=1,\dots,k} \cup
\{ x_j \pm i y_j \}_{j=1,\dots,(N-k)/2}) \Big |.  
\end{equation}
\end{proposition}

\begin{proof}
We see from (\ref{X1}) that
$$
(d  X) = 2^{N(N-1)/2} (\sqrt{c})^{N(N-1)/2} (\sqrt{b})^{-N^2} (d  S) (d  A).
$$
Using this we can change variables in the joint element distribution of $S$ and $A$ to deduce that the joint element distribution of $X$ is equal to
(\ref{X2}).

Denote (\ref{3.1}) by $P_{k, (N-k)/2}(\{\lambda_j\}_{j=1}^k; \{x_j \pm i y_j \}_{j=1}^{(N-k)/2})$,
which corresponds to the eigenvalue PDF in the case $\tau = 0$, $b=1$, conditioned so that there
are exactly $k$ real eigenvalues. For these matrices scale $ X \mapsto
\sqrt{b}  X/ (1 - \tau)^{1/2}$ to obtain that for matrices with PDF
\begin{equation}\label{15.C2}
A_{0,1} b^{N^2/2} (1 - \tau)^{- N^2/2} e^{ - b {\rm Tr} \,  X  X^T / 2 (1 - \tau)},
\end{equation}
the eigenvalue PDF~is equal to
\begin{multline*}
b^{N/2} (1 - \tau)^{- N/2} P_{k, (N-k)/2}\Big (\{ \sqrt{b} \lambda_j/(1 - \tau)^{1/2}\}_{j=1}^k; \\
   \{\sqrt{b} x_j/(1 - \tau)^{1/2} \pm i \sqrt{b}
y_j/(1 - \tau)^{1/2} \}_{j=1}^{(N-k)/2} \Big ).
\end{multline*}
Comparing (\ref{15.C2}) to (\ref{X2}) it follows that for matrices with element PDF~(\ref{X2}) the
eigenvalue PDF, conditioned so that there are exactly $k$ eigenvalues, is equal to
\begin{multline*}
{A_{\tau,b} \over A_{0,1}} (1 - \tau)^{N(N-1)/2}
\exp \Big (  {\tau b \over 2 (1 - \tau) } \Big (
\sum_{j=1}^k \lambda_j^2 + 2 \sum_{j=1}^{(N-k)/2} (x_j^2 - y_j^2) \Big ) \Big ) 
 \\
\times P_{k, (N-k)/2}\Big (\{\sqrt{b} \lambda_j/(1 - \tau)^{1/2}\}_{j=1}^k;
\{\sqrt{b} x_j/(1 - \tau)^{1/2} \pm i \sqrt{b} y_j/(1 - \tau)^{1/2} \}_{j=1}^{(N-k)/2} \Big ).
\end{multline*}
This is (\ref{3.3}).
\end{proof}

We see that (\ref{3.3}) with $b=1/(1+\tau)$ can, in the Coulomb gas picture of \S \ref{S2.1a}, be viewed as a modification of the Coulomb gas identified in relation to elliptic GinUE (see \cite[text below proof of Prop.~2.7]{BF22a}), the modification being that there are imaginary terms and charges on the real line as for GinOE itself; recall \S \ref{S2.1a}. These modifications do not alter the interpretation of the origin of the one-body couplings in terms of the same neutralising background charge as for GinUE. The conclusion then is that after scaling $X \mapsto {1 \over \sqrt{N}} X$, and with $b=1/(1+\tau)$, the limiting eigenvalue density of elliptic GinOE is the constant $1/(\pi(1-\tau^2))$ with shape of the droplet an ellipse specified by semi-axes $A=1+\tau$, $B=1-\tau$. Using different reasoning, this result was first deduced in \cite{SCSS88}.

Denote the analogue of (\ref{GP}) in the elliptic GinOE ensemble by $Z_{k,(N-k)/2}^{\rm eGinOE}[u,v]$, and use this to define $Z_N^{\rm eGinOE}[u,v]$ as specified above (\ref{12.Ca}). Starting with (\ref{3.3}), by following the strategy of  the proof of Proposition \ref{P2.2}, Pfaffian formulas for these quantities can be obtained. It suffices to record that for $Z_N^{\rm eGinOE}[u,v]$. 

\begin{proposition}\label{pp1}
Let $\{p_{l-1}(x) \}_{l=1,\dots,N}$ be a set of monic polynomials of the indexed degree,
and let $\alpha_{j,k}^{\rm e}[u], \beta_{j,k}^{\rm e}[v]$ be defined as for 
$\alpha_{j,k}^{\rm g}[u], \beta_{j,k}^{\rm g}[v]$
in Proposition \ref{P2.2} but with each $\omega^{\rm g}$ replaced by $\omega^{\rm e}$.
For $k, N$ even we have
\begin{equation}\label{Zb}
Z_N^{\rm eGinOE}[u,v] = C_N^{\rm e}(\tau;b)
 {\rm Pf} \, [  \alpha_{j,l}^{\rm e}[u] + \beta_{j,l}^{\rm e}[v] ]_{j,l=1,\dots,N}.
\end{equation}
\end{proposition}

Next we would like to compute the skew orthogonal polynomials associated with $(\alpha_{j,k}^{\rm e}[u] +
\beta_{j,l}^{\rm e}[v])|_{u=v=1}$. While the formulas (\ref{12.212f}) for the latter still apply, according to the definition 
(\ref{X1}), for $\tau \ne 0$ the elements $x_{jk}$ and $x_{kj}$ in $X$ are correlated, so the simple formula (\ref{12.212g}) is no longer valid. Nonetheless, an expansion of $\det(z \mathbb I_{2n} - G)$ can still be used to compute the averages required in (\ref{12.212f}) ---
see \cite[proof of Prop.~11]{FG06} for details of a closely related calculation --- with the result
\begin{equation}\label{RC}
p_{2n}^{\rm e}(z) = C_{2n}(z), \quad p_{2n+1}^{\rm e}(z) = C_{2n+1}(z) - 2n C_{2n-1}(z)=
  - (1 + \tau) e^{z^2 \over 2(1 + \tau)}
{d \over dz} \Big ( e^{- {z^2 \over 2 (1 + \tau)}} C_{2n}(z) \Big ),
\end{equation}
where $C_n(z) :=  ( \tau / 2 )^{n/2} H_n ( z / \sqrt{2 \tau} ) $ and 
for convenience the specific choice of scale
$b= 1/(1 + \tau)$ has been made.
These polynomials were first identified in \cite{FN08p} without knowledge of (\ref{12.212f}). The corresponding normalisation is given by
\begin{equation}\label{RC1}
r_n^{\rm e}  = (2n)! 2 \sqrt{2 \pi} (1 + \tau).
\end{equation}
We can check that (\ref{RC}) and (\ref{RC1}) reduce to the result of Proposition \ref{P2.4} in the limit $\tau \to 0^+$.

From knowledge of the skew orthogonal polynomials, the entries in the correlation kernels specifying the Pfaffian point process are explicit. For example, in the case of the real-real correlations, these entries are specified in Proposition \ref{P2.5} and it is the quantity $S_N^{\rm r}(x,y)$ therein which determines the other entries. It is possible to simplify the summation in its definition to obtain a structure identical to that in the $\tau = 0$ case (\ref{13.222}), in which the first term is recognised from its appearance as the correlation kernel for elliptic GinUE \cite[Eq.~(2.36)]{BF22a} (up to proportionality, with $w,z$ real, and with $N \mapsto N-1$).

\begin{proposition}\label{pp2}
For elliptic GinOE with $b=1/(1+\tau)$ we have
\begin{equation}\label{rttA}
S_{N}^{\rm r, e}(x,y) =
{e^{-(x^2+y^2)/2(1 + \tau)} \over \sqrt{2 \pi} }
\sum_{k=0}^{N-2} {1 \over k!} C_k(x) C_k(y) +
{ e^{- y^2/2(1 + \tau)} \over 2 \sqrt{2 \pi} (1 + \tau) }
{C_{N-1}(y) \Phi_{N-2}(x) \over (N-2)!}.
\end{equation}
\end{proposition}

\begin{proof}
    The derivative formula for $p_{2n+1}$ in (\ref{RC}) implies a summation identity analogous to (\ref{Ae1}). We can furthermore check from (\ref{RC}) and (\ref{RC1}) that analogous to (\ref{Ae2}) there is the further derivative formula
    $$
{2(k+1) \over r_{k+1}^{\rm e}} p_{2k+2}(x) - {1 \over r_k^{\rm e}}p_{2k}(x) = -{1 \over 2 \sqrt{2 \pi}} {1 \over (2k+1)!}
e^{x^2/2(1+\tau)}{d \over dx} e^{-x^2/2(1 + \tau)} C_{2k+1}(x),
$$
which implies a summation identity analogous to (\ref{Ae3}).
\end{proof}

The result (\ref{rttA}) can be used to show that both the bulk scaled, and edge scaled real-real correlation functions of elliptic GinOE are identical to that for the original GinOE, with the change of scale in the latter
$
\lambda_j \mapsto \lambda_j/\sqrt{1-\tau^2}$, 
$z_j \mapsto z_j/\sqrt{1-\tau^2}
$ \cite{FN08p}. This implies that, structurally, the formula (\ref{eks2}) for the variance of the number of real eigenvalues with $N$ large, remains true. Another application of (\ref{rttA}) is in relation to the weakly non-symmetric limit, obtained by the scaling 
$\tau \mapsto 1 - \alpha^2/N$
 before the computation of the large $N$ limit \cite{FN08p,AP14}. Thus
 $$
 \lim_{N \to \infty} {\pi \over \sqrt{N}} 
S_{N}^{\rm r, e}\Big ({\pi x \over \sqrt{N}},{\pi y \over \sqrt{N}} \Big )  \Big |_{\tau = 1 - \alpha^2/N} = \int_0^1 e^{-\alpha^2 u^2} \cos \pi u(x-y) \, du,
$$
which implies in particular that the correlations now decay algebraically, in contrast to the Gaussian decay exhibited by (\ref{11.SXY+}).

Setting $x=y$ in (\ref{rttA}) gives the eigenvalue density.
From this, with $\tau$ fixed,
 it was shown in \cite{FN08p,BKLL21} that the expected number of real eigenvalues $ E_N^{\rm r, e} $ satisfies the large $N$ asymptotic expansion
\begin{equation}
 E_N^{\rm r, e} -\frac12  \mathop{\sim}\limits_{N \to \infty}  \sqrt{ \frac{2}{\pi} \frac{1+\tau}{1-\tau} N } \bigg( 1+ \frac{\tau-3}{8(1-\tau)} \frac{1}{N} + \frac{5\tau^2-14\tau-3}{128(1-\tau)^2} \frac{1}{N^2} + \cdots   \bigg), 
\end{equation}
which recovers \eqref{12.212d} for $\tau=0$.
On the other hand, in the scaling $\tau \mapsto 1 - \alpha^2/N$, one finds
\begin{equation}
  E_N^{\rm r, e}  \mathop{\sim}\limits_{N \to \infty}   c(\alpha) N + O(1), \qquad c(\alpha):= e^{-\alpha^2/2} \bigg( I_0\Big( \frac{\alpha^2}{2} \Big)+ I_1\Big( \frac{\alpha^2}{2} \Big) \bigg);
\end{equation}
see \cite[Th.~2.1]{BKLL21}. 
(We mention that the function $c(\alpha)$ also appears in several different contexts; cf. \cite[Eq.~(3.16) and Prop. 6.4]{BF22a}.)
Furthermore, as an analogue of Proposition~\ref{P2.10c}, it was shown in \cite[Th.~2.3]{BKLL21} that in the weakly non-symmetric regime the variance $(\sigma_N^{\rm r, e})^2$ is again proportional to $ E_N^{\rm r, e}$ as  
\begin{equation} \label{var eg ws}
   (\sigma_N^{\rm r, e})^2 \sim \bigg( 2-\frac{c(\sqrt{2}\alpha)}{ c(\alpha) } \bigg) E_N^{\rm r,e};  
\end{equation}
note that  as $\alpha \to \infty$ the ratio in \eqref{var eg ws} tends to $2-\sqrt{2}$ in \eqref{eks2}.
The global density of real eigenvalues for $\tau$ fixed tends to be uniform in the interval $(-1-\tau,1+\tau)$; see \cite{FN08p}. 
In contrast, in the regime $\tau \mapsto 1 - \alpha^2/N$, it was proposed by Efetov \cite{Efe97} using the supersymmetry method that for $x \in (-2,2)$,
\begin{equation} \label{rho eg ws density}
 \lim_{N \to \infty} \rho_N^{ \rm r,e }(\sqrt{N}x) \Big |_{\tau = 1 - \alpha^2/N} = \frac{1}{c(\alpha)} \frac{1}{2\alpha \sqrt{\pi}} {\rm{erf}}\Big( \frac{\alpha}{2} \sqrt{4-x^2} \Big). 
\end{equation}
This convergence was recently shown in \cite{BKLL21} by properly analysing the double scaling limits of (\ref{rttA}). 
As expected, the limiting density in \eqref{rho eg ws density} interpolates between the uniform density ($\alpha \to \infty$) with the semi-circle law ($\alpha \to 0$). A point of interest is that Efetov's was motivated by an application of almost symmetric random matrices to the study of vortices in disordered superconductors with columnar defects \cite{HN96}.

\begin{remark} $ $ \\
1.~It can be checked from the formula for $p_{2n}^{\rm e}(z)$ and the first formula for $p_{2n}^{\rm e}(z)$ in (\ref{RC}),
together with the formula (\ref{RC1}) for the normalisation 
that the analogue of (\ref{2.37a}), with $K_{N-1}$ now the kernel for elliptic GinUE \cite[Eq.~(2.36)]{BF22a}, is also valid for $S_{N}^{\rm r, e}(x,y)$. \\
2.~Denote the probability density (\ref{3.3}) by $P_{N,k}(\{\lambda_j\}_{j=1}^k;\{x_j\pm i y_j\}_{j=1}^{(N-k)/2};\tau,b)$. It follows from (\ref{X2i}) and (\ref{3.3}) that
$$
P_{N,N}(\{\lambda_j\}_{j=1}^N;\tau,b) = (1 + \tau)^{N(N-1)/4}P_{N,N}(\{\lambda_j\}_{j=1}^N;\tau=0,b=1).
$$
Hence upon integration and use of Proposition \ref{P2.1r},
$p_{N,N}^{\rm eGinOE} =  ( (1 + \tau)/ 2)  )^{N(N-1)/4}$. Note in particular that this is equal to unity for $\tau = 1$, as is consistent with $\tau=1$ corresponding to real symmetric matrices.\\
3.~A variation of the construction of elliptic GinOE matrices (\ref{X1}) is to consider asymmetric random matrices $S + A_0$, where $S \in {\rm GOE}$, while $A_0$ is a fixed real antisymmetric matrix. An analysis of this setting in \cite{FTS98} reclaimed the result of \cite{Efe97} for the density of both the real (as given by (\ref{rho eg ws density})) and complex eigenvalues, although this was shown to break down if $A_0$ was of finite rank.
\end{remark}

\subsection{An application of elliptic GinOE to equilibria counting}
In this subsection, following 
Fyodorov and  Khoruzhenko \cite{FK16}, it will be shown that generalisations of the considerations underpinning the random differential equation (\ref{1.1h}) as formulated in \cite{Ma72a} relate to elliptic GinOE. The starting point (see the review \cite{AT15} for an extended description) is the coupled nonlinear differential equations
\begin{equation}\label{1.r}
{d X_i(t) \over dt} = f_i(\mathbf X(t)), \qquad i=1,\dots,N.
\end{equation}
Here the $f_i$ are not known explicitly.
An equilibrium point $\mathbf X^*$ is when $f_i(\mathbf X^*) = 0$. Stability around the fixed point is probed in terms of the Jacobian $M_{ij} = {\partial f_i(\mathbf X) \over \partial X_j} |_{\mathbf X^*}$. May \cite{Ma72a}, on consideration of the applied setting within ecology, took the diagonal elements to have mean $-1$. Taking the off diagonal elements to have mean zero and all elements to have variance $\alpha^2$ then leads to (\ref{1.1h}).

The study \cite{FK16} takes a global approach to the study of equilibria in (\ref{1.r}), with the RHS written with the addition of $-\mu X_i(t)$. This allows the question: ``what is the probability that a randomly chosen equilibrium is stable?" to be probed. To proceed further, the $f_i$ are decomposed into the sum of a gradient (curl-free) and solenoidal (divergence free) components,
$$
f_i(\mathbf X) = - {\partial V(\mathbf X) \over \partial X_i} +
{1 \over \sqrt{N}} \sum_{j=1}^N {\partial A_{ij}(\mathbf X) \over \partial X_i}.
$$
The matrix $A(\mathbf X)$ is anti-symmetric. Both $V(\mathbf X)$ and $A(\mathbf X)$ are assumed to be statistically independent, isotropic, centred, homogeneous random Gaussian fields with covariances
$$
\langle V(\mathbf X) V(\mathbf Y) \rangle = v^2 \Gamma_V(|\mathbf X - \mathbf Y|^2), \quad
\langle A_{ij}(\mathbf X) A_{nm}(\mathbf Y) \rangle = a^2 \Gamma_A(|\mathbf X - \mathbf Y|^2) ( \delta_{i,n} \delta_{j,m} - \delta_{i,m} \delta_{j,n} ),
$$
and with $\Gamma_V''(0) = \Gamma_A''(0)=1$ as normalisations. 

The primary observable considered in \cite{FK16} is the expected number $\langle \mathcal N \rangle$ of equilibrium points implied by this model. Through use of the Kac-Rice formula, this is reduced to the form of a random matrix average
\begin{equation}\label{tm1}
\langle \mathcal N \rangle 
= \Big \langle  |
\det \Big [\delta_{i,j}(1 + \xi \sqrt{\tau}/m) - {1 \over m \sqrt{N}} J_{ij} \Big ]_{i,j=1}^N \Big \rangle, \quad  m = {\mu \over \sqrt{4N(v^2 + a^2)}}, \: \:
 \tau = {v^2 \over v^2 + a^2},
\end{equation}
where $\xi$ is distributed as N$[0,1/\sqrt{N}]$ and the $J_{ij}$ are entries of a zero mean Gaussian random matrix with correlations (for large $N$)
\begin{equation}\label{tm2}
\langle J_{ij}J_{kl} \rangle =
 \Big ( \delta_{i,k} 
\delta_{j,l} + \tau (\delta_{i,j}
\delta_{k,l} + \delta_{i,l}
\delta_{j,k}) \Big ).
\end{equation}
In the case that $\tau = 0$ the entries of $[J_{ij}]$ are seen to be statistically independent, which relates to May's setting.
On the other hand, when $\tau = 1$, (\ref{tm2}) specifies $[J_{ij}]$ as proportional to a real symmetric GOE matrix. For general $0 < \tau < 1$,
$[J_{ij}]$ defines an elliptic GinOE matrix as specified in Section \ref{Section EGinOE}.

In the theory of the real eigenvalues of GinOE matrices, note has been made of the formula (\ref{eks1}) relating the average of the absolute value of the characteristic polynomial to the density. Such a formula holds equally as well for the elliptic GinOE, and gives \cite[Eq.~(12)]{FK16}
\begin{equation}\label{tm1+}
\langle \mathcal N \rangle 
= {C_N \over m^N} 
\int_{-\infty}^\infty e^{-N S(\lambda)} \rho_{(1),N+1}^{\rm r,e}(\lambda \sqrt{N}) \, d \lambda,
\end{equation}
where
$$
C_N =  {2 \sqrt{1+1/\tau}(N-1)! \over N^{(N-1)/2} (N-2)!!}, \qquad
S(\lambda) = (\lambda - m)^2/(2 \tau) - \lambda^2/(2(1 + \tau)).
$$
The minimum of $S(\lambda)$ occurs at $\lambda^* = m (1 + \tau)$. On the other hand we know that in the global variable $\lambda \sqrt{N}$ the support of the real eigenvalue for elliptic GinOE is $\lambda| < 1 + \tau$. Hence $\lambda^*$ lies inside the support for $m < 1$, and outside otherwise. Moreover, inside the support $\rho_{(1),N+1}^{\rm r,e}(\lambda \sqrt{N})$ has the large $N$ value
$1/\sqrt{2 \pi (1 - \tau^2)}$. It follows that for $m < 1$ 
\cite[Eq.~(14)]{FK16}
\begin{equation}\label{tm2+}
\langle \mathcal N \rangle =
\sqrt{2(1 + \tau) \over (1 - \tau)} e^{N((1/2)(m^2 -1) - \log m)},
\end{equation}
which grows exponentially in $N$.
Analysis in the case of $m > 1$ requires knowledge of $\rho_{(1),N+1}^{\rm r,e}(\lambda \sqrt{N})$ in the large deviation regime outside of the support. This is derived in \cite{FK16} starting from the finite $N$ expression (\ref{rttA}) with $x=y$, which in turn is used to deduce that then $\langle \mathcal N \rangle \to 1$. Moreover, in \cite[Eq.~(16)]{FK16} a formula involving (\ref{11.SXYa}) was given which interpolates between (\ref{tm2+}) and the regime of a single equilibrium for $m$ scaled about unity.

\section{Comparisons between the GinOE and GinUE}
\subsection{Bulk and edge correlations} \label{Section_GinOUE cor}
The bulk correlation kernels for the real-real, real-complex and complex-complex correlations are given by (\ref{11.SXYp}) and the formulas of Remark \ref{R2.10}.3. However with regard to the latter two, it needs to be clarified that here the term bulk is used in the sense of being away from the spectrum boundary, but still in the vicinity of the real axis.
Suppose we take the limits $v,y \to \infty$ with $v-y$ fixed in $\mathcal K_\infty^{\rm c}$ of Remark \ref{R2.10}.3. Then from the large $s$ form of ${\rm erfc}(\sqrt{2} s)$ we calculate
\begin{equation}\label{3.11a}
\lim_{v,y \to \infty \atop |v-y| \: {\rm fixed}} 
\mathcal K_\infty^{\rm c}(w,z)=
{1 \over \pi} \begin{bmatrix}
    0 & e^{-(|w|^2 + |z|^2)/2}e^{w \bar{z}} \\
    - e^{-(|w|^2 + |z|^2)/2}e^{\bar{w} {z}} & 0 
\end{bmatrix}.
\end{equation}
The Pfaffian of this kernel as required according to the first displayed formula of Remark \ref{R2.10}.3 with $k_1=0, k_2 = k$ to compute the $k$-point complex-complex correlation therefore reduces to a determinant. Moreover, the correlation kernel in the determinant is precisely that for bulk scaled GinUE \cite[Eq.~(2.18)]{BF22a}.

The edge scaling of the finite $N$ complex-complex correlation kernel (\ref{cE2}) is readily calculated from (\ref{cE2a}). Thus in the neighbourhood of the spectrum edge of the real eigenvalues one finds \cite[Corollary 9]{BS09}
\begin{equation}\label{3.11b}
    \lim_{N \to \infty}
    S_N^{\rm c}(\sqrt{N} + w,
    \sqrt{N} + z) = 
    {i(\bar{z}-w
) \over \sqrt{2 \pi}}
\Big ( {\rm erfc}(\sqrt{2}u)  e^{2u^2}  {\rm erfc}(\sqrt{2}y)  e^{2y^2}  \Big )^{1/2} e^{-{1 \over 2}(w - \bar{z})^2}.
\end{equation}
Now taking $v,y \to \infty$ with $v-y$ fixed in $\mathcal K_\infty^{\rm c}$ gives for the analogue of (\ref{3.11a})
\begin{equation}\label{3.11c}
\lim_{v,y \to \infty \atop |v-y| \: {\rm fixed}} 
\mathcal K_\infty^{\rm c,e}(w,z)=
 \begin{bmatrix}
    0 & K_\infty^{\rm e}(w,z) \\
    -  K_\infty^{\rm e}(\bar{w},\bar{z}) & 0 
\end{bmatrix}.
\end{equation}
Here $K_\infty^{\rm c,e}$ denotes the correlation kernel for edge scaled GinUE as given in \cite[Eq.~(2.21)]{BF22a} with $\nu =1$.
Again the corresponding Pfaffian as determines the edge complex-complex correlations reduces to a determinant, with the latter that for edge scaled GinUE with origin $(\sqrt{N},0)$.

For real random matrices of identical and independently distributed,  zero mean and fixed standard deviation entries, a universality result for their edge statistics agreeing with those for GinOE has been established in \cite{CES21b}. 

\subsection{Global density and fluctuations}
We know that with global scaling, the density for both the GinOE and GinUE obeys the circular law \cite[Eq.~(2.17)]{BF22a}. In the GUE case convergence to this limit can be probed by computing the radial moments $\langle {1 \over N} \sum_{j=1}^N |r_j|^p \rangle$ ($p=1,2,\dots$); recall \cite[\S 4.2 and Eq. (4.13) with $\beta = 2$]{BF22a}. Instead of the radial moments, more natural in the case of GinOE are averages of the eigenvalue power sums $ \sum_{j=1}^N z_j^k$. We have
\begin{equation}\label{dfd}
\Big \langle \sum_{j=1}^N z_j^k \Big \rangle =
\int_{-\infty}^\infty x^k \rho_{(1),N}^{\rm r}(x) \, dx +
\int_{\mathbb R_+^2} ( (x+iy)^k+
(x-iy)^k ) \rho_{(1),N}^{\rm c}((x,y)) \, dx dy.
\end{equation}
Only for $k$ even is this nonzero, so we write $k=2p$ with $p$ a non-negative integer. Note that for the GinUE, this average vanishes for all positive integers $p$, due to rotation invariance. Although the exact functional forms of $\rho_{(1),N}^{\rm r}(x)$ (\ref{11.SXYa})) and of
$\rho_{(1),N}^{\rm c}((x,y))$
(\ref{cE3}) are known, evaluation 
of (\ref{dfd}) in a structured form following from direct integration does not seem possible. However, a less direct approach in which the average of a general Schur polynomial in GinOE eigenvalues is first computed, implies the sort structured evaluation \cite[\S 4]{SK09} (see also \cite{FR09})
\begin{equation}\label{dfd.1}
\Big \langle \sum_{j=1}^N z_j^{2p} \Big \rangle = \prod_{j=1}^{p} (N + 2p - 2j).
\end{equation}

Changing variables $z_j \mapsto \sqrt{N} z_j$, and making an ansatz of expansions in $1/\sqrt{N}$,
$$\rho_{(1),N}^{\rm c}(\sqrt{N}(x,y))=
{1 \over \pi} \chi_{|z|<1} + {1 \over \sqrt{N}} \mu_{(1),1}^{\rm c}((x,y)) + \cdots, \: \:
\rho_{(1),N}^{\rm r}(\sqrt{N}x)=
{1 \over \sqrt{2\pi}}\chi_{|x|<1} + \cdots,
$$
the equating of powers of $N$ with (\ref{dfd.1}) substituted on the LHS allows deductions about the functional form in the ansatz to be made.

\begin{proposition}\label{pp3}
We have
\begin{equation}\label{fj1}
    \mu_{(1),1}^{\rm c}((x,y)) = - {1 \over \sqrt{2 \pi}} \delta(y) \chi_{|x|<1}, 
\end{equation}
which is valid up to the possible addition of a rotationally invariant functional form which integrates to zero. 
\end{proposition}

\begin{proof}
After making the substitutions, we see that the highest power on the LHS is $N^p$, while on the RHS the highest power is $N^{p+1}$, 
with next highest power of order $N^{p+1/2}$. Thus we begin by equating the term of order $N^{p+1}$ to zero, which gives ${1 \over \pi} \int_{|z|<1} z^{2p} \, d^2 z = 0$. This is identically true since $p$ is a positive integer. Equating the terms of order $N^{p+1/2}$ to zero gives
$$
{1 \over \sqrt{2 \pi}} \int_{-1}^1 x^{2p} \, dx + 
\int_{\mathbb C} z^{2p} \mu_{(1),1}^{\rm c}((x,y)) \, d^2z = 0,
$$
from which we deduce (\ref{fj1}). 
\end{proof}
We remark that the result in (\ref{fj1}) has, in the Coulomb gas picture, the interpretation of perfect screening of the charge density on the unit interval by the charges in the disk. With the next term in the $N^{-1/2}$ expansion of $\rho_{(1),N}^{\rm r}(\sqrt{N}x)$ denoted $\mu_{(1),1}^{\rm r}(x)$,
the exact functional form (\ref{11.SXYa}) of $ \rho_{(1),\infty}^{\rm r, e}(X)$ can be used for its evaluation. Thus
\begin{multline}\label{fj3u}
    \int_{-\infty}^\infty x^{2p} \mu_{(1),1}^{\rm r}(x) \, dx =
2 \lim_{N \to \infty}  N^{-p}
\int_{0}^\infty x^{2p} \Big ( \rho_{(1),N}^{\rm r}(x) - 
{1 \over \sqrt{2\pi}} \chi_{|x| < \sqrt{N}} \Big ) \, dx \\
= 2 \int_{-\infty}^\infty \Big ( \rho_{(1),\infty}^{\rm r, e}(x) - 
{1 \over \sqrt{2\pi}} \chi_{x < 0}
\Big ) \, dx = {1 \over 2},
\end{multline}
where the second equality follows by shifting the origin to the right spectrum edge $x = \sqrt{N}$ according to the change of variables $x \mapsto \sqrt{N} + x$, and the final equality is the explicit evaluation of the integral using (\ref{11.SXYa}). This implies 
\begin{equation}\label{fj2}
\mu_{(1),1}^{\rm r}(x) ={1 \over 4} \Big ( \delta(x-1) + \delta(x+1) \Big ).
\end{equation}

Let the order $1/N$ term in the expansion of $\rho_{(1),N}^{\rm c}(\sqrt{N}(x,y))$ be denoted $\mu_{(1),2}^{\rm c}((x,y))$.
Equating terms of order $N^{p}$ on both sides of (\ref{dfd.1}) and using too (\ref{fj2}) gives
\begin{equation}\label{fj3}
\int_{\mathbb C} z^{2p} \mu_{(1),2}^{\rm c}((x,y)) \, d^2z = {1 \over 2}.
\end{equation}
However unlike the circumstance for (\ref{fj3u}) this does not allow us to deduce $\mu_{(1),2}^{\rm c}((x,y))$. Instead, returning to (\ref{cE3}), with the knowledge that the ratio of the incomplete gamma function to the gamma function therein approaches unity exponentially fast in $N$ for $|z|<1$ (see \cite[below Eq.~(4.14) for references]{BF22a}, simple asymptotics associated with ${\rm erfc}(\sqrt{2N} y)$ gives the expansion \cite[Remark 2.6]{CES21a} 
\begin{equation}\label{fj4}
\mu_{(1),2}^{\rm c}((x,y)) = {1 \over \pi} - {1 \over 4 \pi N y^2}
+ {\rm O}\Big ( {1 \over N^{3/2}} \Big ).
\end{equation}
While an order $1/N$ correction is clearly distinguished,
in general this cannot be integrated in the full unit disk due to the order $1/y^2$ singularity as the real axis is approached. In fact it has been proved by Cipolloni, Erd\"os and Schr\"oder \cite[Th.~2.2]{CES21a} that for a large class of test functions $f$,
\begin{multline}\label{fj5}
   \Big \langle {1 \over N }\sum_{j=1}^N f(z_j) \Big \rangle^{\rm g} - {1 \over \pi} \int_{|z|<1}f(z) \, d^2z =
   {1 \over N}\bigg ({1 \over 4 \pi} \int_{|z|<1} {f(x) - f(z) \over y^2} \, d^2z + {1 \over 2 \pi} \int_{-1}^1 {f(x) \over \sqrt{1 - x^2}} \, dx \\
   - {1 \over 2 \pi} \int_0^{2 \pi} f(e^{i \theta}) \, d \theta + {1 \over 4} \Big ( f(x-1) + f(x+1) \Big ) \bigg ) + {\rm o} (N^{-1}),
   \end{multline}
   where the superscript "g" on the average indicates the use of global scaling. Here $f(x) = f(z) |_{y=0}$
   and similarly the interpretation of $f(e^{i \theta})$.
   This shows that the correction (\ref{fj1}) completely cancels the expected term from the density of real eigenvalues at order $N^{-1/2}$. The final term is identified with the real density correction  (\ref{fj2}). For $f$ with support off the real axis, the contribution from (\ref{fj5}) is recognised. However, if this is not the case,  a formula for $\mu_{(1),2}^{\rm c}((x,y))$ cannot be read off. It can be verified that with $f(z) = z^{2p}$,
   (\ref{fj5}) is consistent with the leading $N$ form of (\ref{dfd.1}) \cite[Remark 2.7]{CES21a}. Also, we can check that the RHS vanishes for $f$ a constant, as it must.

   A fundamental question is the variance associated with the linear statistic $\sum_{j=1}^N f(z_j)$, and furthermore its distribution. In the case of GinUE, we know that the behaviour is different depending on the smoothness of $f$. On the other hand, we have already seen for the real eigenvalues of GinOE that the fluctuations are proportional to $\sqrt{N}$ independent of this property. However, for smooth $f$ it is known that when averaging over all the eigenvalues, this effect is suppressed, and the again the limiting variance is of order unity
   \cite{Ko15,OR16,CES21a}.

   \begin{proposition}\label{P3.2y}
       Let $f,g$ be smooth real or complex valued functions in the plane, and subject to a constraint on their growth. Define the Fourier components of their Fourier expansion on the unit circle in the plane as in \cite[Prop.~3.5]{BF22a}. Define too $f^{\rm s}={1 \over 2}(f(x,y)+f(x,-y)$ and similarly the meaning of $g^{\rm s}$. We have
  \begin{equation}\label{5.2eX} 
  \lim_{N \to \infty} {\rm Cov}^{\rm GinOE} \Big ( \sum_{j=1}^N f(\mathbf r_j/\sqrt{N}),  \sum_{j=1}^N \bar{g}(\mathbf r_j/\sqrt{N}) \Big ) =
  {1 \over 2 \pi  } \int_{ |\mathbf r  | < 1}   \nabla f^{\rm s} \cdot  \nabla \bar{g}^{\rm s} \, dx dy +
   \sum_{n=-\infty}^\infty |n| f_n^{\rm s}  \bar{g}_{-n}^{\rm s};
   \end{equation}
   cf.~\cite[Eq.~(3.21)]{BF22a}.
   \end{proposition}

   \begin{proof} (Comments only)
   With $X \in {\rm GinOE}$, define the $2N \times 2N$ Hermitian matrix
   $$
   H^z = \begin{bmatrix}
   0_{N \times N} & X - z \mathbb I_N \\
   X^\dagger - z \mathbb I_N & 0_{N \times N} \end{bmatrix},
   $$
   and write $G^z(w) = (H^z - w \mathbb I_N )^{-1}$ for the resolvent of $H^z$. The starting point of the calculations of \cite{CES21a} is the formula
   $$
   \sum_{j=1}^N f(z_j) = - {1 \over 4 \pi } \int_{\mathbb C} \nabla^2 f(z) \int_0^\infty {\rm Im} \,
   {\rm Tr} \, G^z(i \eta) \, d \eta d^2z,
   $$
   where $\{z_j\}$ are the eigenvalues of $X$ (both real and complex).
   Hence the task is reduced to first finding the limiting covariance for the traces of the resolvents. Aiding this is the fact that for large $N$, $G^z(w)$ becomes approximately deterministic, with its limit expressed in terms of the solution of the scalar equation
   $$
   - {1 \over m^z} = w + m^z - 
   {|z|^2 \over w + m^z}, \quad
   \eta {\rm Im} \, m^z >0, \:\: \eta = {\rm Im}\, w \ne 0,
   $$
which is a special case of the matrix Dyson equation; see
\cite{Er19} for an introduction to the latter.
\end{proof}

\begin{remark}
1.~It is elementary to compute that for $G \in {\rm GinOE}$,
\begin{equation}\label{sas}
{\rm Var}^{\rm g}( {\rm Tr} \, G) :=
\langle ({\rm Tr} \, G)^2 \rangle^{\rm g} -  ( \langle ({\rm Tr} \, G \rangle^{\rm g}  )^2 = 1.
\end{equation}
Since ${\rm Tr} \, G = \sum_{j=1}^N z_j$, this corresponds to the circumstance that $f^{\rm s} = g^{\rm s} = x$ on the RHS of (\ref{5.2eX}), and furthermore $f_n^{\rm s} = g_n^{\rm s} = {1 \over 2}$ for $n = \pm 1$, $f_n^{\rm s} = g_n^{\rm s} = 0$ otherwise. We verify agreement between the value obtained from the general formula (\ref{5.2eX}) and (\ref{sas}). \\
2.~In the setting of Proposition \ref{P3.2y}, and after centring by subtracting from the linear statistic its mean, the results of \cite{Ko15,OR16,CES21a} give that for $N \to \infty$ a central limit theorem holds, whereby the distribution is a mean zero Gaussian with variance as implied by (\ref{5.2eX}).
\end{remark}
   \subsection{Sum rules}
   In this subsection we view the eigenvalues from the Coulomb gas perspective of \S \ref{S2.1a}. Suppose we fix a real eigenvalue at  the origin. Requiring that the corresponding total charge be cancelled by a redistribution of all the other charges implies the sum rule
\begin{equation}\label{su1} 
2 \int_{\mathbb C_+} \rho_{(2),\infty}^{{\rm c,b},T}(0, z) \, d^2z + \int_{-\infty}^\infty 
\rho_{(2),\infty}^{{\rm r,b},T}( 0, y) \, dy = - \rho^{\rm r}.
\end{equation}
If instead the charge was fixed in the upper half plane at a point $z_0$, then (\ref{su1}) would need to be modified to read
\begin{equation}\label{su2} 
2 \int_{\mathbb C_+} \rho_{(2),\infty}^{{\rm c,b},T}(z_0, z) \, d^2z + \int_{-\infty}^\infty 
\rho_{(2),\infty}^{{\rm (c,r),b},T}(z_0, x) \, dx = - 2 \rho_{(1),\infty}^{\rm c,b}(z_0).
\end{equation}
As for in the one-component case (recall \cite[\S 4.3]{BF22a}) the sum rules (\ref{su1}) and (\ref{su2}) can be generalised to involve higher point correlation functions. In stating the general result it is convenient to make use of the truncated correlations associated with the latter; see e.g.~\cite[\S 5.1.1]{Fo10} for their definition.

\begin{proposition}\label{P3.4}
 We have
 \begin{eqnarray}\label{56c}
&&
\int_{-\infty}^\infty \rho_{(k_1+1,k_2),\infty}^T(\{x_j\}_{j=1,\dots,k_1} \cup \{y\};\{z_j\}_{j=1,\dots,k_2} ) \, dy 
\nonumber \\
&&
\qquad \qquad + 2 \int_{\mathbb R_+^2} \rho_{(k_1,k_2+1),\infty}^T(
\{x_j\}_{j=1,\dots,k_1}; \{z_j\}_{j=1,\dots,k_1} \cup \{z\}) \, d^2 z 
\nonumber \\
&& \qquad =
- (k_1 + 2 k_2) \rho_{(k_1,k_2),\infty}^T(\{x_j\}_{j=1,\dots,k_1}; \{z_j\}_{j=1,\dots,k_2}).
\end{eqnarray} 
\end{proposition}

\begin{proof}
Underlying this sum rule are the matrix correlation kernel identities
 \begin{multline*}
  \int_{-\infty}^\infty \mathcal K^{\rm r}(x,u) \mathcal K^{\rm r}(u,y) \, du +
 2 \int_{\mathbb C_+}  \mathcal K^{\rm r,c}(x,z) \mathcal K^{\rm c,r}(z,y) \, d^2 z
  \\
  = \mathcal K^{\rm r}(x,y) 
 \begin{bmatrix}0 & 0 \\ 0 & 1 \end{bmatrix} +
 \begin{bmatrix}1 & 0 \\ 0 & 0 \end{bmatrix}\mathcal K^{\rm r}(x,y),
 \end{multline*}
 \begin{eqnarray*}
 && \int_{-\infty}^\infty \mathcal K^{\rm c,r}(z,y) \mathcal K^{\rm r,c}(y,v) \, dy +
 2 \int_{\mathbb C_+}  \mathcal K^{\rm c}(z,w) \mathcal K^{\rm c}(w,v) \, d^2 w 
 = 2 \mathcal K^{\rm c}(z,v), \\
  && \int_{-\infty}^\infty \mathcal K^{\rm c,r}(z,y) \mathcal K^{\rm r}(y,x) \, dy +
 2 \int_{\mathbb C_+}  \mathcal K^{\rm c}(z,w) \mathcal K^{\rm c,r}(w,x) \, d^2 w
 =  \mathcal K^{\rm c,r}(z,x), \\
 && \int_{-\infty}^\infty \mathcal K^{\rm r}(x,y) \mathcal K^{\rm r,c}(y,z) \, dy +
 2 \int_{\mathbb C_+}  \mathcal K^{\rm r,c}(x,w) \mathcal K^{\rm c}(w,z) \, d^2 w
 =  \mathcal K^{\rm r,c}(x,z).
  \end{eqnarray*}
  Further details are given in \cite[\S 4.4]{FM11}. We remark that identities of this sort in relation to the correlation kernel for a random matrix ensemble with orthogonal symmetry first appeared in \cite{Dy70,MM91}.
  \end{proof}

 As discussed in \cite[\S 4.3]{BF22a}, the fast decay of the correlations imply that not only the total charge in the screening cloud, but also its integer moments should vanish. We will illustrate this in relation to (\ref{su2}). First, in forming the moments we should interpret $ 2  \rho_{(2),\infty}^{{\rm c,b},T}(z_0, z)$ as $\rho_{(2),\infty}^{{\rm c,b},T}(z_0, z) + \rho_{(2),\infty}^{{\rm c,b},T}(z_0, \bar{z})$, thus making the contribution of the image explicit.  Similarly, we should interpret $- 2 \rho_{(1),\infty}^{\rm c}(z_0)$ as the integral over $z$ of $- \delta(z_0-z)
\rho_{(1),\infty}^{\rm c,b}(z) - \delta(z_0-\bar{z})
\rho_{(1),\infty}^{\rm c,b}(\bar{z})$. Weighting the integrands by $(z-z_0)^k$ then gives the $k$-th moments. We can see that this is identically zero for $k$ odd. Our interest then is in the sum rule implied by the $k=2p$ even case.

\begin{proposition}
For $p$ a non-negative integer we have
\begin{equation}\label{su3}
2 \int_{\mathbb R^2_+} w^{2p} \rho_{(0,2)}^T(z,w) \, d^2 w +
\int_{-\infty}^\infty x^{2p}  \rho_{(1,1)}^T(z,x) \, dx = - 2 z^{2p} \rho_{(1),\infty}^{\rm c, b}(z).
\end{equation}
\end{proposition}

\begin{proof}
Multiplying by  $\alpha^p/p!$, for $|\alpha|<1$ a parameter
and summing over $p$ replaces the corresponding terms by exponentials. This transformed identity can be checked to be a corollary of an appropriately exponential weighted version of the second of the sum rules listed in the proof of Proposition \ref{P3.4} with $v=z$,
$$
\int_{-\infty}^\infty e^{\alpha y^2}\mathcal K^{\rm c,r}(z,y) \mathcal K^{\rm r,c}(y,z) \, dy +
 2 \int_{\mathbb C_+}  e^{\alpha w^2}\mathcal K^{\rm c}(z,w) \mathcal K^{\rm c}(w,z) \, d^2 w 
 = 2 e^{\alpha z^2}\mathcal K^{\rm c}(z,z).
$$
This can be checked component wise; see \cite[proof of Prop.~4.8]{FM11}.
\end{proof}

\subsection{Singular values}
Matrices $X$ from GinOE being real, the squared singular values are the eigenvalues of $X^T X$. With the meaning of GinOE extended to include rectangular $p \times N$ matrices, the random matrices $X^T X$ have an interpretation in multivariate statistics. Thus one interprets each column as a particular trait that is being measured from a population of size $N$ so that $X$ is regarded as a data matrix. Moreover, let the distribution of each of the traits be a standard Gaussian, which serves as a structureless base case. Up to a simple normalisation factor, $X^T X$ is the matrix of sample covariances between the traits. Because of its statistical interest, this class of random matrix attracted attention in one of the earliest works relating to random matrix theory \cite{Wi28}.

The result of \cite{Wi28} was to establish that the Jacobian for the change of variables $A = X^T X$ is proportional to $(\det A)^{(p-N-1)/2}$. Around a decade later, the Jacobian for the change of variables from a real symmetric matrix to its eigenvalues and eigenvectors was calculated \cite{Hs39}. This exhibited a factorised form, with Jacobian $\prod_{1 \le j < k \le N} |\lambda_k - \lambda_j|$ relating only to the eigenvalues. Denoting the squared singular values of a $p \times N$ rectangular GinOE matrix by $\{s_j\}_{j=1}^N$, it then follows that the corresponding joint PDF is proportional to
\begin{equation}\label{3.18}
    \prod_{j=1}^N s_j^{(p-N-1)/2}e^{-s_j/2} \prod_{1 \le j < k \le N} |s_k - s_j|.
\end{equation}

Introducing the global scaling $X^TX \mapsto {1 \over N} X^T X$, 
and with $p$ scaling with $N$ according to $p = \alpha N$, $\alpha > 1$, it is a well known result \cite{Ge80} that the smallest and largest squared singular values tend almost surely to the values $(\sqrt{\alpha \mp 1}+1)^2$. Hence the ratio of the largest to the smallest singular value (i.e.~the condition number $\kappa_N$ of $X$) tends to a constant. However with $p = N$ (or more generally $p=N+p_1$ with $p_1$ independent of $N$), while the largest squared singular value tends to $4$, the smallest now tends to zero at the rate $1/N^2$. Specifically in the square case $p=N$, it has been shown by Edelman \cite{Ed88} that this implies $\kappa_N/N$ has a limiting distribution with the heavy tailed PDF
\begin{equation}
{2 x + 4 \over x^3} e^{-2/x - 2/x^3}.
\end{equation}
In the rectangular case, for any $N \ge 2$, a result of \cite{CD05} gives the bound $$\langle \log \kappa_N \rangle < {N \over |p-N|+1} +2.258,$$  which apart from the last two digits of the constant, is the same as that known for rectangular GinUE matrices (recall \cite[\S 6.3]{BF22a}).

Shifting a global scaled square GinOE matrix $\tilde{X} = {1 \over \sqrt{N}} X$ by defining 
\begin{equation}\label{3.20}
Y = - z\mathbb I_N + \tilde{X}
\end{equation}
leads to a transition effect for the corresponding smallest singular value in the neighbourhood of $|z| =1$. Thus for $|z| > 1$ and independent of $N$, the smallest singular value is bounded away from $0$ as $N \to \infty$. The transition region $|z| < 1 + c/N^{1/2}$ is studied in \cite{CES20}, with a distinction between $z$ complex and $z \approx \pm 1$ quantified.

For $X$ a square GinOE matrix, we turn our attention to the distribution of $(\det X)^2$, which is well known in multivariate statistics \cite{Mu82}.

\begin{proposition}
For $X$ a GinOE matrix
\begin{equation}\label{3.18a}
(\det X)^2 \mathop{=}\limits^{\rm d} \prod_{j=1}^N \chi_j^2.
\end{equation}
\end{proposition}

\begin{proof}
Either of the strategies used to establish \cite[Eq.~(6.10)]{BF22a} can be used. Specifically, to make use of (\ref{3.18}) (with $p=N$), we use the fact that $(\det X)^2 = \prod_{l=1}^N s_l$ and hence that the Mellin transform of the distribution of $(\det X)^2$ is given by
$$
 {1 \over C_N} \int_0^\infty ds_1 \cdots \int_0^\infty ds_N \, \prod_{j=1}^N s_j^{s-3/2} e^{-s_j/2} \prod_{1 \le j < k \le N} |s_k - s_j|^2 =
\prod_{j=1}^{N}{ 2^{s -1}\Gamma(s+j/2) \over \Gamma(1+j/2)},
$$
where $C_N$ is the same multiple integral in the case $s=1$. The gamma function evaluation follows from the Laguerre weight Selberg integral as given in e.g.~\cite[Prop.~4.7.3]{Fo10}. It follows from this that $(\det X)^2 \mathop{=}\limits^{\rm d} \prod_{j=1}^N F_j$, where $F_j$ is independent with the Mellin transform of its PDF given by ${ 2^s \Gamma(s+j/2) \over \Gamma(1+j/2)}$. One verifies that thus $F_j$ is equal to $\chi_j^2$, as required.
\end{proof}

There is an interpretation of $(\det X)^2$ in integral geometry.  The idea is to regard the columns of $X$ as vectors in $\mathbb R^N$. Then $|\det  X|$ is equal to the volume of the parallelotope generated by these vectors. For $X$ a GinUE matrix, the parallelotope is random, and said to be Gaussian. Then (\ref{3.18a}) gives the distribution of the squared volume. For large $N$, (\ref{3.18a}) can be used to establish that
 the distribution of $\log (\det X)^2$ has to leading order mean equal to $-N$, variance $2 \log N$, and after recentring and rescaling satisfies a central limit theorem \cite{Mi71,Ma98}.

 Let $Y_k$ denote the shifted GinOE matrix (\ref{3.20}) restricted to the first $k$ columns. For $z$ real the positive definite matrix $W_k = Y_k^TY_k$ is an example of a non-central Wishart matrix, well known in mathematical statistics \cite{Mu82}. An exact evaluation of $\langle (\det W_k)^\alpha \rangle$ in terms of a generalised hypergeometric function of $k$ variables has been given in \cite{Co63}; see \cite[\S 4]{FZ18} for a discussion of the implications of this result in relation to the Lyapunov spectrum for products of shifted GinOE matrices.

 \begin{remark}
 3.~The squared singular values as specified by the PDF (\ref{3.18}) specify the classical Laguerre orthogonal ensemble and form a Pfaffian point process, see e.g.~\cite[Chapters.~3 and 6]{Fo10}. However there is no known analogous result for elliptic GinOE. Also, with regards to the squared singular values of the various extensions of GinOE to be considered below, in particular the the spherical model and truncated real orthogonal  matrices, the classical Jacobi orthogonal ensemble results, which is a Pfaffian point process. However, no such structure in known for the squared singular values of products of GinOE matrices.
 \end{remark}

\subsection{Eigenvectors}
The basic question of interest is, as for GinUE, the statistical properties of the scaled invariant overlaps of the left and right eigenvectors,
$\mathcal O_{ij} := \langle \ell_i, \ell_j \rangle 
\langle r_i, r_j \rangle.
$
In the diagonal case, the approach of Fyodorov \cite{Fy18} was to study the joint PDF
$$
\mathcal P^{r}(t,\lambda) := \Big \langle \sum_{j=1}^N \delta (  \tilde{\mathcal O}_{jj} - 1 - t) \delta(\lambda - \lambda_j) \Big \rangle, \qquad \tilde{\mathcal O}_{jj}:={\mathcal O}_{jj}/N,
$$
and where $\lambda$ is restricted to real values (thus the superscript "r").
This was shown to permit an exact evaluation generalising the GinOE eigenvalue density formula for
$\rho_{(1),N}^{\rm r}(\lambda)$, (\ref{13.222}) with $x=y$
\cite[Th.2.1 and Eq.~(2.5)]{Fy18}.

\begin{proposition} 
    We have  
    \begin{multline}
\mathcal P^{\rm r}(t,\lambda) =
{1 \over 2 \sqrt{2 \pi}}
e^{{\lambda^2 \over 2} {1 \over 1 + t}} {1 \over t (1 + t)}
\Big ( {t \over 1 + t} \Big )^{(N-1)/2} \\
\times \bigg (
{e^{-\lambda^2} \lambda^{2(N-1)} \over (N - 2)!}
+ {1 \over (N-2)!} \Gamma(N-1;\lambda^2)
\Big ( (N - 1) - \lambda^2 {t \over 1 + t} \Big ) \bigg ).
\end{multline}
\end{proposition}

\begin{proof} (Comments only) The first step is to introduce the Laplace transform $\int_0^\infty e^{-pt} \mathcal P^{\rm r}(t,\lambda) \, dt$. Next it was shown that this Laplace transform can be expressed in terms of a GinOE average involving ratios of characteristic polynomials
$$
\bigg \langle {\det (zI_N - G)(\bar{z} I_N - G^\dagger) \over ( \det (2p I_{N} + (z I_N - G)(\bar{z} I_N - G^\dagger))^{1/2}} \bigg \rangle.
$$
Supersymmetric integration techniques are then used to evaluate this average.
\end{proof}

We remark that an extension of this proposition  to the elliptic GinOE has been given in \cite[Th.~2.1]{FT21}.
From the definition $\rho_{(1),N}^{\rm r}(\lambda) = \int_0^\infty \mathcal P^{r}(t,\lambda) \, dt$, which is readily verified. Bulk and edge scaling limits follow \cite[Eqns.~(2.7) and (2.10)]{Fy18}.

\begin{corollary}
We have
\begin{equation}\label{R.A}
P^{\rm r, b}(s,x):= \lim_{N \to \infty} N \mathcal P^{\rm r}(Ns,\sqrt{N} x) =
{1 \over 2 \sqrt{2 \pi}}
{e^{-{(1-x^2) \over 2s}} \over s^2} (1 - x^2)\chi_{|x| < 1}
\end{equation}
and
\begin{multline}
P^{\rm r, e}(s,x):= \lim_{N \to \infty} \sqrt{N} \mathcal P^{\rm r}(\sqrt{N}s,\sqrt{N} + \delta) 
\\
 = {1 \over 2 \sqrt{2 \pi}} {1 \over \sigma^2} e^{- {1 \over 4 \sigma^2} + {\delta \over \sigma}} \bigg (
 {1 \over  \sqrt{2 \pi}} e^{-2 \delta^2} + \Big ( {1 \over \sigma} - 2 \delta \Big ) {\rm ercf}(\sqrt{2} \delta) \bigg ).
 \end{multline}
\end{corollary}

Here one has the sum rules
$$
\int_0^\infty P^{\rm r, b}(s,x) \, ds = {1 \over \sqrt{2 \pi}} \chi_{|x| < 1}, \qquad
\int_0^\infty P^{\rm r, b}(s,x) \, ds = \rho_{(1),\infty}^{\rm r, e}(x)
$$
reclaiming the global and edge scaled densities (\ref{11.xs}) and (\ref{11.SXYa}). The result (\ref{R.A}) is the real analogue of the result for GinUE \cite[Eq.~(6.22)]{BF22a}. Note that the tail is now proportional to $1/s^2$, so only the zeroth integer moment in $s$ (which gives the global density) is defined. In particular, the mean value with respect to $s$, which for GinUE is given by the result of Mehlig and Chalker \cite[Eq.~(6.23)]{BF22a}, formally diverges telling us that for GinOE, $\mathcal O_{ii}$ averaged over the full spectrum is no longer proportional to $N$. If instead $\mathcal O_{ii}$ is conditioned on a single eigenvalue away from the real away from the real axis, numerical experiments from \cite[Fig.~3]{CES22} indicate that $\mathcal O_{ii}$ is proportional to $N$. Hence eigenvalues on (or close to) the real axis are responsible for this breaking down in general.

\section{Further extensions to GinOE}\label{S2a}
\subsection{Common structures}\label{S3.1C}
Comparison between the development of the theory of the eigenvalue statistics for GinOE and elliptic GinOE shows a number of common structures:
\begin{itemize}
    \item[(S1)] The joint eigenvalue PDF has the functional form (\ref{3.1}) for a certain weight function $\omega(z)$, and a certain normalisation $C_N$; cf.~(\ref{3.3}). As a consequence the summed generalised partition function $Z_N[u,v]$ has the Pfaffian form (\ref{12.Ca}) involving the same weights and normalisation.
    \item[(S2)] The functional form for the joint eigenvalue PDF in the sector that all eigenvalues are real leads to a closed form expression for the probability that all eigenvalues are real.
    \item[(S3)] Associated with $Z_N[u,v]$ in the case $u=v=1$ is a particular skew inner product, while associated with the latter are a family of skew orthogonal polynomials. These skew orthogonal polynomials admit the matrix theoretic form (\ref{12.212f}), and their normalisations $r_{j-1}$ can be computed from $C_N$. Moreover, the odd degree skew orthogonal polynomials can be written in a derivative form involving the weight $(\omega(x))^{1/2}$.
    \item[(S4)] The general $m$-point correlation function between real-real, complex-complex and real-complex eigenvalues is of the form of a Pfaffian point process and as such are determined by  particular $2 \times 2$ correlation kernels. The latter in turn are in fact fully determined by their upper off diagonal entry, denoted $S_N(x,y)$. In the real-real case, due to the derivative formula for the odd degree skew orthogonal polynomials, and a further derivative formula for a certain linear combination of successive even degree skew orthogonal polynomials, the normalisations and the weight $(\omega(x))^{1/2}$, the summation can be written in a simplified form involving a rank one perturbation of the kernel known in the case of the GinUE and elliptic GinUE. Another form valid in both cases is (\ref{2.37a}), while its analogue in the complex-complex case (\ref{cE2a})
is similarly valid in both cases.
\item[(S5)] The structural form
of the large $N$ asymptotic formula (\ref{eks2}) for the variance of the number of real eigenvalues is valid, or more generally the formula (\ref{eks2.b}) for the variance of a linear statistic.
\item[(S6)] The global density of the complex eigenvalues is the same as found for GinUE (the circular law) and elliptic GinUE.
\end{itemize}
A structural property valid for GinOE but not elliptic GinOE is that of Proposition \ref{P2.4} for the skew orthogonal polynomials. Another is the determinant formula (\ref{12.212b}) for the generating function of the number of real eigenvalues, and the associated local central limit theorem Proposition \ref{P2.4b}.

As we further develop the theory of ensembles related to GinOE below, we will see that most of the properties listed above again hold true.
\subsection{Induced GinOE} \label{Section_iGinOE}
Let $G$ be an $n \times N$ ($n \ge N)$ random matrix with independent real standard Gaussian entries. Let $R$ be a Haar distributed real orthogonal matrix (for this notion see \cite{DF17}). In the case that $G$ is square, it is easy to check that $(G^T G)^{1/2} R$ has the same element distribution as $G$, and thus is an element of GinOE; see e.g.~\cite[Prop.~2.10]{BF22a} in the case that $G$ is an element of GinUE, and $R$ is a Haar complex unitary matrix. In the rectangular case, the element distribution of $\tilde{G}=(G^T G)^{1/2} R$ is the same as the element distribution of $G$ multiplied by the factor $(\det G^T G)^{(n-N)/2}$ \cite{FBKSZ12}; cf.~\cite[Prop.~2.11]{BF22a}. The essential fact required to establish this is that the factor is the Jacobian for the change of variables $A=(G^T G)^{1/2}$, which itself can be established in two steps: first change variables $G^\dagger G = B$, then change variables $B=A^2$. The Jacobian in both these steps can be read off the change of variables required in the theory of real Wishart matrices; see \cite[Ch.~2]{Mu82}.

Of relevance is the explicit form of the proportionality.

\begin{proposition}\label{P3.1}
Set $L:=n-N$.
The element distribution on $N \times N$ real random matrices $\tilde{G}$ specified by
$$
{1 \over C_{L,N}} (\det \tilde{G}^T \tilde{G})^{L/2} e^{-{\rm Tr} \, \tilde{G}^T \tilde{G}/ 2},
\quad
C_{L,N} = \pi^{N^2/2} 2^{N^2/2 + N L/2} \prod_{j=1}^N {\Gamma((j+L)/2) \over \Gamma(j/2)}
$$
is correctly normalised. 
\end{proposition}

\begin{proof}
  According to the QR decomposition, we can write $
  \tilde{G}=QR$, where $Q$ is an $N\times N$ matrix real orthogonal matrix, and $R$ is an $N \times N$ real upper triangular matrix with positive diagonal entries. With this change of variables, it is known that \cite[Th.~2.1.13 with $n=m=N$]{Mu82},
  $$
  (d\tilde{G}) = \prod_{j=1}^N r_{jj}^{N-j} (dR) (Q^T dQ),
  $$
  where $(Q^T dQ)$ is the Haar measure on real orthogonal matrices. Hence
  \begin{multline*}
       \int (\det \tilde{G}^\dagger \tilde{G})^{L/2} e^{-{1 \over 2} {\rm Tr} \, \tilde{G}^\dagger \tilde{G}} \, (d\tilde{G}) \\ =
 \Big ( \prod_{j=1}^N \int_0^\infty r^{L+j-1} e^{-{1 \over 2} t^2} \, dr \Big ) \Big ( \int_{-\infty}^\infty e^{-{1 \over 2} r^2} \, dr \Big )^{N(N-1)/2} \int (Q^T dQ).
  \end{multline*}
Computing the integrals over $t$, and using the known result for the volume of the orthogonal group \cite[Corollary 2.1.16]{Mu82}
\begin{equation}\label{Qv}
\int (Q^T dQ) = {2^N \pi^{N(N+1)/4} \over \prod_{l=1}^N \Gamma(l/2)},
\end{equation}
  we can identify the RHS with $C_{n,N}$.
\end{proof}

For random matrices $\{ \tilde{G} \}$ --- referred to as induced GinOE matrices --- the element distribution of Proposition \ref{P3.1} tells us that the eigenvalue PDF in the sector of $k$ real eigenvalues is again given by (\ref{3.1}), multiplied by ${C_{0,N} \over C_{L,N}} \prod_{l=1}^k |\lambda_l|^{L} \prod_{l=1}^{(N-k)/2} |z_l|^{L}$. 
With $Z_N^{\rm iGinOE}[u,v]$ denoting the analogue of $Z_N[u,v]$ as specified above (\ref{12.Ca}), we see that (\ref{12.Ca}) again holds true, but with the modification that there is an extra factor of $C_{0,N}/C_{L,N}$, 
and that $\alpha_{j,k}, \beta_{j,k}$ are replaced by $\alpha_{j,k}^{\rm i},\beta_{j,k}^{\rm i}$, defined as in Proposition \ref{P2.2}, but with extra factors of $(xy)^{L}$ and
$(x^2+y^2)^{L}$ in the integrands respectively. Most crucially, the formulas (\ref{12.212g}) for the skew orthogonal polynomials corresponding to $(\alpha_{j,k}^{\rm i}+\beta_{j,k}^{\rm i})|_{u=v=1}$ are again valid with $\langle {\rm Tr} \, G^2 \rangle = (2j + L)$ (we replace the index $n$ in (\ref{12.212g}) by $j$ to avoid conflict with the $n$ appearing in the definition of the induced GinOE, and specifically the parameter $L$). Moreover, with $C_j(z) = z^j$ they permit the structural form
\begin{equation}\label{2.46a}
p_{2j}^{\rm i}(z) = C_{2j}(z), \qquad
p_{2j+1}^{\rm i}(z) = - z^{-L} e^{z^2/2}{d \over dz} \Big (z^L e^{-z^2/2} C_{2j}(z) \Big );
\end{equation}
cf.~(\ref{RC}), and
the derivation of (\ref{12.212h}) implies
\begin{equation}\label{12.+}
r_{j-1}^{\rm i}= 2 \sqrt{2\pi} \Gamma(L+2j-1).
\end{equation}

Taking into consideration these modifications, we can check that the determinant formula (\ref{12.212b}) remains valid for $Z_N^{\rm i}(\zeta)$, except that $L$ needs to be added to the argument of each of the gamma functions. As a consequence, for the expected number of real eigenvalues, we obtain the formula
\begin{equation}\label{12.+1}
E_N^{\rm r, i} =
\sqrt{2 \over  \pi} \sum_{k=1}^{N/2}
{\Gamma(L+2k-3/2) \over \Gamma(L+2k-1)};
\end{equation}
cf.~the first equality in 
(\ref{12.212c}). From the determinant formula the validity of the analogue of the local central limit theorem Proposition \ref{P2.4b} can also be checked. Another consequence of (\ref{2.46a}) and (\ref{12.+}), substituted in the formula for $S_N^{\rm r}$ given in (\ref{12.Sr}), together with the validity of a derivative formula analogous to (\ref{Ae2}), is the simplified summation form
\begin{equation}\label{rttB}
S_{N}^{\rm r, i}(x,y) =
{e^{-(x^2+y^2)/2} \over \sqrt{2 \pi} }
\sum_{k=0}^{N-2} {(xy)^{k+
L} \over (k+L)!}  +
{y^{L+N-1} e^{- y^2/2} \over 2 \sqrt{2 \pi}  }
{ \Phi_{N-2}(x) \over (L+N-2)!},
\end{equation}
which reduces to (\ref{13.222}) 
when $L=0$.

In studying the large $N$ limit, the most interesting circumstance is to simultaneously set
$L=N \alpha$, ($\alpha > 0)$. It follows from (\ref{12.+1}) that to leading order the expected number of real eigenvalues is then $\sqrt{2N / \pi}( \sqrt{\alpha+1} - \sqrt{\alpha})$. Setting $x=y$ in (\ref{rttB}) gives the eigenvalue density. Further introducing the global scaling $x \mapsto \sqrt{N} x$ we can check that the final term goes to zero, while up to proportionality the summation can be identified with $K_{N-1}^{\rm iG}(x,x)$ as given by \cite[Eq.~(2.44)]{BF22a}. From the known global scaled limit of the latter \cite[Eq.~(2.45)]{BF22a} it follows
\begin{equation}\label{rttC}
\lim_{N \to \infty}
\rho_{(1),N}^{\rm r,i}(\sqrt{N}x) = {1 \over \sqrt{2\pi}} \Big ( \chi_{|x|<\sqrt{\alpha +1}} - \chi_{|x|<\sqrt{\alpha}} \Big ).
\end{equation}
Note that this can be used to provide an alternative derivation of the leading form of the expected number of real eigenvalues as stated above.
Bulk and edge scalings in this limit exhibit the same functional forms as for GinOE, given by (\ref{11.SXY+}) and (\ref{11.SXY}) respectively. The former of these implies that the analogue of (\ref{eks2}) remains valid. Also, using (\ref{cE1}), the analogue of (\ref{cE2a}) can be obtained for $S_N^{\rm c, i}(w,z)$, thus giving an expression in terms of the induced GinUE kernel $K_{N-1}^{\rm iG}(w,z)$, with in fact precisely the same prefactor.  A consequence is that the global scaling limit of the complex eigenvalues coincides with that of the induced GinUE \cite[Eq.~(2.45)]{BF22a},
\begin{equation}\label{rttC+}
\lim_{N \to \infty}
\rho_{(1),N}^{\rm c,i}(\sqrt{N}x) = {1 \over {\pi}} \Big ( \chi_{|z|<\sqrt{\alpha +1}} - \chi_{|z|<\sqrt{\alpha}} \Big ).
\end{equation}

\subsection{Spherical GinOE} \label{Section_spherical GinOE}
For $(X_1,X_2)$ a pair of $N \times N$ matrices, the generalised eigenvalues are defined as the solution of the equation $\det(X_1 - \lambda X_2) = 0$. We see that for $X_2$ invertible, the generalised eigenvalues then coincide with the eigenvalues of $X_2^{-1} X_1$. Here our interest is in the case that both $X_1, X_2$ are GinOE matrices --- then the ensemble specified by $\{ X_2^{-1} X_1 \}$ is referred to as the spherical GinOE.

Already in the pioneering paper on this topic, it was identified that statistical properties of real eigenvalues relate to integral geometry \cite{EKS94}. To see this, the matrix pair $(X_1,X_2)$ is viewed as two vectors in $\mathbb R^{N^2}$. The plane spanned by these two vectors intersects the sphere $S^{N^2 - 1}$ to define a great circle. Introducing $X = c (X_1 - \lambda X_2) $, with $c$ such that ${\rm Tr} (X^T X) = 1$, we see that real generalised eigenvalues correspond to intersections of the great circle with a unit vector corresponding to a singular matrix of unit Frobenius norm. 
In fact this viewpoint was used to deduce that the expected number of real eigenvalues, $E_N^{\rm r, s}$ say, is given by
\begin{equation}\label{3.2.1}
E_N^{\rm r, s} = {\sqrt{\pi} \Gamma((N+1)/2) \over \Gamma(N/2)} \mathop{\sim}\limits_{N \to \infty} \sqrt{\pi N \over 2}\Big ( 1 - {1 \over 4 N} + \cdots \Big );
\end{equation}
cf.~ (\ref{12.212c}), (\ref{12.212d}). With $E_N^{\rm r, s}$ determined, it was also noted in \cite{EKS94} that the density of real eigenvalues can be deduced. To see this, set $\lambda = \tan \theta$ so the  equation determining the generalised eigenvalues can be written $\det ( \cos \theta X_1 - \sin \theta X_2) = 0$. Since $X_1, X_2$, being GinOE matrices, are invariant with respect to rotations it follows that $(\cos \theta,\sin \theta) $ must be uniformly distributed on the unit circle and so $\lambda$ must have a Cauchy distribution. Consequently
\begin{equation}\label{3.2.2}
\rho_{(1),N}^{\rm sGinOE}(\lambda) = {1 \over \pi} {E_N^{\rm r, s} \over 1 + \lambda^2}.
\end{equation}

Subsequent to \cite{EKS94} it has been established that the eigenvalues of the spherical GinOE form a Pfaffian point process, with (\ref{3.2.1}) and
(\ref{3.2.2}) corollaries of this general structure \cite{FM11}. We require the normalisation 
\begin{equation}\label{3.12}
    C_{N}^{\nu \rm s}=  
    \Gamma(\nu-{N \over 2}+{1\over 2})^{N}
    \prod_{j=1}^N {1 \over \Gamma(j/2) \Gamma(\nu - N + j/2)} 
    \end{equation}
and weight
\begin{equation}\label{WsA}
\omega^{\nu \rm s}(z) = {2 \over \sqrt{\pi}}
    {\Gamma(\nu - {N \over 2} + 1) \over\Gamma(\nu - {N \over 2} + {1 \over 2})} 
{1 \over | 1 + z^2|^{2\nu - N+1}}
\int_{2y/(1+|z|^2)}^\infty {dt \over (1 + t^2)^{\nu - N/2 + 1}},
\end{equation}
where $z=x+iy$. With $z=x$ real, this reduces to 
 $\omega^{\nu \rm s}(x)=(1+x^2)^{ -2\nu+N-1 }$.

\begin{proposition}\label{P3.2}
    In the above notation and that of (\ref{3.1}), the joint eigenvalue PDF for $k$ real eigenvalues 
    $\{\lambda_l\}_{l=1,\dots,k} $
and the $(N-k)/2$ complex eigenvalues $\{z_j:= x_j + i y_j \}_{j=1,\dots,(N-k)/2}$
in the upper half plane for the spherical GinOE is 
\begin{equation}\label{3.1S}
C_{N}^{\nu \rm s}
{2^{(N-k)/2} \over k! ((N-k)/2)! } \prod_{s=1}^k (\omega^{\nu \rm s}(\lambda_s))^{1/2} \prod_{j=1}^{(N-k)/2} \omega^{\nu \rm s}(z_j) \Big | \Delta(\{\lambda_l\}_{l=1,\dots,k} \cup
\{ x_j \pm i y_j \}_{j=1,\dots,(N-k)/2}) \Big |,
\end{equation}
with $\nu = N$.
\end{proposition}

\begin{proof}
To derive (\ref{3.1S}), aspects of the derivation of the eigenvalue PDF for the spherical GinUE, and the GinOE, are required. In order, these are:

\begin{itemize}
\item[(i)] Determine the analogue of the joint element distribution \cite[Eq.~(2.47)]{BF22a}. Using the same method, one obtains that for $Y=A^{-1}B$ with $A,B$ GinOE matrices, the joint element distribution is
\begin{equation}\label{3.8}
A_{N,\nu} \det (\mathbb I_N + Y^T Y)^{-\nu} \Big |_{\nu = N} , \qquad
A_{N,\nu} = \pi^{-N^2/2}
\prod_{j=1}^{N} {\Gamma((2\nu - N +j)/2) \over \Gamma((2\nu - 2N + j)/2)}.
\end{equation}
This is an example of the matrix $t$-distribution \cite{Co54}. There is some advantage in keeping the variable $\nu$ general throughout the remainder of the calculation. In particular, knowledge of the corresponding functional forms becomes useful in the computation of the skew orthogonal polynomials to be undertaken subsequently.
\item[(ii)] Introduce the real Schur decomposition $Y = QR Q^T$, conditioned on $k$ real eigenvalues, with the notations of \S \ref{S2.1}. Integrate out the strictly upper triangular entries of $R$ using a modification of the technique used for spherical GinUE \cite[\S 2.5]{BF22a}. This results in the conditional PDF
\begin{multline}\label{3.1Sa}
C_{N}^{\nu \rm s} \bigg (
    {\Gamma(\nu - {N \over 2} + 1) \over \sqrt{\pi} \Gamma(\nu - {N \over 2} + {1 \over 2})} \bigg )^{(N-k)/2}\Big | \tilde{\Delta}(\{\lambda_l\}_{l=1,\dots,k} \cup
\{ x_j \pm i y_j \}_{j=1,\dots,(N-k)/2}) \Big | \\
\times \prod_{s=1}^k {1 \over (1 + \lambda_s^2)^{\nu-(N-1)/2}}
\prod_{s=1}^{(N-k)/2} {2|b_s - c_s| \over \det(\mathbb I_2 + R_s R_s^T)^{\nu - N/2+1}},
\end{multline}
where $C_{N}^{\nu \rm s}$ is given by (\ref{3.12}).
The quantity $\tilde{\Delta}$ is defined as for $\Delta$ in (\ref{3.1}) except that the difference factor between each pair $x_j \pm i y_j$ is omitted. Each $R_s$ is a $2 \times 2$ matrix on the part of the diagonal of the block diagonal matrix $R$ corresponding to the complex eigenvalues.
\item[(iii)] The final step is to change variables from $\{b_j,c_j\}$ to $\{y_j,\delta_j\}$ and to integrate over $\delta_j$ as for GinOE;
recall \S \ref{S2.1}. Now setting $\nu = N$, after some minor manipulation (\ref{3.1S}) results.
\end{itemize}
\end{proof}

Let $\alpha_{j,k}^{\nu\rm s}[u], \beta_{j,k}^{\nu \rm s}[v]$ be defined as for 
$\alpha_{j,k}^{\rm g}[u], \beta_{j,k}^{\rm g}[v]$
in Proposition \ref{P2.2} but with each $\omega^{\rm g}$ replaced by $\omega^{\nu \rm s}$ as defined in (\ref{WsA}).
The analogue of (\ref{12.Ca}) can then be checked to be
\begin{equation}\label{12.CaS}
Z_N^{\nu {\rm s}}[u,v] = 
    C_N^{\nu \rm s}  \,
 {\rm Pf} [  \alpha_{j,k}^{\nu \rm s} + \beta_{j,k}^{\nu \rm s}  ]_{j,k=1,\dots,N}.
\end{equation}
Further progress relies on the polynomials $\{ p_{l-1}(x) \}$ having the skew orthogonality property (\ref{12.212e})
with respect to the skew inner product corresponding to $\gamma_{j,k}^{\nu \rm s} := \alpha_{j,k}^{\nu \rm s} |_{u=1} + \beta_{j,k}^{\nu \rm s} |_{v=1}$. We see from (\ref{WsA}) that with the replacement $\nu \mapsto (\nu+N)/2$ the weights in the definitions of $\alpha_{j,k}^{\nu \rm s}, \beta_{j,k}^{\nu \rm s}$ are independent of $N$. On the other hand there is a known construction of an $N \times N$ matrix with element distribution (\ref{3.8}) and this value of $\nu$.
 Thus form the product $W^{-1/2}B$ with $B$ an $N \times N$ GOE matrix and $W$ the Wishart matrix $W = \tilde{G}^T \tilde{G}$ with $\tilde{G}$ a rectangular $\nu \times N$ GinOE matrix --- generally the exponent is $1/2$ of the sum of the row and column size in $\tilde{G}$;
 see e.g.~\cite[Exercises 3.6 q.3]{Fo10}.
 It is immediate how to modify the construction to a  $2n \times 2n$ matrix with the corresponding weight similarly independent of $n$. This can then be used in the formulas of Proposition \ref{P2.4} to determine the skew orthogonal polynomials, first derived using evaluations of the averages in (\ref{12.212f}) deduced from symmetric function theory \cite{Fo13a,Fi12}.

\begin{proposition}\label{P3.3}
    Consider the ensemble specified by $N \times N$ matrices with the element distribution (\ref{3.8}) modified so that $\nu \mapsto (\nu+N)/2$, and the corresponding skew inner product implied by $\gamma^{\nu \rm s}_{j,k}$.
    We have
\begin{align*}
&p_{2n}^{\nu \rm s}(z)  = z^{2n},  \\
&p_{2n+1}^{\nu \rm s}(z)  = z^{2n+1}-{2 n \over \nu - 2n - 1}
z^{2n-1}=  - { (1+z^2)^{(\nu+1)/2} \over \nu - 2n- 1} 
{d \over dz} \Big ( (1 + z^2)^{-(\nu-1)/2} z^{2n} \Big ). 
\end{align*}
Furthermore  
$$
r_{n}^{\nu \rm s} = {{\pi}2^{2-\nu }  \Gamma(2n + 1) \Gamma( \nu - 2n -1) \over   \Gamma((\nu+1)/2)^2}.
$$
\end{proposition}

\begin{proof} From  the theory of the above discussion, the matrix $G$ in 
(\ref{12.212f}) can be constructed as the matrix product $W^{-1/2} B$ with $B$ a $2n \times 2n$ GinOE matrix and $W$ a real Wishart matrix constructed from a rectangular $\nu  \times 2n $ GinOE matrix. Thus this gives a $2n \times 2n$ random matrix with the same eigenvalue weight (\ref{WsA}) 
made independent of $N$ by $\nu \mapsto (\nu + N)/2$. The matrix product is isotropic and so in particular has the distribution of each element unchanged by negation, and moreover the joint first moment of distinct pairs of elements are uncorrelated. As a result the skew orthogonal polynomials are again given by the simple formula  (\ref{12.212g}). For the quantity $\langle {\rm Tr} \, G^2 \rangle$ therein, we have
$$
\langle {\rm Tr} \, G^2 \rangle = \langle {\rm Tr} \,B^2 \rangle \langle {\rm Tr} \,W^{-1} \rangle = 2n \Big ( {1 \over \nu - 2n - 1} \Big ).
$$
Here we have used the elementary result that for $B$ a $2n \times 2n$ GinOE matrix,
$\langle {\rm Tr} \,B^2 \rangle  = 2n$, and the fact that for a Wishart matrix $W$ constructed from a rectangular $\nu  \times 2n $ GinOE matrix, $\langle {\rm Tr} \,W^{-1} \rangle = 1/(\nu - 2n - 1)$; see \cite[Appendix A]{Ma12}, \cite{vR88}. 

The formula for the normalisation $r_n^{\nu \rm s}$ again relies on the weights for $\nu \mapsto (\nu + N)/2$ being independent of $N$, together with the facts that for $u=v=1$, 
$Z_N^{\nu {\rm sGinOE}}[u,v]$ in
(\ref{12.CaS}) has the value unity by its construction, while by  skew orthogonality the Pfaffian equals $\prod_{l=0}^{N/2-1}r_l^{\nu \rm s}$. 
An important point for the calculation is that the product in (\ref{3.12}) can be written in the form
$$
\prod_{l=1}^N {1 \over \Gamma(l/2) \Gamma(\nu - N + l/2)} \bigg |_{\nu \mapsto (\nu + N)/2} =
\prod_{l=1}^{N/2} {1 \over \pi 2^{2 - \nu } \Gamma(2l-1)  \Gamma(\nu -2l + 1)}
$$
so that the only dependence on $N$ is in the terminal.
\end{proof}

Using the derivative formula for $p_{2n+1}^{\nu \rm s}(z)$ from Proposition \ref{P3.3} in the definition above (\ref{12.CaS}) of $\alpha_{j,k}^{\nu \rm s}$ 
shows
$$
\alpha_{2j-1,2k}^{\nu \rm s} \Big |_{\nu \mapsto (\nu + N)/2 \atop u = v = 1} =
{2 \over \nu - 2k +  1} {\Gamma(j+k-3/2) \Gamma(3/2-j-k+\nu) \over \Gamma(\nu)}=:
{\tilde{\alpha}_{j,k}^{\nu \rm s} \over \nu -2k +1}
$$
With $\tilde{r}_{j-1}^{\nu \rm s} := (\nu - 2j + 1){r}_{j-1}^{\nu \rm s} $
It follows that the analogue of (\ref{12.212b}) is the formula
\begin{equation}\label{12.CaT}
Z_N^{\nu \rm s}(\zeta) \Big |_{\nu \mapsto (\nu + N)/2} =
\det \Big [ \delta_{j,k} + {\zeta - 1 \over (\tilde{r}_{j-1}^{\nu s} \tilde{r}_{k-1}^{\nu s})^{1/2} }
\tilde{\alpha}_{j,k}^{\nu \rm s}
\Big ]_{j,k=1,\dots,N/2}.
\end{equation}
This can be used to deduce the analogue of (\ref{eks2}) and the local central limit theorem Proposition \ref{P2.4b} (the case $\nu = N$ of the local central limit theorem is known from the earlier study \cite[Prop.~3.5]{FM11}, based instead on (\ref{3'a}) below).
The derivative formula for $p_{2n+1}^{\nu \rm s}(z)$, supplemented by the further derivative relation 
\begin{multline*}
{2(k+1) \over (\nu - 2k - 3) r_{k+1}^{\nu \rm s}} p_{2k+2}(z) - {1 \over r_{k}^{\nu \rm s}}p_{2k}(z)
\\ =
- {\Gamma((\nu+1)/2))^2 \over \pi 2^{2-\nu} \Gamma(2k+2) \Gamma(\nu - 2k - 1)}
(1+z^2)^{(\nu+1)/2} {d \over dz} {z^{2k+1} \over (1+z^2)^{(\nu - 1)/2}},
\end{multline*}
gives for the analogue of the summation identity in the first line of (\ref{13.222}).

\begin{proposition}
We have
\begin{multline}\label{12.CaT1}
S_N^{{\rm r},\nu \rm s}(x,y) \Big |_{\nu \mapsto (\nu+N)/2} =
{\Gamma((\nu+1)/2)^2 \over \pi 2^{1 - \nu}} (1+y^2)^{-(\nu+1)/2}\\ \times \bigg (  (1+x^2)^{-(\nu-1)/2}
\sum_{k=0}^{N-2} {(xy)^k \over k! ( \nu - k - 1)!} 
+ {1 \over 2 \Gamma(N-1) \Gamma(\nu - N + 1)}
y^{N-1}  \Phi_{N-2}(x) \bigg ).
\end{multline}
In the case $\nu = N$ this further simplifies to give \cite{FM11}
\begin{equation}\label{12.CaT2}
S_N^{{\rm r},\nu \rm s}(x,y) \Big |_{\nu = N} =
{\Gamma((N+1)/2)^2 \over \pi 2^{1 - N}\Gamma(N)} 
(1+x^2)^{-(N-1)/2} (1+y^2)^{-(N+1)/2} 
 (1 + xy)^{N-1} .   \end{equation}

\end{proposition}

For $\nu = N + p$ for any fixed $p>0$ the bulk scaling limit involves the scalings $x \mapsto X/\sqrt{N}$. With this assumed, the result (\ref{11.SXY+}) first derived for GinOE itself, is reclaimed. We note there is no edge for spherical GinOE. In (\ref{12.CaT1}), setting $\nu = N \alpha_2$, $\alpha_2 \ge 1$ as is consistent with the notation \cite[below Eq.~(2.52)]{BF22a}, setting $x=y$ and taking the large $N$ limit gives for the global density of real eigenvalues
\cite[Th.~4.2.14]{Fi12}
\begin{equation}\label{12.CaT3}
\lim_{N \to \infty} {1 \over \sqrt{N}}
\rho_{(1),N}^{{\rm r},\nu {\rm s}}(x) = \sqrt{{\alpha_2 \over 2\pi}}
{1 \over 1 + x^2} \chi_{|x| < \sqrt{1/(\alpha_2 -1)}}.
\end{equation}

In relation to the complex eigenvalues, the key formula is
 (\ref{cE1}), further simplified to the form  analogous to  (\ref{cE2a}) with this involving the (generalised) GinUE spherical ensemble kernel. 
 As is consistent with (\ref{12.CaT3}) we first suppose that in (\ref{3.8}) the replacement $\nu \mapsto (N+\nu)/2$ has been made, and then set $\nu = \alpha_2 N$, $\alpha_2 > 1$. The corresponding GinUE spherical ensemble kernel is then given by \cite[Eq.~(2.54) with $M=N, n = \nu$]{BF22a}. From this one finds the same functional form as obtained in the GinUE case 
 \cite[Eq.~(2.55)]{BF22a}
\begin{equation}\label{12.CaT3b}
\lim_{N \to \infty} {1 \over {N}}
\rho_{(1),N}^{{\rm c},\nu {\rm s}}(x) = {\alpha_2 \over \pi}
{1 \over (1 + |z|^2)^2} \chi_{|z| < \sqrt{1/(\alpha_2 -1)}};
\end{equation}
see  \cite[Th.~4.2.14]{Fi12}.
 One observes that the global real density (\ref{12.CaT3}) is proportional to the square root of the global complex density (\ref{12.CaT3b}) continued to the real line. The proportionality constant is $1/\sqrt{2}$. This inter-relation (apart from the value of the proportionality) has been predicted by Tarnowski \cite{Ta22} to hold for a wide class of asymmetric real random matrices. For example, it holds for the GinOE itself (compare (\ref{11.xs}) and (\ref{cE3+})), and the induced GinOE (compare (\ref{rttC}) and (\ref{rttC+})).

\begin{remark} $ $ \\
1.~Setting $x=y$ in (\ref{12.CaT2}) reclaims (\ref{3.2.2}). Also, since the bulk GinUE correlations are the limit of bulk scaling, the analogue of the asymptotic variance formula (\ref{eks2}) holds \cite[stated below (14)]{FM11}.\\
2.~In \cite{FM11} an approach to studying the eigenvalue distribution for the spherical GinOE (i.e.~the case $\nu = N$ of the above) making use of a stereographic projection to the sphere was detailed. In this approach it was found that the corresponding skew orthogonal polynomials have the skew orthogonality property with respect to the analogues of $\alpha_{j,k}^{\rm s}$ and $\beta_{j,k}^{\rm s}$ individually. Consequently, the generating function for the probability of $k$ real eigenvalues, $Z_N^{\rm s}$, was shown to be given in the product form
\begin{multline}
Z_N^{\rm s}(\zeta) = \frac{(-1)^{(N/2)(N/2-1)/2}}{2^{N(N-1)/2}} \\ \times \Gamma((N+1)/2)^{N/2}\Gamma(N/2+1)^{N/2} \prod_{s=1}^{N}\frac{1}{\Gamma(s/2)^2}\prod_{l=0}^{N/2-1}(\zeta \alpha_{l} + \beta_{l}),
\end{multline}
where
\begin{eqnarray}\label{3'a}
\nonumber \alpha_{l}&=&\frac{2\pi}{N-1-4l}\frac{\Gamma((N+1)/2)}{\Gamma(N/2+1)},\\
\beta_{l}&=&\frac{2\sqrt{\pi}}{N-1-4l}\left( 2^N\frac{\Gamma(2l+1)\Gamma(N-2l)}{\Gamma(N+1)}-\sqrt{\pi}\frac{\Gamma((N+1)/2)}{\Gamma(N/2+1)}\right);
\end{eqnarray}
cf.~the determinant formula (\ref{12.CaT}) with $\nu=N$. \\
3.~Again in the case $\nu=N$, the matrix element $S_N^{\rm c, \rm s} $ as specified by (\ref{cE1}) has the explicit form \cite[Prop.~4.3]{FM11}
\begin{align*}
S_N^{\rm c, \rm s}(w,z)&=\frac{N(N-1)}{2^{N+1}\pi r_w r_z}\left[\int_{\frac{r_w^{-1}-r_w}{2}}^{\infty}\frac{dt}{\left(1+t^2\right)^{N/2+1}}\right]^{1/2}\left[\int_{\frac{r_z^{-1}-r_z}{2}}^{\infty}\frac{dt}{\left(1+t^2\right)^{N/2+1}}\right]^{1/2}
\\
 &\quad \times \left( \frac{e^{i(\theta_z-\theta_w)/2}}{(r_wr_z)^{1/2}}+\frac{e^{-i(\theta_z-\theta_w)/2}}{(r_wr_z)^{-1/2}} \right)^{N-2} \left( \frac{e^{i(\theta_z-\theta_w)/2}}{(r_wr_z)^{1/2}}-\frac{e^{-i(\theta_z-\theta_w)/2}}{(r_wr_z)^{-1/2}} \right),
\end{align*}
where $w,z:=r_we^{i\theta_w},r_ze^{i\theta_z}$. Used here are  coordinates inside the unit disk, which result from the conformal mapping of the half plane
$Z = i(1-z)/(1+z)$. Setting $w=z=re^{i\theta}$ in this gives for the transformed complex density
\cite[Eq.~(17)]{FM11}
$$
\rho_{(1),N}^{\rm c, s}(z) =
{N(N-1) \over 2^{N+1} \pi r^2} \Big ( {1 \over r} + r \Big )^{N-2} \Big ( {1 \over r} - r \Big ) \int_{r^{-1} - r \over 2}^\infty
{dt \over (1 + t^2)^{N/2+1}}.
$$
We calculate from this the limiting form
\begin{equation}\label{lak}
\lim_{N \to \infty} {1 \over N} \rho_{(1),N}^{\rm c, s}(z) =
{1 \over \pi} {1 \over (1 + r^2)^2},
\end{equation}
which is in fact invariant upon the inverse of the applied conformal mapping and so applies too in the original coordinates. 
Note that this is consistent with (\ref{12.CaT3b}) in the case $\alpha_2=1$.

\end{remark}

\subsection{Truncations of Haar real orthogonal matrices} \label{Section_trunc GinOE}
Starting with an $(n + N) \times (n+N)$ Haar distributed real orthogonal matrix (see \cite{DF17} for a discussion of the origin of this class of random matrices), we consider an $N \times N$ sub-block $A_N$ say  and seek the corresponding eigenvalue PDF \cite{KSZ09,Ma11,Fi12}. As for truncations of Haar distributed unitary matrices, the eigenvalues must be contained inside the unit disk $|z| < 1$.

\begin{proposition}
With $z = x + i y$, introduce the weights
\begin{equation}\label{341}
\omega^{\rm t}(z) = \begin{cases}
{n (n - 1) \over 2 \pi} |1 - z^2|^{n-2}
\int_{2|y|/|1 - z^2|}^1 (1-u^2)^{(n-3)/2} \, du, & n \ne 1 \\
 {1 \over \pi} |1 - z^2|^{-1}, & n = 1.
 \end{cases}
 \end{equation}
 After denoting (\ref{Qv}) as ${\rm vol}(O(N))$, define the normalisation
 $$
 C_{N,n}^{\rm t} = {{\rm vol}(O(n)) {\rm vol}(O(N)) \over
 {\rm vol}(O(n+N))} \Big ( {(2 \pi)^n \over n!} \Big )^{N/2}=
 \Big ( {2^n \over n!} \Big )^{N/2} \prod_{j=1}^N {\Gamma((n+j)/2) \over \Gamma(j/2)}.
 $$
 The joint eigenvalue PDF for $k$ real eigenvalues 
    $\{\lambda_l\}_{l=1,\dots,k} $
and the $(N-k)/2$ complex eigenvalues $\{z_j:= x_j + i y_j \}_{j=1,\dots,(N-k)/2}$
in the upper half plane for the an $N \times N$ truncation of a Haar distributed $O(n+N)$
matrix is then equal to
\begin{multline}\label{3.1St}
C_{N,n}^{\rm t}{2^{(N-k)/2} 
    \over k! ((N-k)/2)! }  \prod_{s=1}^k (\omega^{\rm t}(\lambda_s))^{1/2}
    \prod_{s=1}^{(N-k)/2}
    \omega^{\rm t}(z_s) \\
    \times
  \Big |  \Delta(\{\lambda_l\}_{l=1,\dots,k} \cup
\{ x_j \pm i y_j \}_{j=1,\dots,(N-k)/2}) \Big |. 
\end{multline}

\end{proposition}

\begin{proof}
Following the working given to derive the analogous result for the GinUE in \cite[\S 2.6]{BF22a}, which is based on \cite[Appendix B]{ABKN13}, the calculation can be carried out according to a number of specific steps.
\begin{itemize}
    \item[(i)] Make use of a matrix integral form (normalisation $(2 \pi)^{-N(N+1)/2}$ of the matrix delta function constraint for the block decomposition of the $(n+N) \times (n+N)$ orthogonal matrix (normalisation
    ${\rm vol}(O(n)) /
 {\rm vol}(O(n+N))$)
    to deduce that the element PDF of an $N \times N$ sub-block $A_N$ is up to these normalisations equal to \cite[Eq.~(2.58)]{BF22a}, now with the integral over the space of $N \times N$ real symmetric matrices $\{ H_N\}$. 
    \item[(ii)] Introduce the real Schur decomposition $A_N =
     QR Q^T$ conditioned on the number of real eigenvalues, and integrate over the off (block diagonal) entries of $R$ therein with the columns as variables. This requires carrying out the analogue of the steps used to deduce \cite[Eq.~(2.58)]{BF22a}. However there is the complication that in the columns corresponding to the complex eigenvalues, each column has a block structure. From this viewpoint, we then need to proceed analogous to step (ii) of the proof of Proposition \ref{P3.2}, the details of which are given in \cite{FM11}. With $n \ne 1$ (this case needs to be considered separately) the conditional eigenvalue PDF
     \begin{multline}\label{3.1Sat}
C_{N,n}^{\nu \rm t} \Big (
    {n(n-1) \over 2 \pi} \Big )^{N/2}\Big | \tilde{\Delta}(\{\lambda_l\}_{l=1,\dots,k} \cup
\{ x_j \pm i y_j \}_{j=1,\dots,(N-k)/2}) \Big | \\
\times \prod_{s=1}^k \Big ( {\sqrt{\pi} \Gamma((n-1)/2) \over \Gamma(n/2)} |1 - \lambda_s^2|^{n-2} \Big )^{1/2}
\prod_{s=1}^{(N-k)/2} 2|b_s - c_s|  \det(\mathbb I_2 - R_s R_s^T)^{(n-3)/2}
\end{multline}
results.
     \\
     \item[(iii)] As for the derivations of the eigenvalue PDF detailed for the GinOE and spherical GinOE, the final step is to change variables from $\{b_j,c_j\}$ to $\{y_j,\delta_j\}$ and to integrate over $\delta_j$.
\end{itemize}
\end{proof}

Next we let $\alpha_{j,k}^{\rm t}[u], \beta_{j,k}^{\rm t}[v]$ be defined as for 
$\alpha_{j,k}^{\rm g}[u], \beta_{j,k}^{\rm g}[v]$
in Proposition \ref{P2.2}, but now with weight $\omega^{\rm t}$ as defined in (\ref{341}).
For the analogue of (\ref{12.Ca}) we then have
\begin{equation}\label{12.CaSt}
Z_N^{ {\rm t}}[u,v] = 
    C_{N,n}^{ \rm t}  \,
 {\rm Pf} [  \alpha_{j,k}^{\rm t} + \beta_{j,k}^{ \rm t}  ]_{j,k=1,\dots,N}.
\end{equation}
We now come to the skew orthogonal polynomials associated with $\gamma_{j,k}^{ \rm t} := \alpha_{j,k}^{ \rm t} |_{u=1} + \beta_{j,k}^{ \rm t} |_{v=1}$.
 For this the formulas of Proposition \ref{P2.4} can be used.
 
 \begin{proposition}\label{P3.3t}
    Consider the ensemble specified by $N \times N$ random matrices obtained by deleting $n$ rows and columns of an $(n+N) \times (n+N)$ Haar distributed orthogonal matrix. For the skew orthogonal polynomials associated with $\gamma_{j,k}^{\rm t}$
    we have
$$
p_{2j}^{\rm t}(z)  = z^{2j}, \quad
p_{2j+1}^{\rm t}(z)  = z^{2j+1}-{2 j \over n + 2j}
z^{2j-1}=  - { (1-z^2)^{-n/2 + 1} \over n + 2j } 
{d \over dz} \Big ( (1 - z^2)^{n/2} z^{2j} \Big ). 
$$
Furthermore  
$
r_{j}^{\rm t} = {n! (2j)! \over (n+2j)!}.
$
\end{proposition}

\begin{proof} The matrix $G$ in 
(\ref{12.212f}) (with $n$ replaced by $j$ to avoid a conflict in notation) is constructed as the $2j \times 2j$ sub-block of an $(n+2j) \times (n+2j)$ Haar distributed orthogonal matrix.
The requirements of the applicability of the formulas (\ref{12.212g}), namely that 
  the distribution of each element is unchanged by negation, and that the joint first moment of distinct pairs of elements are uncorrelated is therefore valid. The square of an element of $G$ is then equal in distribution to $a_1^2/(a_1^2+\cdots+a_{n+2j}^2)$, where each $a_j$ is a standard Gaussian, and is thus a beta B$[1/2,(n+2j-1)/2]$ random variable. Its mean is $1/(n+2j)$ and so ${\rm Tr} \, \langle G^2 \rangle = 2j/(n+2j)$,
  which so specifies $p_{2j+1}$. This reasoning is a special case of a calculation first given in \cite{FIK20}.
   The formula for the normalisation $r_j^{\rm t}$  follows from setting $u=v=1$ in (\ref{12.CaSt}), and noting that the LHS then has the interpretation of a sum over probabilities of eigenvalues being real and so is unity, while by the skew orthogonality property the RHS equals $ C_{N,n}^{ \rm t} \prod_{l=0}^{N/2 -1} r_l^{\rm t}$.
\end{proof}

Next we use the derivative formula for $p_{2j+1}^{\rm i}(z)$ of Proposition \ref{P3.3t}, and the definition of $\alpha_{j,k}^{\rm t}$ in Proposition \ref{P2.2} with $\omega^{\rm g}$ replaced by $\omega^{\rm t}$ to compute
\begin{equation}\label{3.15x}
\alpha_{2j-1,2k}^{\rm t} \Big |_{u=v=1} = {1 \over (n + 2j) \sqrt{\pi}} {\Gamma((n+1)/2) \Gamma(n+1) \over \Gamma(n/2)}
{\Gamma(-3/2+j+k) \over \Gamma(-3/2+j+k+n)}=: {\tilde{\alpha}_{j,k}^{\rm t} \over n + 2j}.
\end{equation}
Hence for the generating function for the number of real eigenvalues we obtain a formula structurally identical to (\ref{12.CaT}), but with $\tilde{\alpha}_{j,k}^{\nu \rm s}$ replaced by $\tilde{\alpha}_{j,k}^{\rm t}$ and $\tilde{r}_{j-1}^{\nu \rm s}$ replaced by $(n+2j)r_{j-1}^{\rm t}$.
This in turn allows
the analogues of (\ref{eks2}) and  the local central limit theorem Proposition \ref{P2.4b}
to be deduced.

In relation to the real-real correlation kernel $S_N^{\rm r}$ from Proposition \ref{P2.5}, we have seen that supplementing 
the derivative formula for $p_{2j+1}^{\rm t}(z)$ by a further derivative relation involving the normalisations of the skew orthogonal polynomials, allows for a simplification \cite{Fo10a}.

\begin{proposition}
We have
\begin{multline}\label{12.CaT1t}
S_N^{\rm r, t}(x,y)  =
{1 \over \pi} {\Gamma(1/2) \Gamma((n+1)/2) \over \Gamma(n/2)}
(1-x^2)^{n/2-1}(1-y^2)^{n/2}
\sum_{j=0}^{N-2} {(n+j-1)! \over (n-1)! j!} (xy)^j \\
- {1 \over r_{N/2 - 1}} \Phi_{N-2}(y) \omega^{\rm t}(x) x^{N-1}.
\end{multline}
\end{proposition}

\begin{proof}
The required additional derivative formula is 
$$
{2(j+1) \over (n+2j+2) r_{j+1}^{\rm t}} p_{2j+2}(z) - {1 \over r_{j}^{\rm t}}p_{2j}(z)
 =
- {(n+2j)! \over n! (2j+1)!}
(1-z^2)^{(-n/2+1)} {d \over dz}  \Big ( (1-z^2)^{n/2} z^{2j+1} \Big ).
$$
\end{proof}

Introducing the incomplete beta function $I_x(a,b)$ and beta integral $B(x,y)$ in standard notation, one observes that (\ref{12.CaT1t}) admits an expression in terms of these functions. Specifically, with $x=y$ as corresponds to the density \cite{KSZ09,Ma11,Fi12},
\begin{multline}\label{3.25}
S_N^{\rm r, t}(x,x) \\ =
{1 \over B(n/2,1/2)}
{I_{1-x^2}(n+1,N-1) \over 1 - x^2} + {(1-x^2)^{(n-2)/2} |x|^{N-1} \over B(N/2,n/2)}
I_{x^2}((N-1)/2,(n+2)/2).
\end{multline}
From this form the scaled large $N,n$ asymptotics with $\alpha = N/(N+n)$ can be computed as \cite{KSZ09}
\begin{equation}\label{3.26B}
\lim_{N \to \infty} {1 \over \sqrt{N}} \rho_{(1)}^{\rm r, t}(x) = \sqrt{1-\alpha \over 2 \pi \alpha} {1 \over 1 - x^2} \chi_{-\sqrt{\alpha}< x < \sqrt{\alpha}}.
 \end{equation}
A corollary of this, obtained by integrating the RHS, is the asymptotic formula for the number of real eigenvalues
\begin{equation}\label{3.26}
    E_N^{\rm r, t} \Big |_{\alpha = N/(N+n)} \sim \sqrt{2N(1-\alpha) \over \pi \alpha} 
    {\rm Artanh}\sqrt{\alpha}.
 \end{equation} 
 The corresponding bulk and edge scaling of $S_N^{\rm r, t}(x,y)$ in this setting reclaims the GinOE results (\ref{11.SXY+}) and (\ref{11.SXY}). Hence we know that the analogue of the asymptotic formula (\ref{eks2}) for the variance of the number of real eigenvalues will hold true.

For finite $N$ the formula (\ref{cE1}) upon a simplification analogous to  (\ref{cE2a}) provides a simple to analyse expression for the complex density. From this one can calculate \cite{KSZ09}
\begin{equation}\label{3.26A}
\lim_{N \to \infty} {1 \over N} \rho_{(1),N}^{\rm c,t}
\Big |_{\alpha = N/(N+n)}=  {1 - \alpha \over \pi \alpha} 
{1 \over  (1 - |z|^2)^2}
\chi_{|z|< \sqrt{\alpha}}.
\end{equation}
Here one observes that the global real density (\ref{3.26B}) is proportional to the square root of the global complex density (\ref{3.26A}) continued to the real line, with proportionality $1/\sqrt{2}$ in keeping with  \cite{Ta22}.  We note too that the functional form (\ref{3.26A}) is identical to that given in \cite[Eq.~(2.61)]{BF22a} for the eigenvalue density of truncated unitary random matrices
in the same global scaling limit, in accordance with common feature (S6) of subsection \ref{S3.1C}.

With $n$ fixed, and thus $\alpha = 1$, the asymptotic formula (\ref{3.26}) breaks down. In keeping with this is that for $n$ fixed the $N \to \infty$ limit of (\ref{12.CaT1t}) is well defined without any scaling of $x$ and $y$. Thus
\begin{equation}\label{45}
\lim_{N \to \infty}S^{\rm r, t}(x,y) = {1 \over \pi}
{\Gamma(1/2) \Gamma((n+1)/2) \over \Gamma(n/2)} {(1 - x^2)^{n/2 -1}(1 - y^2)^{n/2} \over (1-xy)^n}.
\end{equation}
Also, there is an edge scaling distinct from (\ref{11.SXY}). Thus from (\ref{3.25}) we have \cite{KSZ09}
\begin{equation}
\lim_{N \to \infty} {1 \over N}
\rho_{(1)}^{\rm r,t}\Big ( 1 - {x \over N} \Big ) = \tilde{\rho}_{(1)}^{\rm e,t}(x),
\end{equation}
where, with $\bar{\gamma}(n,x) :=
(x^n \Gamma(n))^{-1} \gamma(n,x)$,
$$
{\rho}_{(1)}^{\rm e,t}(x) =
{x^{n/2-1} e^{-x} \over 2 \Gamma(n/2)} \Big (1 - x^{n/2 + 1}
\bar{\gamma}(n/2+1,x) + {(2x)^n \over B(n/2,1/2)} \bar{\gamma}(n+1,2x) \Big ).
$$
This exhibits the large $x$ tail behaviour $1/(B(n/2,1/2)x)$ which is in keeping with the large $N$ form of the expected number of real eigenvalues for $n$ fixed
\cite{KSZ09}
\begin{equation}\label{3.26a}
    E_N^{\rm r, t} \sim 
    { \log N \over B(n/2,1/2)}.
 \end{equation}

The limit $N \to \infty$ with $n=1$ is of special interest due to an interpretation in terms of the zeros of a certain random power series \cite{Kr06,Fo10a}.

\begin{proposition}\label{P3.9}
    Consider the random power series $f(z) = \sum_{p=0}^\infty a_p z^p$, where each $a_p$ is an independently distributed standard real Gaussian random variable. We have that the distribution of the zeros of $f(z)$ coincides with the limiting distribution of the eigenvalues of an $N \times N$ sub-block of  $(N+1) \times (N+1)$ Haar distributed real orthogonal matrices for $N \to \infty$.
\end{proposition}

\begin{proof}
Let $U_{N+1}$ denote the Haar distributed orthogonal matrix. Set $A_{N+1} = {\rm diag} \, (a,1,\dots,1)$, $a>0$. Then the matrix $U_{N+1} A_{N+1}$ has the same eigenvalues as $A_{N+1}^{1/2}U_{N+1} A_{N+1}^{1/2}$. In the limit $a\to 0$ these eigenvalues are
$0$ and the eigenvalues of the $N \times N$ sub-block of $U_{N+1}$ obtained by deleting the first row and column. Introducing 
$\mathbb I_{N+1}' = {\rm diag} \, (0,1,\dots,1)$ and $\mathbf e_1 = (1,0,\dots,0)$, manipulations including use of (\ref{6.79a}) reduce the  characteristic equation for the nonzero eigenvalues of $U_{N+1}A_{N+1}$ in the limit $a \to 0$ to the secular equation
$$
 \mathbf e_1^T (\mathbb I_{N+1} - \lambda U^\dagger_{N+1} \mathbb I_{N+1}')^{-1} U^\dagger_{N+1} \mathbf e_1 = 0;
$$
see \cite[proof of Prop.~1]{FI19}. Since for $|a| < 1$ the eigenvalues also have modulus less than, the inverse can be expanded according to the geometric series, with coefficient of $\lambda^k$ equal to $\mathbf e_1^T (U^\dagger_{N+1} \mathbb I_{N+1}')^k U^\dagger_{N+1} \mathbf e_1$. Noting that for large $N$ this coefficient has the form $\mathbf e_1^T (U^\dagger_{N+1} \mathbb I_{N+1}')^{k} \mathbf e_1(1 + {\rm O}(k/N))$, then applying the known result \cite{Kr06}
$$
\{\sqrt{N}(e_1^T (U^\dagger_{N+1} \mathbb I_{N+1}')^{k} \mathbf e_1\}_{k=0,1,\dots}\mathop{=}\limits^{\rm d} \{a_k \}_{k=0,1,\dots},
$$
where each $a_k$ is an independent standard real Gaussian, implies the statement of the proposition.
\end{proof}

The random power series in this result can itself be viewed as the $N\to \infty$ limit of the random polynomial $p_N(z)  = \sum_{p=0}^N a_k z^k$ with standard real Gaussian coefficients. The density of the real zeros of this random polynomial, $\rho_{(1),N}^{\rm r,K}(x)$ say, were first investigated by Kac \cite{Ka43}, making use of what is now referred to as the Kac-Rice formula; for an introduction see \cite{Ni14}. It was found
$$
\rho_{(1),N}^{\rm r,K}(x) = {1 \over \pi} \bigg ({1 \over (1 - x^2)^2} - {(N+1)^2 x^{2N} \over (1 - x^{2N+2})^2} \bigg )^{1/2}.
$$
Taking the limit $N \to \infty$ with $|x|<1$ reclaims the functional form (\ref{45}) with $x=y$ in the case $n=1$, as is consistent with this being the limiting density for the random matrix problem of Proposition \ref{P3.9}. The higher point correlation functions for the real zeros of the random polynomial are, according to the Kac-Rice formalism, structurally given by a $k$-dimensional Gaussian integral with a particular covariance matrix, weighted by a product of the absolute value of the integration variables \cite{BD04}. In addition, this structure is maintained upon taking the $N \to \infty$ limit. On the other hand the result of Proposition \ref{P3.9}, together with  Proposition \ref{P2.5} as it applies to truncations of Haar real orthogonal matrices, tells us that these same correlations can be written as a Pfaffian; see also \cite{MS13}. This is also true of for the correlations between the complex zeros, where the Kac-Rice formalism leads an expression in terms of a so-called Hafnian (a Hafnian relates to a Pfaffian as a permanent does to a determinant) \cite{Pr96}.

We have remarked that the generating function for the number of real eigenvalues is given by a formula structurally identical to (\ref{12.CaT}), but with $\tilde{\alpha}_{j,k}^{\nu \rm s}$ replaced by $\tilde{\alpha}_{j,k}^{\rm t}$ defined in (\ref{3.15x}) and $\tilde{r}_{j-1}^{\nu \rm s}$ replaced by $(n+2j)r_{j-1}^{\rm t}$. Specifically, with $n=1$ this shows \cite{PS18}
\begin{equation}\label{se1}
Z_N^{\rm t}(\zeta) \Big |_{n=1} =
\det \Big [ \delta_{j,k} + {\zeta -1\over \pi (j+k-3/2)} \Big ]_{j,k=1,\dots,N/2}
\mathop{\sim}\limits_{\zeta = 0 \atop N \to \infty} N^{-3/8},
\end{equation}
where the asymptotic result, applying when $\zeta = 0$, relates to the probability of no real eigenvalues; for more on the asymptotics see \cite{GP19,FTZ21}. The work \cite{PS18} (see also \cite{DPSZ02,SM07,SM08,DM15,CP17,Do18}) relate the exponent in the asymptotic formula, interpreted in terms of the probability of no real zeros for random Kac polynomials, to the so-called persistence exponent for two-dimensional ($d=2$) diffusion. 

The diffusion in question is of a scalar field $\phi(\mathbf x,t)$, which at time $t=0$ is a zero mean Gaussian random field with short range (delta function) correlations. The persistence is the probability that, for a system of linear size $L$ and with the origin $\mathbf 0$ in the bulk, $\phi(\mathbf 0,t)$ does not change sign from its initial value. The persistence exponent $\theta(d)$ quantifies the expected large $L$ asymptotic form that the probability decays as $L^{-2\theta(d)}$ for $t \gg L^d$. Through a common inter-relation with a Gaussian stationary process with (in logarithmic time) covariance ${\rm sech}(T/2)$ \cite{DPSZ02,SM07,SM08} it is predicted that $\theta(2)$ is equal to $1/4$ of the exponent in the analogue of the asymptotic formula (\ref{se1}) for Kac random polynomials, which in turn is twice the exponent in (\ref{se1}). The leads to the conclusion that $\theta(2) = 3/16$.

\subsection{Products of GinOE} \label{Section_product GinOE}
As for products of GinUE matrices, the most general case of products of compatibly sized rectangular matrices $G_i$ can be reduced to products of $N \times N$ square matrices $\tilde{G}_i$, now with element distribution proportional to
$$
| \det \tilde{G}_i \tilde{G}_i^T|^{\nu_i/2} e^{-{\rm Tr} \, \tilde{G}_i \tilde{G}_i^T/2},
$$
again with $\nu_i > 0$ equal to the difference in the number of rows in $G_i$ and the number of rows in $G_i$ $(=N)$ \cite{IK14}. The eigenvalue probability density function, conditioned on the number of real eigenvalues, can be computed for a general number $M$ of matrices in the product \cite[Prop.~8]{FI16}.

\begin{proposition}
    Define the real weight
    \begin{equation}\label{3.30}
w_r^\nu(\lambda) =
\prod_{j=1}^m \bigg[\int_\R d \lambda^{(j)}\Big(\frac{\lambda^{(j)}}{2}\Big)^{\nu_j/2}e^{-\frac12(\lambda^{(j)})^2}\bigg] \,
\delta( \lambda - \lambda^{(1)} \cdots \lambda^{(M)}).
\end{equation}
and the complex weight
\begin{equation}
 w_c^\nu(x,y) = 2\pi\, \int_\R d\delta\,\frac{|\delta|}{\sqrt{\delta^2+4y^2}}\,
 W^\nu\Big(\begin{bmatrix}\mu_+&0\\0&\mu_-\end{bmatrix}\Big),
 \: \:
  \mu_{\pm} = \frac12 \Big ( \pm | \delta | + [ \delta^2 + 4(x^2 + y^2) ]^{1/2} \Big )
\end{equation}
with 
\begin{equation}
 W^\nu(G)  =
 \prod_{l=1}^M\bigg[\int_{\R^{2\times 2}}(dX^{(l)})\det\Big(\frac{X^{(l)}X^{(l)T}}{2}\Big)^{\nu_l/2}
 \frac{e^{-\frac12 {\rm Tr}\,X^{(l)}X^{(l)T}}}{\sqrt{2\pi^3}}\bigg]
 \delta(X - X^{(1)} \cdots X^{(M)}).
\end{equation}
Define too the normalisation
\begin{equation}\label{nr}
Z_{N,k}^{M,\nu}=2^{M N(N+1)/4}\prod_{l=1}^M\prod_{j=1}^N\Gamma\Big(\frac{j+\nu_l}{2}\Big).
\end{equation}
With this notation, we have that the joint eigenvalue eigenvalue PDF of the random product matrix
$\tilde{G}_1 \cdots \tilde{G}_M$, conditioned on their being $k$ eigenvalues, is given by the functional form (\ref{3.1}) with the normalisation $C_N^{\rm g}$ replaced by $1/Z_{N,k}^{M,\nu}$,
$(\omega^{\rm g}(\lambda_s))^{1/2}$ replaced by $w_r^\nu(\lambda_s)$ and $\omega^{\rm g}(z_j)$ replaced by $ w_c^\nu(x_j,y_j)$.
\end{proposition}

\begin{proof}
Enabling this calculation is the real matrix version of the periodic Schur form \cite[Eq.~(2.69)]{BF22a}, where now each $Z_i$ is a block diagonal matrix of the structure of $R$ in the real Schur decomposition of Section \ref{S2.1}, and is thus conditioned on $k$ real eigenvalues. In keeping with the notation of Section \ref{S2.1}, we write $\tilde{G}_l=Q_l(D_l+R_l)Q_{l+1}$,$(l=1,\ldots,M)$, where each $D_l$ is block diagonal, and each $R_l$ is strictly block upper triangular. We denote the scalar diagonal elements of $D_l$ by $\{\lambda_s^{(l)} \}_{s=1}^k$ and the $2 \times 2$ block diagonal entries by $\{X_s^{(l)} \}_{s=1}^k$.

The first $k$ diagonal entries are scalars, $\{\lambda_i:=\lambda_t^{(1)}\cdots\lambda_t^{(l)}\}_{t=1}^k$, while the latter $(N-k)/2$ entries are $2 \times 2$ matrices, $\{G_s:=G_s^{(1)}\cdots G^{(l)}_s \}_{s=k}^{(N+k)/2}$. With this notation, the Jacobian for the above given change of variables reads~\cite[Prop. A.26]{Ip15}, \cite{FI16}
\begin{multline}
\prod_{l=1}^M (d \tilde{G}_l) =
\Big | \tilde{\Delta}(\{\lambda_l\}_{l=1,\dots,k} \cup
\{ x_j \pm i y_j \}_{j=1,\dots,(N-k)/2}) \Big | \\ \times
   \prod_{l=1}^M(d R_l) ( Q_l^T dQ_l) 
 \prod_{l=1}^M \Big (\prod_{j=1}^k d \lambda_j^{(l)} \prod_{s=k+1}^{(N+k)/2} dX_s^{(l)} \Big ),
\end{multline}
where $\lambda_s := \prod_{j=1}^M \lambda_s^{(j)}$ and $X_s := \prod_{j=1}^M X_s^{(j)}$, while $\tilde{\Delta}$ is as in (\ref{3.1}) but with the difference between each pair $x_s \pm i y_s$ (which are the eigenvalues of $X_s$) omitted.
Making the change of variables in the element PDF of the matrix product gives
 \[
 \prod_{l=1}^M e^{-\frac12 {\rm Tr} \, \tilde{G}_l \tilde{G}_l^T}  =
 \prod_{l=1}^M
 e^{- \frac{1}{2} \sum_{s=1}^k (\lambda_s^{(l)})^2 - \frac{1}{2}
 \sum_{s=k+1}^{(N+k)/2} {\rm Tr} \, X_s^{(l)} ( X_s^{(l)})^T}
 e^{- \frac12 \sum_{i < j} (r_{ij}^{(l)})^2}.
 \]
 Hence, after integrating over $\{r_{ij}^{(l)}\}$, we obtain that up to proportionality the eigenvalue PDF is equal to
 \begin{multline}\label{L1}
 \prod_{j=1}^k \delta(\lambda_j - \lambda_j^{(1)} \cdots \lambda_j^{(M)} )\, d \lambda_j
 \prod_{s=k+1}^{(N+k)/2} \delta(G_s - G_s^{(1)} \cdots G_s^{(M)})\, (d G_s)\\
 \times \frac{1 }{Z_{k,N}^{M,\nu}}
 \prod_{l=1}^M \bigg[
 \prod_{j=1}^k \Big( e^{- \frac12  (\lambda_j^{(l)} )^2 }d \lambda_j^{(l)} \Big) 
 \prod_{s=k+1}^{(N+k)/2} \Big( \frac{e^{- \frac12 {\rm Tr} \, X_s^{(l)} X_s^{(l) T}}}{\sqrt{2\pi^3}}(dX_s^{(l)}) \Big)
 \bigg ].
\end{multline}
The final step is to change variables from the off diagonal elements of the $X_s$, parameterised as in \S \ref{S2.1}, to the quantities $y_j$ (the imaginary part of $j$-th complex eigenvalue) and $\delta_j$; recall the final paragraph in \S \ref{S2.1}.
 \end{proof}

\begin{remark} \label{Remark Meiger G}
The real weight (\ref{3.30}) can be written in terms of the 
Meijer $G$-function according to
$$
\omega_{\rm r}^\nu(\lambda)=\MeijerG[\bigg]{M}{0}{0}{M}{-}{\frac{\nu_1}2, \dots, \frac{\nu_{M}}2}{\frac{\lambda^2}{2^M}};
$$
cf.~\cite[Eq.~(2.72)]{BF22a}. With $M=2$, $\nu_1=\nu_2=0$ this simplifies to 
\begin{equation}\label{2.71a}
\omega_{\rm r}^\nu (\lambda) = 2 K_0(|\lambda|).
\end{equation}
For general $M \ge 2$ no special function evaluation of the matrix integral defining $\omega_{\rm c}^\nu$ is known.  An exception is $M=2,\nu_1 = \nu_2=0$, for which \cite{APS10}, \cite{FI16}
\begin{equation}\label{L3d}
\omega_{\rm c}^\nu(x,y) = 4
\int_0^\infty \frac{1 }{ t} \exp \Big ( -  2 (x^2 - y^2)t - \frac{1 }{ 4t} \Big ) K_0(2(x^2+y^2)t) \, {\rm erfc} (2\sqrt{t} y) \,dt.
\end{equation}
Comparing (\ref{2.71a}) and (\ref{L3d}) makes it clear that unlike the structure highlighted in point (S3) of Section \ref{S3.1C}, it is no longer true that $\omega_{\rm r}(x) = (\omega_{\rm c}(x,0))^{1/2}$.
\end{remark}

We turn our attention now to the calculation of the skew orthogonal polynomials associated with $\gamma_{j,k}^{\rm p} := \alpha_{j,k}^{\rm p}|_{u=1} + \beta_{j,k}^{\rm p}|_{v=1}$. Here we define $\alpha_{j,k}^{\rm p}[u]$, $\beta_{j,k}^{\rm p}[v]$ as in Proposition \ref{P2.2}, but with 
$(\omega^{\rm g}(x)\omega^{\rm g}(y))^{1/2}$ replaced by $w_r^\nu(x)w_r^\nu(y)$ and $\omega^{\rm g}(z)$ replaced by $ w_c^\nu(x,y)$. Through the use  of (\ref{12.212g}), a simple calculation gives that the skew orthogonal polynomials are given by \cite[Prop.~9]{FI16}
\begin{equation}\label{L3e}
p_{2n}^{\nu,M}(z) = z^{2n}, \qquad 
p_{2n+1}^{\nu,M}(z) = z^{2j+1} - z^{2j-1}\prod_{k=1}^M(2j+\nu_k),
\end{equation}
with normalisation
\begin{equation}\label{L3f}
h_{j-1}^{M} = \prod_{k=1}^M {2 \sqrt{2\pi} \over 2^{\nu_k}}
\Gamma(2j + \nu_k - 1).
\end{equation}
However there is no longer a derivative formula for $p_{2n+1}^{\nu, M}(z)$ analogous to that for $p_{2n+1}^{\rm g}(z)$ in (\ref{12.212g+}). Associated with the latter is that the structural feature (S2) of earlier cases for the probability of all eigenvalues being real  is no longer true; see \cite[Eq.~(5.2)]{FI16} for a determinant formula involving particular Meijer $G$-functions. Asymptotics of the latter allows for an effect first noticed in \cite{La13} --- that the probability of all eigenvalues being real tends to $1$ as $M \to \infty$ for large classes of product random matrices --- to be proven in the Gaussian case \cite{Fo14}. Subsequent references on this topic include
\cite{HJL15,Re17,Re19}.

In keeping with the effect of the probability of all eigenvalues being real increasing as $M$ increases, is the result of Simm \cite[Th.~1.1]{Si17a} for the expected number of real eigenvalues, $E_N^{\nu,  M}$ say.
Here the point to draw attention to is the factor of $\sqrt{M}$ relative to the $M=1$ case.

\begin{proposition}
We have
\begin{equation}\label{L3g}
E_N^{\nu,  M} \mathop{\sim}\limits_{N \to \infty}
\sqrt{2N M \over \pi} + {\rm O}(\log N).
\end{equation}
\end{proposition}

\begin{proof} (Sketch)
 The starting point for this result, obtained with the specialisation each $\nu_i = 0$, is the exact expression for the density $\rho_{(1),N}^{{\rm r},M}(x)$ implied by knowledge of the skew orthogonal polynomials, written in terms of the corresponding product GinUE correlation kernel \cite[Eq.~(2.76)]{BF22a} according to the the structural form (\ref{2.37a}) with $x=y$.  The implied sum is then integrated term-by-term, which results in a sum of particular Meijer $G$-functions.
The summand is now in a form suitable for asymptotic analysis
in the large $N$ limit.
\end{proof}

The different behaviour of $E_N^{\nu,  M}$ for $N$ fixed and $M \to \infty$, relative to $M$ fixed and $N \to \infty$, make it natural to inquire about the asymptotic behaviour of $E_N^{\nu,  M}$ in the circumstance that $M = \alpha N$.
This has been determined in \cite[Th.~1.1]{AB22} as being given by
\begin{equation}\label{L3h}
\lim_{N \to \infty} {E_N^{\nu,  M} \over N} \bigg |_{M = \alpha N} = \Big ( 1 + {\alpha \over 4}\Big )
{\rm erf} \Big ( \sqrt{\alpha \over 8} \Big ) - {\alpha \over 4} +
\sqrt{\alpha \over 2 \pi} e^{- {\alpha \over 8}},
\end{equation}
which can be demonstrated to provide an interpolation between the previously found behaviours.

\begin{remark}\label{Rem_product GinOE} $ $ \\
1.~A further application of $\rho_{(1),N}^{{\rm r},M}(x)$ expressed in the structural form (\ref{2.37a}) is to the calculation of the global scaling limit. Thus in \cite[Th.~1.2]{Si17a} this was used to establish the validity of the corresponding product GinUE result \cite[Eq.~(2.77)]{BF22a}. \\
2.~As noted in \cite[Remark 2.18.4]{BF22a}, there is interest is general interest in a product of $M$ random matrices in the limit $M \to \infty$ from a dynamical systems perspective. For GinOE matrices, the explicit Lyapunov spectrum was computed long ago by Newman \cite{HJL15}. This exact result in the case of the largest Lyapunov exponent was generalised in \cite{Ka14,FZ18}
to allow for the GinOE matrices to be left multiplied by a positive definite matrix.  A direct study of the stability exponents associated with the modulus of the eigenvalues --- these becoming all real as $M \to \infty$ --- was undertaken in \cite{Ip15a}. \\
3.~The averaged absolute value of the product of GinOE matrices, which analogous to (\ref{eks1}) relates to the density of the real eigenvalues, has appeared in a counting problem for equilibria in an analysis of a discrete analogue of the random nonlinear differential equations (\ref{1.r})
\cite{IF18}. \\
4.~The products of $M$ truncated orthogonal matrices in \S\ref{Section_trunc GinOE} have also been studied in the literature \cite{FK18,FIK20,LMS21}. 
In particular, in the strong non-orthogonality, it was shown in \cite{LMS21} that  the leading order asymptotic of the expected number of real eigenvalues is of the form \eqref{3.26} multiplied by $\sqrt{M}$; cf. \eqref{L3g}. 
On the other hand, in the weak non-orthogonality, the asymptotic form \eqref{3.26a} with $B(n/2,1/2)$ replaced by $B(Mn/2,1/2)$ holds. 
\end{remark}

The structural form (\ref{2.37a}) for $S_N^{{\rm r}, M}(x,y)$ also has consequences for the analysis of the two-point correlation and associated fluctuation formula. In particular, it is used in \cite{FS23} to establish that the structural formula (\ref{eks2.b}) for the variance of a linear statistic as $N \to \infty$ holds true independent of $M$. For the particular linear statistic corresponding to the counting function for the number of real eigenvalues, this formula was shown to break down in the case that $M$ simultaneously tends to $\infty$ with $N$. Rather then $(\sigma_N^{{\rm r}, M})^2/E_N^{{\rm r}, M} |_{M = \alpha N}$ tends to an $\alpha$ dependent quantity; see \cite[Th.~1.2]{AB22}.

\section{Eigenvalue statistics for GinSE and elliptic GinSE}\label{S3}

\subsection{Eigenvalue PDF}
 Using a similarity transformation, the symplectic Ginibre matrix $G$ can be identified via the matrix-valued version of the quaternion realisation \eqref{1.1} as
\begin{equation}
    \begin{bmatrix}
        A & B \\
        -\bar{B} & \bar{A}
    \end{bmatrix} \in \C^{2 N \times 2 N}, 
\end{equation}
where $A, B$ are independent copies of GinUE. 
Due to this form, the matrix $G$ has $2N$ eigenvalues that come in complex conjugate pairs $\pm z_j$, where we take ${\rm Im} \, z_j > 0$. Importantly from the viewpoint of calculating the Jacobians associated with a change of variables involving eigenvalues, the eigenvectors of complex conjugate eigenvalues are related by a linear transformation. This can be encoded by forming a block Schur decomposition $G = U Z U^\dagger$ with $U$ a $2N \times 2N$ unitary matrix with $2 \times 2$ block entries of the form \eqref{1.1} --- a conjugation equivalent symplectic unitary matrix --- and $Z$ a block triangular matrix, with diagonal blocks $$\Big\{ \begin{bmatrix} 0 & z_j \\ \bar{z}_j & 0 \end{bmatrix} \Big\}_{j=1}^N$$ containing the eigenvalues of $G$, and off diagonal elements having the quaternion form \eqref{1.1}. As in the proof of \cite[Prop.~2.1]{BF22a}, the Jacobian calculation is carried out by forming the wedge product of the matrix of differentials of $U^\dagger dG U$, in the order of the block indices $(j,k)$ with $j$ decreasing from $N$ to 1, and $k$ increasing from $1$ to $N$. This gives the eigenvalue dependent factor
$$
\prod_{j=1}^N e^{ -2 W(z_j, \bar{z}_j) } |z_j-\bar{z}_j|^2 \prod_{1 \le j < k \le N} | z_k - z_j|^2 |z_k-\bar{z}_j|^2,
$$
and moreover (after some working) the product of differentials can be shown to factorise as in \cite[Eq.~(2.6)]{BF22a}. The eigenvalue PDF now follows by integrating over the independent off diagonal entries of $Z$. The resulting functional form \eqref{1.1g} is formally the same as the GinUE eigenvalue PDF  \eqref{1.1f} with $N \mapsto 2N$, and where $z_{j+N}$ is identified with $\bar{z}_j$, but with terms involving only differences of   $\{ \bar{z}_j \}$ ignored (due to dependencies in the quaternion structure these are not associated with independent differentials). As previously mentioned, the eigenvalue PDF \eqref{1.1g} of GinSE was derived already in the original work \cite{Gi65} of Ginibre.  

In general, an eigenvalue PDF of the non-Hermitian random matrices in the same symmetry class of the GinSE is of the form 
\begin{equation}\label{PDF S}
\frac{1}{ N!\, Z_N^{ \mathbb{H} }(W) } \prod_{j=1}^N e^{ -2 W(z_j, \bar{z}_j) } |z_j-\bar{z}_j|^2 \prod_{1 \le j < k \le N} | z_k - z_j|^2 |z_k-\bar{z}_j|^2, \quad {\rm Im} \, z_j > 0,
\end{equation}
where $Z_N^{ \mathbb{H} } \equiv Z_N^{ \mathbb{H} }(W) $ is the partition function, which turns \eqref{PDF S} into a probability measure.  
Here $W$ can be an arbitrary real function such that $Z_N^{ \mathbb{H} }$ exists, which furthermore is assumed to satisfy the complex conjugation symmetry $W(z,\bar{z})=W(\bar{z},z)$.
The ensemble of the form \eqref{PDF S} is called the planar symplectic ensemble \cite{ABK22}.

\subsection{Coulomb gas perspective} \label{Section Coulomb gas S}

The eigenvalue PDF \eqref{PDF S} can be rewritten in terms of the Hamiltonian 
\begin{equation} \label{Ham S}
H_{\mathbb{H}}(z_1,\dots,z_N):= \sum_{1 \le j < l \le N} \log \frac{1}{|z_j-z_l|^2 |z_j-\bar{z}_k|^2}+ \sum_{ j=1}^N 2 W(z_j,\bar{z}_j) + \log \frac{1}{ |z_j-\bar{z}_j|^2 }
\end{equation}
as
\begin{equation} \label{Gibbs symplectic}
\frac{1}{ N!\, Z_N^{ \mathbb{H} }(W) }\,e^{-H_\mathbb{H}(z_1,\dots,z_N)}
\end{equation}
Analogous to the discussion of \S\ref{S2.1a}, this can be regarded as a two-dimensional Coulomb gas \cite{Fo16} in the upper-half plane $\mathbb{H}$ with image charges in the lower half plane; cf. \cite{KS99}. 
As such, this is an image system counterpart of the eigenvalue PDF of the normal matrix model, which by way of comparison is given by 
\begin{equation} \label{Gibbs complex}
\frac{1}{ N!\,Z_N^{ \mathbb{C} }(W) }\,e^{-H_\mathbb{C}(z_1,\dots,z_N)}, \qquad H_{ \mathbb{C} }(z_1,\dots,z_N):= \sum_{1 \le j < l \le N} \log \frac{1}{|z_j-z_l|^2}+ \sum_{j=1}^N W(z_j,\bar{z}_j);
\end{equation}
see \cite[\S 5]{BF22a} and references therein. 

The Coulomb gas interpretation \eqref{Gibbs symplectic} allows one to describe the limiting spectral distribution using  logarithmic potential theory. 
In the scaling $W(z,\bar{z})=N Q(z)$, chosen to make the order of the interaction and the potential term in \eqref{Ham S} of the same order, one can observe that the continuum limit of the Hamiltonian \eqref{Ham S} divided by $2N^2$ is given by 
\begin{align}
I_Q[\mu]&:= \frac12 \int_{ \mathbb{C}^2 } \log \frac{1}{ |z-w|^2 |z-\bar{w}|^2 }\, d\mu(z)\, d\mu(w) +\int_{ \mathbb{C} } Q \,d\mu \nonumber
\\
&=  \int_{ \mathbb{C}^2 } \log \frac{1}{ |z-w| }\, d\mu(z)\, d\mu(w) +\int_{ \mathbb{C} } Q \,d\mu. \label{IQ mu}
\end{align}
Here, we have used the complex conjugation symmetry $Q(z)=Q(\bar{z})$ for the second line. 
Thus it is natural to expect that the empirical measure $\mu_Q$ of \eqref{Gibbs symplectic} converges to Frostman's equilibrium measure, a unique probability measure minimising the energy functional \eqref{IQ mu}. 
This convergence was shown by Benaych-Georges and Chapon \cite{BC12} for a general potential $Q$. 
In particular it shows that the limiting spectral distributions of \eqref{Gibbs symplectic} and \eqref{Gibbs complex} are identical \cite[\S 5.2]{BF22a}, extending the universal appearance of the circular law for the GinUE and GinSE. 

With regards to quantitative features, let us first recall that under minor assumptions on $Q$, the equilibrium measure $\mu_Q$ is absolutely continuous and takes the form 
\begin{equation} \label{eq msr form}
  d\mu_Q(z)= \frac{ \partial_z \partial_{ \bar{z} } Q(z)}{\pi} \chi_{ z\in S_Q }  \,d^2 z, 
\end{equation}
where the compact set $S_Q$ is called the droplet. 
For a radially symmetric potential $q(r)=Q(|z|=r)$ which is strictly subharmonic in $\C$, the droplet is of the form $S_Q=\{ R_1 \le |z| \le R_2 \}$, where the pair of constant $(R_1,R_2)$ is characterised by
\begin{equation} \label{droplet annulus}
R_1 q'(R_1)=0, \quad R_2 q'(R_2)=2, 
\end{equation}
see \cite[\S~IV.6]{ST97}. 
In particular, for $Q(z)=|z|^2$, it follows that $ \partial_z \partial_{ \bar{z} } Q(z)=1$ and $R_0=0,R_1=1$, which coincides with the circular law of the GinSE.

\begin{remark}\label{R5.1}
Underlying the image charge viewpoint of (\ref{Ham S}) is the pair potential (\ref{pp}), which we know is the solution of the two-dimensional Poisson equation in Neumann boundary conditions along the $x$-axis. If instead we take the point $\vec{r}$ to be inside a disk of radius $R$ and require Neumann boundary conditions on the boundary of the disk, the pair potential is \cite[Eq.~(15.188) with $\epsilon = 0$]{Fo10}
$$
\phi(\vec{r},\vec{r}') = - \log \Big (
|z - z'||R - z z'/R| \Big ).
$$
Imposing a smeared out charge neutral background of density $1/\pi$ (and hence taking $R=\sqrt{N}$) the corresponding charge neutral Boltzmann factor is
\cite[Eq.~(15.190)]{Fo10}
\begin{equation}\label{im}
A_{N,\beta} e^{-\beta \sum_{j=1}^N | z_j|^2/2}
\prod_{1 \le j < k\le N} |z_k - z_j|^\beta |1 - z_j \bar{z}/N|^\beta  \prod_{j=1}^N
(1 - |z_j|^2/N)^{\beta/2}, 
\end{equation}
where $A_{N,\beta} = e^{-\beta N^2((1/4) \log N - 3/8)}$.
Here $\beta > 0$ is the inverse temperature, with the case $\beta = 2$ being the disk analogue of (\ref{Ham S}), although this viewpoint gives the self energy term of exponent $1$ rather than $2$ as in (\ref{1.1g}); see \cite[\S 2.1]{Fo16} for more on this point.
\end{remark}

\subsection{Skew orthogonal polynomials}

We define the skew-symmetric form $\langle \cdot , \cdot  \rangle_{s,S}$ by
\begin{equation} \label{skew product}
\langle f, g \rangle_{s,S} := \int_{\C} \Big( f(z) g(\bar{z}) - g(z) f(\bar{z}) \Big) (z - \bar{z}) e^{-2 W(z, \bar{z} )} \,d^2z,
\end{equation}
 where in keeping the notation of (\ref{sO}) the subscripts indicate (s)kew and Gin(S)E.
 Note the similarity with the second term in \eqref{sO}.
As discussed in \S \ref{S2.4x}, a family $\{q_{m}\}_{m \geq 0}$ of monic polynomials $q_m$ of degree $m$ is said to be a family of skew-orthogonal polynomials if the following skew orthogonality conditions hold: for all $k, l \in \mathbb{N}$
\begin{equation}\label{SOP S}
\langle q_{2k}, q_{2l} \rangle_{s,S} = \langle q_{2k+1}, q_{2l+1} \rangle_{s,S} = 0, \qquad \langle q_{2k}, q_{2l+1} \rangle_{s,S} = -\langle q_{2l+1}, q_{2k} \rangle_{s,S} = r_k  \,\delta_{k, l}.
\end{equation}
We mention that the matrix averages formulas
(\ref{12.212f}) are again valid; see \cite{Ka02}. 
In distinction to ensembles based on GinOE, there are also alternative methods which in fact have a broader scope, so these instead will be discussed below.
\begin{proposition} \label{Prop_SOP for radial}
For a radially symmetric potential $W(z,\bar{z})=\omega(|z|),$ let 
\begin{equation}
h_k=2\pi \int_0^\infty r^{2k+1} e^{-2\omega(r)}\,dr
\end{equation}
be the squared orthogonal norm. 
Then 
\begin{equation} \label{skew op_rad}
	q_{2k+1}(z)=z^{2k+1}, \qquad
	q_{2k}(z)=z^{2k}+\sum_{l=0}^{k-1}  z^{2l} \prod_{j=0}^{k-l-1} \frac{h_{2l+2j+2}  }{ h_{2l+2j+1} }
\end{equation}
forms a family of skew-orthogonal polynomials.
Furthermore, the skew-norm is given by $r_k=2h_{2k+1}$.
\end{proposition}
\begin{proof}
Since the monomials $\{z^k\}$ are orthogonal polynomials with respect to a rotationally symmetric weight function, it follows that 
\begin{align*}
    \langle z^k, z^l \rangle_{s,S}  
    &=  \int_\C \Big( z^{k+1} \bar{z}^l-z^{l+1} \bar{z}^k-  z^k \bar{z}^{l+1}+z^l \bar{z}^{k+1}\Big) e^{-2\omega(|z|)}\,d^2 z
  = 2\delta_{k+1,l} h_{k+1} - 2\delta_{l+1,k} h_k .
\end{align*}
Note here that the indices in the the Kronecker delta differs by one, which in turn immediately leads to $ \langle q_{2k+1}, q_{2l+1} \rangle_{s,S} = 0$.
Furthermore, it follows that $\langle q_{2k}, q_{2l+1} \rangle_{s,S}=0$ if $l>k.$ 
Let us write 
\begin{equation}
a_l= \prod_{j=0}^{k-l-1} \frac{h_{2l+2j+2}  }{ h_{2l+2j+1} }.
\end{equation}
Then we have
\begin{align*}
\langle q_{2k}, q_{2k+1} \rangle_{s,S} &= \Big\langle z^{2k}+\sum_{l=0}^{k-1}  a_l z^{2l} \, , \,  z^{2k+1} \Big\rangle_{s,S} = 2 h_{2k+1}. 
\end{align*}
On the other hand, for the case $k<l$, after straightforward computations, we obtain
\begin{align*}
\langle q_{2k+1}|q_{2l} \rangle_{s,S}= 2 (  a_k h_{2k+1}-a_{k+1} h_{2k+2} )=0.
\end{align*}
Therefore, we have shown that $ \langle q_{2k+1}, q_{2l+1} \rangle_{s,S} = r_k \delta_{k,j}$. 
The other cases follow from similar computations with minor modifications. 
\end{proof}

As an example of Proposition~\ref{Prop_SOP for radial}, for $W^{\rm g}(z,\bar{z})=|z|^2$, the associated skew orthogonal polynomials $q_k^{ \rm g }$ are given by  
\begin{equation} \label{SOP Ginibre}
 q_{2k+1}^{ \rm g }(z)=z^{2k+1}, \qquad q_{2k}^{ \rm g }(z)= \sum_{l=0}^k \frac{k!}{ l! }z^{2l}, \qquad r_k^{\rm g}= \frac{(2k+1)!}{2^{2k+1}}\pi.    
\end{equation}
These are a special case of the skew orthogonal polynomials obtained in \cite{Ka02} for the elliptic GinSE; see the sentence below \eqref{SOP eGinibre}. 
We also refer to \cite[Exercises 15.9 q.2]{Fo10} and \cite[Prop.~1]{Fo13a} for a derivation of \eqref{SOP Ginibre}. 

The crux of Proposition~\ref{Prop_SOP for radial} is that one can construct the skew orthogonal polynomials using the associated (monic) orthogonal polynomials $p_j$  with respect to the same weight, i.e. 
\begin{equation} \label{planar OP}
\int_\C p_j(z) \overline{p_k(z)} e^{-2 W(z,\bar{z})}\,d^2z = h_k \delta_{j,k}. 
\end{equation}
A setting beyond the radially symmetric case in which we can construct the skew orthogonal polynomials is when the associated orthogonal polynomials satisfy a three-term recurrence relation.

\begin{proposition} \label{Prop_SOP for 3 term}
Suppose that the sequence of monic orthogonal polynomials $(p_j)$ satisfies the three term recurrence relation 
\begin{equation} \label{3 term recur}
z p_k(z)= p_{k+1}(z) +b_k p_k(z)+c_k p_{k-1}(z), \qquad b_k, c_k \in \mathbb{R}. 
\end{equation}
Then 
\begin{equation}
q_{2k+1}(z)= p_{2k+1}(z), \qquad q_{2k}(z)=  \sum_{l=0}^k a_{l} z^l, \qquad a_l:= \prod_{ j=0 }^{ k-l-1 } \frac{ h_{2l+2j+2}-c_{2l+2j+2}h_{2l+2j+1} }{ h_{2l+2j+1}-c_{2l+2j+1} h_{2l+2j} } 
\end{equation}
satisfies \eqref{SOP S} with $r_k=2(h_{2k+1}-c_{2k+1}h_{2k})$.
\end{proposition}
\begin{proof}(Sketch)
The essential idea of the proof has already been given in the proof of Proposition~\ref{Prop_SOP for radial} above. 
The notable difference is that while computing the skew-symmetric form $\langle q_k, q_l \rangle_{s,S}$, we expand the term 
\begin{equation}
\Big( q_k(z) \overline{ q_l(z) } - \overline{ q_k(z) } q_l(z) \Big) (z-\bar{z})
\end{equation}
in the integrand using the three term recurrence relation \eqref{3 term recur}.
\end{proof}

The idea of constructing skew orthogonal polynomials in Proposition~\ref{Prop_SOP for 3 term} first appeared in \cite{Ka02}, where the Hermite polynomials were considered.  
Later, this was extended to the Laguerre polynomials in \cite{Ak05}. 
The general statement in Proposition~\ref{Prop_SOP for 3 term} was given in \cite{AEP22}. 
As an example, we consider the elliptic GinSE
potential 
\begin{equation} \label{W elliptic}
W^{ \rm e }(z,\bar{z})= \frac{1}{1-\tau^2} (|z|^2-\tau \, {\rm{Re}} \,z^2), \qquad \tau \in [0,1). 
\end{equation}
The associated monic orthogonal polynomials and norms are given by 
\begin{equation}
p_k^{ \rm e }(z) = \Big( \frac{\tau}{4} \Big)^{k/2}  H_k \Big( \frac{z}{ \sqrt{\tau} } \Big), \qquad h_k^{ \rm e }= \sqrt{1-\tau^2} \frac{k!}{2^{k+1}}\pi,
\end{equation}
see e.g. \cite[Lem.~7]{ACV18}.
Then by using the recurrence relation 
\begin{equation}
z p_k^{ \rm e }(z) = p_{k+1}^{ \rm e }(z) + \frac{\tau}{2}\,k\, p_{k-1}^{ \rm e }(z) 
\end{equation}
and Proposition~\ref{Prop_SOP for 3 term}, we have 
\begin{equation} \label{SOP eGinibre}
 q_{2k+1}^{ \rm e }(z)=p_{2k+1}^{ \rm e }(z), \qquad q_{2k}^{ \rm e }(z)= \sum_{l=0}^k \frac{k!}{ l! } p_{2l}^{ \rm e }(z), \qquad r_k^{\rm e}= (1-\tau) \sqrt{1-\tau^2} \frac{(2k+1)!}{2^{2k+1}}\pi.    
\end{equation}
Note that \eqref{SOP Ginibre} can be recovered by taking the $\tau \to 0$ limit of \eqref{SOP eGinibre}.

\subsection{Correlation functions and sum rules}\label{Section_correlation GinSE}

The $k$-point correlation function $\rho_{(k),N}^{\rm s}$ of the ensemble \eqref{Gibbs symplectic} is given by 
\begin{equation}\label{bfRNk def}
\rho_{(k),N}^{\rm s}(z_1,\dots, z_k) := \frac{N!}{(N-k)!} \frac{1}{ Z_N^{ \mathbb{H} }(W) }  \int_{\mathbb{C}^{N-k}} e^{-H_{ \mathbb{H} }(z_1,\dots,z_N) }\prod_{j=k+1}^N \, d^2 z_j.
\end{equation}
As an analogue of Proposition~\ref{P2.5}, a Pfaffian formula for $\rho_{(k),N}^{\rm s}$ follows \cite{Ka02}.

\begin{proposition}\label{Prop_Pf S}
Let $q_j$ be the skew orthogonal polynomials as specified in \eqref{SOP S}. 
Using these polynomials, define 
\begin{equation}\label{kappaN skewOP}
\kappa_N^{\rm s}(z,w)=\sum_{k=0}^{N-1} \frac{q_{2k+1}(z) q_{2k}(w) -q_{2k}(z) q_{2k+1}(w)}{r_k}.
\end{equation}
Then we have 
\begin{equation} \label{rhoNk Pf}
\rho_{(k),N}^{\rm s}(z_1,\dots, z_k) =\prod_{j=1}^{k} (\overline{z}_j-z_j)  {\rm Pf} \, [\mathcal K_N^{\rm s}(z_j,z_l)]_{j,l=1,\dots,k},
\end{equation}
where 
\begin{equation}
\mathcal{K}_N^{\rm s}(z,w)= e^{ -W(z,\bar{z})-W(w,\bar{w}) } 
\begin{bmatrix} 
\kappa_N^{\rm s}(z,w) & \kappa_N^{\rm s}(z,\bar{w})
\smallskip 
\\
\kappa_N^{\rm s}(\bar{z},w) & \kappa_N^{\rm s}(\bar{z},\bar{w}) 
\end{bmatrix}. 
\end{equation}
\end{proposition}

Using Proposition~\ref{Prop_Pf S} and \eqref{SOP Ginibre}, it follows that the matrix entry $\kappa_N^{ \rm g }(z,w)$--- referred to as the pre-kernel --- of the GinSE is given by 
\begin{align} 
\kappa_N^{ \rm g }(z,w)
&= \frac{ \sqrt{2}  }{ \pi } \bigg( \sum_{k=0}^{N-1} \frac{( \sqrt{2}z )^{2k+1}}{(2k+1)!!}    \sum_{l=0}^k \frac{( \sqrt{2} w)^{2l} }{(2l)!!} -   \sum_{k=0}^{N-1} \frac{( \sqrt{2} w)^{2k+1}}{(2k+1)!!}    \sum_{l=0}^k \frac{( \sqrt{2} z)^{2l} }{(2l)!!} \bigg). \label{kappaN S g} 
\end{align}
One strategy for analysing the double summation is to derive a suitable differential equation for the kernel.
This idea essentially goes back to an early work \cite{MS66} of Mehta and Srivastava. 

As an analogue of Proposition~\ref{P2.9}, the following holds true; see \cite{ABK22}.  

\begin{proposition}\label{Prop_CDI g}
 Letting $\widehat{\kappa}_N^{\rm g}(z,w):=e^{-2zw} \kappa_N^{\rm g}(z,w)$, we have 
\begin{equation} \label{CDI for kappa g v2}
\partial_z \widehat{\kappa}_N^{ \rm g }(z,w) = 2(z-w) \widehat{\kappa}_N^{ \rm g }(z,w)  + \frac{2}{\pi}  \frac{ \Gamma(2N;2zw) }{(2N-1)!} -\frac{1}{\pi} (2z)^{2N} e^{-z^2} \frac{ \Gamma(N;w^2) }{ (2N-1)! }. 
\end{equation}
\end{proposition}
\begin{proof}
Note that 
\begin{align*}
&\quad \partial_z    \sum_{k=0}^{N-1} \frac{( \sqrt{2}z )^{2k+1}}{(2k+1)!!}    \sum_{l=0}^k \frac{( \sqrt{2} w)^{2l} }{(2l)!!}  = \sqrt{2}   \sum_{k=0}^{N-1} \frac{( \sqrt{2}z )^{2k}}{(2k-1)!!}    \sum_{l=0}^k \frac{( \sqrt{2} w)^{2l} }{(2l)!!}
\\
&= \sqrt{2}   \sum_{k=1}^{N-1} \frac{( \sqrt{2}z )^{2k}}{(2k-1)!!} \sum_{l=0}^{k-1} \frac{( \sqrt{2} w)^{2l} }{(2l)!!} + \sqrt{2}   \sum_{k=0}^{N-1} \frac{(  \sqrt{2}z )^{2k}}{(2k-1)!!}   \frac{(  \sqrt{2} w)^{2k} }{(2k)!!}.
\end{align*}
Rearranging the terms, we have 
\begin{align*}
\partial_z   \sum_{k=0}^{N-1} \frac{(  \sqrt{2}z )^{2k+1}}{(2k+1)!!}    \sum_{l=0}^k \frac{( \sqrt{2} w)^{2l} }{(2l)!!} 
&=  2 z   \sum_{k=0}^{N-1} \frac{(\sqrt{2}z )^{2k+1}}{(2k+1)!!} \sum_{l=0}^{k} \frac{( \sqrt{2} w)^{2l} }{(2l)!!}
\\
&\quad +\sqrt{2}    \sum_{k=0}^{N-1} \frac{( \sqrt{2}z )^{2k}}{(2k-1)!!}   \frac{( \sqrt{2} w)^{2k} }{(2k)!!}-\sqrt{2}  \frac{( \sqrt{2}z )^{2N}}{(2N-1)!!} \sum_{l=0}^{N-1} \frac{( \sqrt{2} w)^{2l} }{(2l)!!}.
\end{align*}
Similarly, we have
\begin{align*}
\partial_z    \sum_{k=0}^{N-1} \frac{( \sqrt{2} w)^{2k+1}}{(2k+1)!!}    \sum_{l=0}^k \frac{( \sqrt{2} z)^{2l} }{(2l)!!}
&= 2 z \sum_{k=0}^{N-1} \frac{( \sqrt{2} w)^{2k+1}}{(2k+1)!!}    \sum_{l=0}^{k} \frac{( \sqrt{2} z)^{2l} }{(2l)!!}
- 2z \sum_{k=0}^{N-1} \frac{(\sqrt{2} w)^{2k+1}}{(2k+1)!!}  \frac{(\sqrt{2} z)^{2k} }{(2k)!!}.
\end{align*}
Combining above identities, we conclude \eqref{CDI for kappa g v2}.  
\end{proof}

As a consequence of Proposition~\ref{Prop_CDI g}, the uniform asymptotic expansion \eqref{Tr} can be used to derive a linear inhomogeneous differential equation of order one satisfied by the limiting correlation kernels.
Then the anti-symmetry of the limiting pre-kernels characterises a unique solution, which in turn determines the limiting kernels; see \cite{ABK22,BE22}. 

For the bulk case, we have
\begin{equation}
\lim_{N \to \infty} \rho_{(k),N}^{\rm s}(z_1,\dots, z_k) =\prod_{j=1}^{k} (\overline{z}_j-z_j)  {\rm Pf} \, [\mathcal K_\infty^{\rm s,b}(z_j,z_l)]_{j,l=1,\dots,k}, 
\end{equation}
where $\mathcal{K}_\infty^{\rm s,b}$ is of the form 
\begin{equation}\label{kappa S bulk}
\mathcal{K}_\infty^{\rm s,b}(z,w):= e^{-|z|^2-|w|^2} \begin{bmatrix} 
\kappa_\infty^{\rm s,b}(z,w) & \kappa_\infty^{\rm s,b}(z,\bar{w})
\smallskip 
\\
\kappa_\infty^{\rm s,b}(\bar{z},w) & \kappa_\infty^{\rm s,b}(\bar{z},\bar{w}) 
\end{bmatrix}, \quad \kappa_\infty^{\rm s,b}(z,w):=\frac{ e^{ z^2+w^2 }  }{\sqrt{\pi}}  {\rm{erf}} (z-w).
\end{equation}
In particular, this gives the bulk scaled density
\begin{equation} \label{rho bulk Dawson}
\rho_{(1),\infty}^{\rm s,b}(x+iy) = \frac{4}{\pi} y F(2y), 
\end{equation}
where $F(z):=e^{-z^2} \int_0^z e^{t^2}\,dt$ is Dawson's integral function. As $y \to 0^+$, this tends to zero with leading term ${8y^2 \over \pi}$.
The limiting pre-kernel \eqref{kappa S bulk} first appeared in the second edition of Mehta's book \cite{Me91} using a different approach; cf. \cite[Lem.~3.5.2]{Ly21}. 
Later, this was rederived by Kanzieper in \cite{Ka02} using a similar method described above.

Contrary to the bulk case, the analogous result for the edge case appeared in the literature \cite{ABK22} only recently. 
The same Pfaffian structure was found
\begin{equation} \label{scaling limit edge S}
\lim_{N \to \infty} \rho_{(k),N}^{\rm s}(\sqrt{N}+z_1,\dots, \sqrt{N}+z_k) =\prod_{j=1}^{k} (\overline{z}_j-z_j)  {\rm Pf} \, [\mathcal K_\infty^{\rm s,e}(z_j,z_l)]_{j,l=1,\dots,k}, 
\end{equation}
 with $\mathcal K_\infty^{\rm s,e}$  as for $\mathcal K_\infty^{\rm s,b}$ in \eqref{scaling limit edge S} but with $\kappa_\infty^{\rm s,b}$ replaced by $\kappa_\infty^{\rm s,e}$ throughout, where 
\begin{equation}\label{kappa S edge}
\kappa_\infty^{\rm s,e}(z,w):=\frac{ e^{2zw}  }{\pi} \int_{-\infty}^{0} e^{-t^2} \sinh( 2t(w-z) ) {\rm{erfc}}(z+w-t)\,dt. 
\end{equation}
Setting $w=\bar{z}$, this gives for the density 
\begin{equation} \label{limiting edge density S}
\rho_{(1),\infty}^{\rm s,e}(x+iy) = -\frac{2y}{\pi}  \int_{-\infty}^{0}e^{-s^2}\sin(4sy) {\rm{erfc}}(2x-s) \, ds. 
\end{equation}
Note here that unlike the bulk case, the edge scaling limit does not have the translation invariance along the horizontal $x$-direction. 
We also remark that
\begin{equation} \label{edge density y inf}
\rho_{(1),\infty}^{\rm s,e}(x+iy) \sim \frac{{\rm{erfc}}(2x)}{2\pi} + {e^{-4 x^2} \over 8 \pi^{3/2} y^2} + \cdots ,  \qquad y \to \infty.
\end{equation}
The leading term of the RHS of \eqref{edge density y inf} as a function of $x$ coincides with the limiting edge density of the GinUE up to scaling.
As in \S~\ref{Section_GinOUE cor}, such a limiting relation holds too for the general $k$-point function, which in particular exhibits a deformation from a Pfaffian to a determinant; see \cite[Cor.~2.2]{ABK22}.

A structural feature, highlighted in \cite{ABK22,BES23}, is that the kernels \eqref{kappa S bulk} and \eqref{kappa S edge} can be expressed in a unified way
\begin{equation} \label{kappa Wronskian}
	\kappa_\infty^{\rm s}(z,w):=\frac{ e^{z^2+w^2}  }{\sqrt{\pi}} \int_{E} W(f_{w},f_{z})(u) \, du,
\end{equation} 
where $W(f,g):=fg'-gf'$ is the Wronskian, and 
  \begin{equation} \label{fE NH bulk edge}
    	f_z(u):=\frac{ {\rm{erfc}}(\sqrt{2}(z-u)) }{2}, \qquad 
    	E:=\begin{cases}
    	    (-\infty,\infty) &\textup{for the bulk case},
    	    \smallskip 
    	    \\
    	    (-\infty,0) &\textup{for the edge case}.
    	\end{cases}
    \end{equation}
As in \S\ref{S2.5c}, the expression \eqref{kappa Wronskian} allows one to observe the bulk limiting form from the edge limiting form with $z,w \to -\infty$.

We now discuss sum rules for the limiting densities \eqref{rho bulk Dawson} and \eqref{limiting edge density S}; cf.\cite[Prop.~4.3, 4.4]{BF22a}.

\begin{proposition}
We have 
\begin{equation} \label{balance bulk s}
\int_{-\infty}^\infty \Big(\rho_{(1),\infty}^{\rm s,b}(x+iy) - \frac{1}{\pi} \Big)\,dy=0
\end{equation}
and 
\begin{equation} \label{balance edge s}
\int_{  (-\infty,\infty)^2  } 
\Big( \rho_{(1),\infty}^{\rm s,e}(x+iy) -  \rho_{(1),\infty}^{\rm s,b}(x+iy) \chi_{ x<0  } \Big)\,dx\,dy  =- \frac{1}{8}. 
\end{equation}
\end{proposition}
\begin{proof}
 The first identity \eqref{balance bulk s} immediately follows from the property of Dawson's integral $ F(2y)/2= \int (1- 4y F(2y)) \,dy$.  
 For the second identity \eqref{balance edge s}, letting $D_R$ be a disk of radius $R>0$, we first observe that by the change of variables, 
 \begin{align*}
&\quad   \int_{  D_R  }  \Big(\rho_{(1),\infty}^{\rm s,b}(x+iy) \chi_{ x<0  }-\rho_{(1),\infty}^{\rm s,e}(x+iy) \Big)\,dx\,dy - \int_{  D_R  }  \Big(\rho_{(1),\infty}^{\rm s,b}(x+iy) - {1 \over \pi} \Big ) \chi_{ x<0  }\,dx\,dy 
\\
&=  \int_{  D_R  }  \Big( \frac{\chi_{ x<0  }}{\pi}-\rho_{(1),\infty}^{\rm s,e}(x+iy) \Big)\,dx\,dy =  \int_{ D_R }  \Big( \frac{1}{2\pi}-\rho_{(1),\infty}^{\rm s,e}(x+iy) \Big)\,dy\,dx 
\\
&= \frac{1}{\pi}  \int_{ D_R  }  \Big( \frac1{2}+2y\int_{-\infty}^{0}e^{-s^2}\sin(4sy){\rm{erfc}}(2x-s) \, ds \Big)\,dy\,dx
\\
&= \frac{1}{\pi}\int_{  D_R  }  \Big( \frac12-2y\int_{0}^{\infty}e^{-s^2}\sin(4sy)\Big(2-{\rm{erfc}}(2x-s)\Big) \, ds \Big)\,dy\,dx.
 \end{align*}
By adding the last two expressions and using 
\begin{equation}
4y \int_0^\infty e^{-s^2} \sin(4s y)= 4y F(2y)= \pi  \rho_{(1),\infty}^{\rm s,e}(x+iy), 
\end{equation}
it follows that 
\begin{align*}
& \int_{  D_R  }  \Big(\rho_{(1),\infty}^{\rm s,b}(x+iy) \chi_{ x<0  }-\rho_{(1),\infty}^{\rm s,e}(x+iy) \Big)\,dx\,dy \\
&\qquad = \frac{1}{\pi} \int_{  D_R  }  y \Big( \int_{ -\infty }^\infty  \, e^{-s^2} \sin(4sy){\rm{erfc}}(2x-s)\,ds \Big) \,dy\,dx. 
\end{align*}
Furthermore, by using $\int_{-\infty}^\infty e^{-s^2} \sin(4sy)\,ds=0$ and the principal value integral
\begin{equation}
\lim_{ R \to \infty } \int_{-R}^R {\rm{erf}}(s-2x) \,dx = -s, 
\end{equation}
we obtain 
\begin{align*}
&\quad -\lim_{ R \to \infty } \frac{1}{\pi}  \int_{D_R} y \Big( \int_{-\infty}^\infty  \, e^{-s^2} \sin(4sy) {\rm{erf}}(s-2x) \,ds \Big) \,dy\,dx
 \\
 &=\frac{1}{\pi} \int_{-\infty}^\infty y  \Big( \int_{-\infty}^\infty   e^{-s^2} \sin(4sy) s  \,ds \Big) \,dy 
=\frac{2}{\sqrt{\pi}} \int_{-\infty}^\infty y^2 \,e^{-4y^2}\,dy  =\frac{1}8,
\end{align*}
which leads to \eqref{balance edge s}.
\end{proof}

Note that compared to the analogous identity \cite[Eq.~(4.9)]{BF22a} for the GinUE edge scaling limit, the RHS of \eqref{balance edge s} takes on the nonzero value $-1/8$. Curiously, the companion identity to \cite[Eq.~(4.9)]{BF22a}, namely \cite[Prop.~4.4 with $\beta = 2$]{BF22a}, which relates to the dipole moment of the edge density profile takes on the nonzero value $-1/8 \pi$. 
 
We remark that the bulk multipole screening sum rule analogous to (\ref{su3})
\begin{equation}\label{su3q}
 \int_{\mathbb C_+} w^{2p} \rho_{(2),\infty}^{{\rm s,b} \,T} (z,w) \, d^2 w  = -  z^{2p} \rho_{(1),\infty}^{{\rm s, b} }(z), \quad p \in \mathbb Z_{\ge 0},
\end{equation}
has been established in \cite{Fo15}.

\subsection{Elliptic GinSE}
The elliptic GinSE is defined in a similar way as in \S\ref{Section EGinOE}. 
Namely, for a parameter $\tau$, $0\le \tau<1$, set
\begin{equation}
X= \sqrt{1 + \tau} \, S + \sqrt{1 - \tau} \, A.
\end{equation}
Here $S$ is a member of Gaussian symplectic ensemble (GSE), whereas $A$ is a member of anti-symmetric GSE.
Its eigenvalue PDF is of the form \eqref{Gibbs symplectic} with the elliptic GinSE potential \eqref{W elliptic} previously mentioned. 

As discussed in \S~\ref{Section Coulomb gas S}, the global scaled eigenvalues $z_j \to \sqrt{N} z_j$ distribute in a way to minimise the energy \eqref{IQ mu} with the potential \eqref{W elliptic}.
This minimisation problem can be exactly solved, which gives that the limiting spectrum is given by the ellipse 
\begin{equation} \label{ellipse}
\Big\{ (x,y) \in \R^2 : \Big( \frac{x}{1+\tau} \Big)^2+ \Big( \frac{y}{1-\tau} \Big)^2 \le 1  \Big\}; 
\end{equation}
see e.g. \cite[\S 2.3]{BF22a} and \cite{By23}. 
As in the elliptic GinU/OE, this is the elliptic law for the elliptic GinSE. 

Turning to the correlation functions, by combining \eqref{kappaN skewOP} and \eqref{SOP eGinibre}, one can show that the associated pre-kernel $\kappa_N^{ \rm e }(z,w)$ is evaluated as 
\begin{align}
\kappa_N^{ \rm e }(z,w) &=   \frac{ \sqrt{2} }{ \pi(1-\tau) \sqrt{1-\tau^2} }\sum_{k=0}^{N-1} \frac{ (\tau/2)^{k+1/2} }{(2k+1)!!} H_{2k+1} \Big( \frac{z}{ \sqrt{\tau} } \Big) \sum_{l=0}^k \frac{(\tau/2)^l}{(2l)!!}  H_{2l} \Big( \frac{w}{\sqrt{\tau}} \Big)  \nonumber
\\
& \quad - \frac{ \sqrt{2} }{ \pi(1-\tau) \sqrt{1-\tau^2} }\sum_{k=0}^{N-1} \frac{ (\tau/2)^{k+1/2} }{(2k+1)!!} H_{2k+1} \Big( \frac{w}{ \sqrt{\tau} } \Big) \sum_{l=0}^k \frac{(\tau/2)^l}{(2l)!!}  H_{2l} \Big( \frac{z}{\sqrt{\tau}} \Big). \label{kappaN S e}
\end{align}
It reduces to the pre-kernel $\kappa_N^{ \rm g }$ in \eqref{kappaN S g} in the limit $\tau \to 0^+$.
In the weakly non-Hermitian regime when $\tau=1-\alpha^2/N$, the scaling limit of the correlation functions at the origin was derived in \cite{Ka02}. 
The analogous result at the edge of the spectrum was later obtained in \cite{AP14}. 
(We also remark that a mapping between the elliptic GinSE with a fermion field
theory was suggested by Hastings \cite{Ha00}.)
Fairly recently, it was shown in \cite{AEP22} that for a fixed $\tau$ (also called the regime of strong non-Hermiticity), the universal scaling limit \eqref{kappa S bulk} appears at the origin. 
Let us mention that the analysis in \cite{Ka02,AP14} was based on proper Riemann sum approximations, whereas a double contour integral representation was used in \cite{AEP22}. 
These methods provide a short way to find an explicit formula of the limiting pre-kernel, but it is not easy to perform the asymptotic analysis in a more general setup or to precisely control the error term. 
For these purposes, extending Proposition~\ref{Prop_CDI g}, the idea of using a proper differential equation was established in \cite{BE22,Eb21}, which reads as follows.

\begin{proposition}\label{Prop_CDI e}
 We have 
 \begin{align}
\partial_z \kappa_N^{ \rm g }(z,w)& = \frac{2z}{1+\tau} \, \kappa_N^{ \rm g }(z,w)  +  \frac{2}{\pi(1-\tau^2)^{3/2}}   \sum_{k=0}^{2N-1}  \frac{ (\tau/2)^{k} }{k!}   H_{k}\Big( \frac{z}{\sqrt{\tau}} \Big) H_{k}  \Big( \frac{w}{\sqrt{\tau}} \Big) \nonumber
\\
&\quad -  \frac{2}{\pi (1-\tau^2)^{3/2}}  \frac{ (\tau/2)^{N} }{(2N-1)!!}  H_{2N} \Big( \frac{z}{\sqrt{\tau}} \Big)  \sum_{l=0}^{N-1}  \frac{(\tau/2)^l}{(2l)!!} H_{2l} \Big( \frac{w}{\sqrt{\tau}} \Big) . 
 \label{CDI for kappa e}
 \end{align}
\end{proposition}
\begin{proof}(Sketch)
The general idea to derive such identity is the same as that used in the proof of Proposition~\ref{Prop_CDI g}; differentiate the pre-kernel and properly rearrange the indices in the summations to extract the pre-kernel itself multiplied by $z$ 
up to proportionality, and then collect all the remaining additive terms.
Contrary to the proof of Proposition~\ref{Prop_CDI g}, the well-known functional relations
 \begin{equation} \label{three term}
H_j'(z)=2j H_{j-1}(z), \qquad 
 H_{j+1}(z)=2zH_j(z)-H_j'(z)
\end{equation}
of the Hermite polynomials are crucially used in the computations and we refer to \cite{BE22} for more details. 
\end{proof}

We now bring to attention the fact that the first inhomogeneous term in \eqref{CDI for kappa e} corresponds to the kernel of the complex elliptic Ginibre ensemble with $N \mapsto 2N$; see \cite[Prop.~2.7]{BF22a}.
(A similar feature for the elliptic GinOE is highlighted above Proposition~\ref{pp2}.)
As will be discussed below, such a relation can be observed in further extensions to GinSE. 
We also refer the reader to \cite{AFNM00} for a similar relation for the Hermitian random matrix models. 

Using Proposition~\ref{Prop_CDI e}, one can derive the scaling limits of the elliptic GinSE correlation functions in various regimes. 
For $\tau$ fixed, it was shown in \cite{BE22} that in the real bulk and edge of the spectrum the universal scaling limits \eqref{kappa S bulk} and \eqref{kappa S edge} arises. 
Furthermore, in the edge scaling limits, as a counterpart of \cite[Prop.~2.8]{BF22a}, the subleading correction term was derived. 
In the weakly non-Hermitian regime when $\tau=1-\alpha^2/N$, the bulk and edge scaling limits were obtained in \cite{BES23}, extending previous results \cite{Ka02,AP14}; see also \cite{Eb21}.  
In particular, it was shown in \cite{BES23} that the limiting pre-kernel in the bulk scaling limit is of the unified Wronskian form \eqref{kappa Wronskian} with 
 \begin{equation} \label{fE AH bulk}
    f_z(u):=\frac{1}{2\pi} \int_{-C(\alpha)}^{C(\alpha)} e^{-t^2/2} \sin(2t(z-u))\,\frac{dt}{t}, \qquad E:=\R, 
\end{equation} 
where $C(\alpha)$ is an explicit constant depending on $\alpha$ and the position we zoom the point process. 
In the same spirit, the edge scaling limit is again of the form \eqref{kappa Wronskian} with  
\begin{equation}  \label{fE AH edge}
    	f_z(u):=2\alpha\int_{0}^{u} e^{ \alpha^3(z-t)+\frac{\alpha^6}{12} } {\rm{Ai}}\Big(2\alpha(z-t)+\frac{\alpha^4}{4}\Big)\,dt, \qquad E:=(-\infty,0).
    \end{equation} 

\begin{remark}
Note that the bulk scaling limit \eqref{fE AH bulk} has again the translation invariance. 
Conversely, it was shown in \cite{ABK22} that if a scaling limit of \eqref{Gibbs symplectic} satisfies the translation invariance, then it is of the form \eqref{fE AH bulk}.  
The main idea for this characterisation was a use of Ward's identity for the ensemble \eqref{Gibbs symplectic}, which says that $\mathbb{E}_N W_N^+[ \psi ]=0$, where $\psi$ is a test function and
\begin{equation} \label{Ward} 
 W_N^+[\psi]:= \sum_{ j \not =k } \psi(z_j) \Big(\frac{ 1 }{ z_j-z_k }+\frac{1}{z_j-\bar{z}_k} \Big) +2 \sum_{j=1}^N \frac{\psi(z_j)}{z_j-\bar{z}_j} - 2 \sum_{ j=1 }^N [ \partial_z W \cdot \psi ](z_j) + \sum_{j=1}^N \partial \psi(z_j).
\end{equation}    
\end{remark}

\subsection{Partition functions and gap probabilities}

Recall that the normal matrix ensemble \eqref{Gibbs complex} forms a determinantal point process. 
This integrable structure allows an explicit expression of the partition function; see \cite[Eq.~(5.15)]{BF22a}. 
Similarly, using the Pfaffian structure \eqref{rhoNk Pf} and de Bruijn type formulas \cite{dB55}, one can express the partition function $Z_N^{ \mathbb{H} }(W)$ in terms of the skew norms \eqref{SOP S} as 
\begin{equation} \label{ZN S SOP}
Z_N^{ \mathbb{H} }(W)= \prod_{j=0}^{N-1} r_j; 
\end{equation}
see e.g. \cite[Remark 2.5]{AEP22}.
For instance, for the elliptic GinSE, it follows from \eqref{SOP eGinibre} that the associated partition function $Z_N^{ \mathbb{H} }(W^{\rm e})$ is given by
\begin{equation}
 \label{ZN S e}
Z_N^{ \mathbb{H} }(W^{\rm e})  = \frac{((1-\tau)\sqrt{1-\tau^2} \pi )^N}{ 2^{N^2} } \prod_{k=0}^{N-1} (2k+1)!.   
\end{equation}
This explicit expression leads to the following asymptotic expansion; cf.~\cite[Prop.~4.1]{BF22a} for its counterpart for $Z_N^{ \mathbb{C} }(W^{\rm g})$.

\begin{proposition}
We have 
\begin{align}
   \log  Z_N^{ \mathbb{H} }(W^{\rm e})  &= N^2 \log N - \frac32 N^2 +\frac12 N \log N +\Big( \frac{\log( 4\pi^{3} (1-\tau)^{3}(1+\tau) )}{2} -\frac12 \Big) N  \nonumber
   \\
   &\quad -\frac1{24}\log N + \frac{5\log 2}{24} +\frac1{2} \zeta'(-1)-\frac{1}{48N}-\frac{1}{1920N^2}+O(\frac{1}{N^3}). \label{ZN S e asym}
\end{align}
In particular, we have 
\begin{align}
   \log  Z_N^{ \mathbb{H} }(NW^{\rm g})  &= - \frac32 N^2 -\frac12 N \log N +\Big( \frac{\log( 4\pi^{3} )}{2} -\frac12 \Big) N  \nonumber
   \\
   &\quad -\frac1{24}\log N + \frac{5\log 2}{24} +\frac1{2} \zeta'(-1)-\frac{1}{48N}-\frac{1}{1920N^2}+O(\frac{1}{N^3}). \label{ZN S g asym}
\end{align}
\end{proposition}
\begin{proof}
One can rewrite \eqref{ZN S e} in terms of the Barnes $G$-function as 
\begin{align*}
Z_N^{ \mathbb{H} }(W^{\rm e}) = \frac{((1-\tau)\sqrt{1-\tau^2} \pi )^N}{ 2^{3N^2/2-N/2} } \frac{G(2N+1)}{ G(N+1) } = ((1-\tau)^3(1+\tau) \pi )^{N/2} G(N+1) \frac{ G(N+\frac32) }{G(\frac32)}. 
\end{align*}
Then the asymptotic behaviour \eqref{ZN S e asym} follows from the knowledge of the known asymptotic expansion of the $G$-function (see e.g.~\cite[Th.~1]{FL01}).
The second expansion \eqref{ZN S g asym} immediately follows from  \eqref{ZN S e asym} with $\tau=0$, where the additional difference $(N^2+N) \log N$ is due to the simple scaling $z_j \mapsto \sqrt{N} z_j.$
\end{proof}

We now discuss the asymptotics of the partition functions in a more general setup. 
For a fixed $Q$, the aymptotic expansion of the partition function $Z_N^{ \mathbb{C} }(NQ)$ in \eqref{Gibbs complex} was discussed in \cite[\S5.3]{BF22a} in detail.
For radially symmetric potentials, the use of \eqref{ZN S SOP} and \eqref{skew op_rad} was made in \cite{BKS22} to show that
\begin{align}
\log Z_N^{ \mathbb{H} }(NQ)  & =-2N^2 I_Q[\mu_Q] - \frac{1}{2}N\log N  \nonumber
\\
&\quad + \Big( \frac{\log(4\pi^2)}{2}- \frac12 E_Q[\mu_Q] - U_{\mu_Q}(0) \Big) \, N - \frac{\chi}{24}\log N +O(1),  \label{ZN symp exp}
\end{align}
where 
\begin{equation}
  E_Q[\mu_Q] := \int_{\C} \mu_Q(z)\,\log \mu_Q(z) \, d^2z
\end{equation} 
is entropy of the equilibrium measure $\mu_Q$. 
Here $\chi$ is the Euler index of the droplet; for instance $\chi=1$ for the disk and $\chi=0$ for the annulus. 
Compared to the expansion of $Z_N^{\mathbb{C}}(NQ)$ given in \cite[Eq.(5.17)]{BF22}, 
a notable difference is the additional $U_{\mu_Q}(0)$ in the $O(N)$ term, which is the logarithmic potential 
\begin{equation} \label{logarithmic potential}
U_\mu(z) = \int \log\frac1{|z-w|}\, d\mu(w) 
\end{equation}
evaluated at the origin.
This term is closely related to the notion of renormalised energy of the Hamiltonian \eqref{Ham S}; cf. \cite{LS18}. 

For a radially symmetric potential $q(r)=Q(|z|=r)$ with the droplet specified by the radii \eqref{droplet annulus}, we have
\begin{equation} \label{energy radially sym}
I_Q[\mu_Q]=q(R_1)-\log R_1 -\frac14 \int_{R_0}^{R_1} rq'(r)^2\,dr, \quad  U_{\mu_Q}(0) = - \log R_1 + \frac{q(R_1)-q(R_0)}{2}.  
\end{equation}
This gives that for $Q=W^{ \rm g }$, 
\begin{equation}
I_Q[\mu_Q]= \frac34, \qquad U_{ \mu_Q }(0)= \frac12, \qquad E_Q[\mu_Q]= -\log \pi. 
\end{equation}
Substituting these in the formula \eqref{ZN symp exp} reclaims \eqref{ZN S g asym}.

We now discuss a relation between $Z_N^{\mathbb{C}}$ and $Z_N^{\mathbb{H}}$.
\begin{proposition}\label{Prop_Z relations}
For a radially symmetric potential $W_{ \mathbb{C} }(z,\bar{z}) \equiv \omega_{ \mathbb{C} }(|z|)$, let 
\begin{equation} \label{W relations CS}
W_{ \mathbb{H} }(z,\bar{z}) \equiv \omega_{ \mathbb{H} }(|z|) :=  \frac12 \omega_{ \mathbb{C} }(|z|^2).
\end{equation}
Then we have 
\begin{equation}
    Z_N^{ \mathbb{H} }( W_{ \mathbb{H} } )= Z_N^{ \mathbb{C} }( W_{ \mathbb{C} } ). 
\end{equation}
\end{proposition}

\begin{proof}
By the change of variables, 
 \begin{align*}
4  \int_0^\infty r^{4k+3} e^{ -2\omega_{ \mathbb{H} } (r) } \,dr  =  2  \int_0^\infty r^{2k+1} e^{ -2\omega_{ \mathbb{H} } ( \sqrt{r} ) } \,dr= 2  \int_0^\infty r^{2k+1} e^{ -\omega_{ \mathbb{C} } ( r ) } \,dr. 
 \end{align*}
 Then result now follows from Proposition~\ref{Prop_SOP for radial}, \cite[Eq.~(5.15)]{BF22a} and \eqref{ZN S SOP}.
\end{proof}

As an example, note that by \cite[Eq.~(5.15)]{BF22a}, we have 
\begin{equation}
Z_N^{ \mathbb{C} }( 2|z|  )= \prod_{k=0}^{N-1} 2\pi\int_0^\infty r^{2k+1}e^{-2r}\,dr= \prod_{k=0}^{N-1} \frac{(2k+1)!}{2^{2k+1}}\pi. 
\end{equation}
Then one can observe that this coincides with \eqref{ZN S e} with $\tau=0.$
We also mention that if the droplet $S_{\mathbb{H}}$ associated with $W_{ \mathbb{H} }$ is 
\begin{equation}
S_{\mathbb{H}}=\{ z \in \mathbb{C}: R_0 \le |z| \le R_1 \},
\end{equation}
then its counterpart $ S_{ \mathbb{C} } $ for $W_{ \mathbb{C} }$ is given by 
\begin{equation}
S_{\mathbb{C}}=\{ z \in \mathbb{C}: R_0^2 \le |z| \le R_1^2 \}. 
\end{equation}
This can be directly checked using \eqref{droplet annulus}.
Such a relation holds in general beyond the radially symmetric potentials; see \cite[Lem.~1]{BM15}.

As a consequence of Proposition~\ref{Prop_Z relations}, one can obtain various statistics of the ensemble \eqref{Gibbs symplectic} from the analogous results for \eqref{Gibbs complex}. 
To be more concrete, let us focus on the gap probabilities.

\begin{proposition} \label{Prop_gap prob rel}
For a radially symmetric domain $D$, let $E_N^{ W_\mathbb{H} }(0;D)$ be the probability that the ensemble \eqref{Gibbs symplectic} with a potential $W_{\mathbb{H}}$ has no particle inside $D$.
We define $E_N^{ W_\mathbb{C} }$ in a same way for \eqref{Gibbs complex}. 
Then we have 
\begin{equation}
E_N^{ W_\mathbb{H} }(0;D)= E_N^{ W_\mathbb{C} }(0;\tilde{D}),
\end{equation}
where $\tilde{D}$ is the image of $D$ under the map $z \mapsto z^2$. 
\end{proposition}
\begin{proof}
Let us write
\begin{align*}
W_{ \mathbb{H},D }(z,\bar{z}):= \begin{cases}
  W_{ \mathbb{H} } (z,\bar{z}) &\textup{if } z \in D^c,
  \\
  +\infty &\textup{otherwise},
\end{cases}, \qquad W_{ \mathbb{C},\tilde{D} }(z,\bar{z}):= \begin{cases}
  W_{ \mathbb{C} } (z,\bar{z}) &\textup{if } z \in \tilde{D}^c,
  \\
  +\infty &\textup{otherwise}.
\end{cases}
\end{align*}    
Then by Proposition~\ref{Prop_Z relations}, 
\begin{equation}
E_N^{ \mathbb{H} }(0;D)= \frac{ Z_N^{\mathbb{H}}(W_{ \mathbb{H},D }) }{ Z_N^{\mathbb{H}}(W_{ \mathbb{H} })  }= \frac{ Z_N^{\mathbb{C}}(W_{ \mathbb{C},\tilde{D} }) }{ Z_N^{\mathbb{C}}(W_{ \mathbb{C} })  } = E_N^{ W_\mathbb{C} }(0;\tilde{D}),
\end{equation}
which completes the proof. 
\end{proof}

This proposition is particularly helpful in the context of the Mittag-Leffler ensembles.
They are two-parameter generalisations of the GinU/SE for which the associated potential is of the form
\begin{equation} \label{W ML}
\omega^{\rm ML}(|z|)=|z|^{2b}-2\alpha \log |z|, \qquad  b>0, \quad \alpha>-1. 
\end{equation}
For the complex Mittag-Leffler ensemble, the precise asymptotic behaviours of the gap probabilities were obtained in \cite{Ch21}; cf. see \cite[\S3.2]{BF22a} for a summary and further references for the GinUE case when $\alpha=0, b=1$.  
Then as a consequence of Proposition~\ref{Prop_gap prob rel}, the analogous results for the symplectic Mittag-Leffler ensemble immediately follows. 
In particular, the gap probabilities of the GinSE can be obtained from the result of complex Mittag-Leffler ensemble with $\alpha=0, b=1/2$. 
(See also \cite{APS09} for an earlier work.)
Beyond the gap probabilities, Proposition~\ref{Prop_Z relations} can be used to investigate various counting statistics \cite{BC22a,Ch22,ABE23} as well as fluctuations of the maximal modulus \cite{Ri03a,Du21}.    

\begin{remark}
The Neumann boundary conditions disk Coulomb gas with Boltzmann factor (\ref{im}) is exactly solvable for $\beta = 2$ \cite{Sm82}. In distinction to the expansion (\ref{ZN S g asym}), it
is found that the large $N$ form of the logarithm of the partition function is at order $N$ and order $\log N$ the same for the GinUE \cite[Eq.~(4.3)]{BF22a}, although now there is also an $O(\sqrt{N})$ surface tension term \cite[Eq.~(3.42)]{Te01}.
Let us also mention that a generalisation to a two-component Coulomb gas and its Pfaffian structure have been studied in \cite{JS01}. 
\end{remark}

\subsection{Singular values}
Consider a $p \times N$ ($p \ge N$) rectangular GinSE matrix $X$. As a complex matrix, $X$ is of size $2n \times 2N$. Forming $X^\dagger X$ gives the well known construction of quaternion Wishart matrices \cite[\S 3.2.1]{Fo10}.
The $2N$ eigenvalues of $X^\dagger X$ --- which are the square singular values of $X$ --- are doubly degenerate. Let the $N$ independent eigenvalues be denoted $\{s_j\}$. Their PDF is proportional to (see \cite[Prop.~3.2.2]{Fo10})
\begin{equation}\label{L4}
\prod_{l=1}^N s_l^{2(p-N) + 1} e^{-2s_l}
\prod_{1 \le j < k \le N} (s_k - s_j)^4.
\end{equation}

Upon the global scaling $X^\dagger X \mapsto {1 \over N} X^\dagger X$, and with $p = \alpha N$ ($ \alpha > 1$), as for (\ref{3.18}) the smallest and largest tend almost surely to $(\sqrt{\alpha \mp 1}+1)^2$, as for the real case. This coincidence can be understood as being a consequence of $\{s_j\}$ in both (\ref{L1}) and (\ref{L4}) as being well approximated by the zeros of the Laguerre polynomials $L_N^{p-N}(px)$ \cite{DI07,HH22}. As in the real case, we thus have that the condition number tends to a constant in the circumstance that $\alpha > 1$. On the other hand, this breaks down for $\alpha = 1$. Specifically, we will consider the square case $p=N$. It seems that the limiting condition number has not previously been reported in the literature. The starting point is to calculate the PDF for the smallest eigenvalue. This  is equal to the differentiation operation $-{d \over ds}$ applied to the 
 the gap probability $E_N^{\rm q W}(0,(0,s))$ of their being no eigenvalues from the origin to a point $s$. 
 The latter is defined by integrating the PDF (\ref{L4}) over $s_l \in (s,\infty)$, $(l=1,\dots,N)$. Changing variables $s_l \mapsto s_l+s$ then 
 shows
$$
E_N^{\rm q W}(0,(0,s)) = {e^{-2Ns} \over C_N}
 \int_0^\infty ds_1 \cdots \int_0^\infty ds_N \, \prod_{l=1}^N (s + s_l) e^{-2s_l} \prod_{1 \le j < k \le N} (s_k - s_j)^4,
$$
where $C_N$ is such that the LHS equals unity for $s=0$. According to \cite[Eq.~(13.44) with $a=0,m=1,\beta=4,t_1=-s$]{Fo10}, the normalised multiple integral is equal to the hypergeometric polynomial ${}_1 F_1(-N;1/2;-s)$, which in turn is proportional to the Laguerre polynomial $L_N^{(-1/2)}(-s)$. 
From the confluent limit of the hypergeometric polynomial, we conclude the simple result
\begin{equation}\label{L4a}
\lim_{N \to \infty} E_N^{\rm q W}(0,(0,x/N)) =e^{-2x} {}_0F_1(1/2;x) = e^{-2x} \cosh(2 \sqrt{x});
\end{equation}
this is equivalent to \cite[Eq.~(2.15a)]{Fo94}.
Consequently the scaled condition number $\kappa_N /(2N)$ has the limiting PDF
$$
- {2 \over y^3} {d \over dx} \Big (
 e^{-2x} \cosh(2 \sqrt{x}) \Big )
\Big |_{x=1/y^2}.
$$

In keeping with our previous discussion relating to singular values, we record here too the explicit form of the distribution of $|\det X|^2$ for $X$ a square GinSE matrix \cite[Prop.2 with $\beta = 4, \sigma_l^2 = 1/4$]{FZ18},
\begin{equation}\label{3.18b}
|\det X|^2 \mathop{=}\limits^{\rm d} \prod_{j=1}^N {1 \over 4} \chi_{4j}^2.
\end{equation}
Here we have defined 
$|\det X|^2 = \prod_{l=1}^Ns_l$, even though the eigenvalues of $X$ are doubly degenerate.
One approach to the derivation of (\ref{3.18b}) is to make use of
knowledge of the joint PDF (\ref{L4}) and the Laguerre weight version of Selberg's integral to compute first the Mellin transform of the distribution; recall \cite[Eq.~(6.9)]{BF22a}.

\section{Further extensions to GinSE}\label{S6}

\subsection{Common structures}

The eigenvalue PDF of several extensions to GinSE discussed in this section is of the form \eqref{Gibbs symplectic} with a radially symmetric potential $W(z,\bar{z})=\omega(|z|)$. 
Before moving on to each example, let us draw some common structures. 

We first note the following consequence of Proposition~\ref{Prop_SOP for radial} and \eqref{kappaN skewOP}.

\begin{proposition} \label{Prop_kappa structure}
 For a radially symmetric potential $W(z,\bar{z})=\omega(|z|)$  the associated pre-kernel is given by 
 \begin{equation}
\kappa_N^{\rm s}(z,w)=\frac{1}{\pi} \sum_{ 0 \le l \le  k \le N-1 } a_{2k+1} a_{2l} (z^{2k+1}  w^{2l}-z^{2l}  w^{2k+1}  ), 
\end{equation}
where 
\begin{equation}
a_{2k+1}  := \frac{1}{\sqrt{2}} \frac{ h_{2k} }{ h_{2k+1} } \frac{ h_{2k-2} }{ h_{2k-1} } \cdots \frac{h_0}{h_1}, \qquad   a_{2l} :=  \frac{1}{\sqrt{2}}  \frac{h_{2l-1} }{ h_{2l} } \frac{ h_{2l-3} }{ h_{2l-2} } \cdots \frac{h_1}{h_2} \frac{1}{h_0}.
\end{equation}
\end{proposition}

Due to Proposition~\ref{Prop_kappa structure}, we have 
\begin{align} \label{1 pt function radial}
\rho_{(1),N}^{\rm s}(z)  
= \frac{e^{-2 \omega(|z|) }}{\pi}  \sum_{ 0 \le l \le  k \le N-1 }  a_{2k+1} a_{2l}   ( z^{2k+1}\bar{z}^{2l+1}-z^{2l} \bar{z}^{2k+2}- z^{2k+2}\bar{z}^{2l} + z^{2l+1} \bar{z}^{2k+1}  ). 
\end{align}
The asymptotic behaviour of $\rho_{(1),N}^{\rm s}(z)$ in the global scaling $W=NQ$ can be deduced from the Coulomb gas approach previously discussed in \S~\ref{Section Coulomb gas S}, and so is independent of knowledge of (\ref{1 pt function radial}). It gives
\begin{equation} \label{global density s}
\lim_{N \to \infty} \frac{\rho_{(1),N}^{\rm s}(z)}{N} \Big|_{W=NQ} = \frac{ \partial_z \partial_{ \bar{z} } Q(z)}{\pi} \chi_{ z\in S_Q },
\end{equation}
where we recall that $S_Q$ is the droplet.

As a feature dependent on 
(\ref{1 pt function radial}),
we now consider the radial density 
\begin{equation} \label{radial density}
\widehat{\rho}_{(1),N}^{\rm{s}} (r):= \frac{1}{2\pi} \int_0^{2\pi} \rho_{(1),N}^{\rm s}( r e^{i\theta} )\,d\theta = \frac{ e^{-2 \omega(r) } }{ \pi }\sum_{k=0}^{N-1} \frac{r^{4k+2}}{h_{2k+1}}. 
\end{equation}
Here the second identity readily follows from \eqref{1 pt function radial} and $2a_k a_{k+1}=1/h_{k+1}$. This same formula, with $2k+1$ in the summation replaced by $k$,
holds for the radial density 
of the normal matrix ensemble \eqref{Gibbs complex}, due to the 
determinantal structure \cite[Eq.~(5.18)]{BF22a}. 
A consequence of \eqref{radial density} is that the expected number of eigenvalues $E_N^{\rm s}(R)$ in the disk of radius $R>0$ can be expressed as 
\begin{equation}\label{6.6}
E_N^{\rm s}(R) =2\pi \sum_{ k=0 }^{N-1} \frac{ \int_0^R r^{2k+1} e^{-2\omega(r)}\,dr}{ h_{2k+1} }; 
\end{equation}
see \cite[Prop.~1.1]{ABE23} for a similar expression for the number variance. Furthermore, (\ref{6.6}) is consistent with the fact that in the case of a general radially symmetric potential, the joint density of moduli of the eigenvalues forms a permanental process; see e.g. \cite{AIS14}.

In the asymptotic analysis of the pre-kernel, the differential equations (\ref{CDI for kappa g v2}) and (\ref{CDI for kappa e}) have been key. This technique will be shown to be of further utility in the study of the extensions of GinSE considered below.

\subsection{Induced GinSE}
The induced GinSE can be constructed in a similar way described in \S~\ref{Section_iGinOE}. 
For this, we begin with an $n \times N$ ($n \ge N$) Gaussian random matrix $G$ with independent real quaternion elements.  
Then we define $\tilde{G}=(G^\dagger G)^{1/2} U$, where $U$ is a Haar unitary matrix with elements from the real quaternion field  (see e.g.~\cite{DF17}). 
Then letting $L:=n-N$, the element distribution on $N \times N$ matrices is proportional to 
\begin{equation}
 (\det \tilde{G}^\dagger \tilde{G})^{2L} e^{-2{\rm Tr} \, \tilde{G}^\dagger \tilde{G}};
\end{equation}
see \cite[\S~2.4]{Fo16}. 
Its eigenvalue PDF follows \eqref{Gibbs symplectic} with 
\begin{equation} \label{W iG}
W^{\rm i}(z,\bar{z})= |z|^2-2 L \log|z|. 
\end{equation}
This can be regarded as a special case of \eqref{W ML}. 
We remark here that when we consider the model as a Coulomb particle system
governed by the law \eqref{Gibbs symplectic}, we can allow $L$ to be an
arbitrary real number as long as $L>-1$. 
Let us also mention that the system \eqref{Gibbs symplectic} with \eqref{W iG} can be interpreted as a distribution of $N$ random eigenvalues of $(N+L) \times (N+L)$ GinSE, conditioned to have zero eigenvalues with multiplicity $L$.

We first discuss the global scaling $z \mapsto \sqrt{N}z$ in the regime $L=N\alpha$. 
As discussed in \S\ref{Section Coulomb gas S}, the empirical measure of the eigenvalues converges to the equilibrium measure $\mu_{ Q^{\rm i} }$, where 
\begin{equation}
Q^{\rm i}(z)=|z|^2-2\alpha \log |z|.
\end{equation}
Then it follows from \eqref{eq msr form} and \eqref{droplet annulus} that
\begin{equation} 
\lim_{N \to \infty}
\rho_{(1),N}^{\rm i}(\sqrt{N}z) = {1 \over {\pi}} \Big ( \chi_{|z|<\sqrt{\alpha +1}} - \chi_{|z|<\sqrt{\alpha}} \Big ),
\end{equation}
which agrees with the limiting density \eqref{rttC+} for the induced GinOE.  
Notice that for $\alpha=0$, we recover the circular law. 

Turning to the higher correlation functions, we first note that due to Proposition~\ref{Prop_SOP for radial}, the associated skew orthogonal polynomials are given by
\begin{equation} \label{SOP iGinibre}
 q_{2k+1}^{ \rm i }(z)=z^{2k+1}, \qquad q_{2k}^{ \rm i }(z)= \sum_{l=0}^k \frac{ \Gamma(k+1+L) }{ \Gamma(l+1+L) }z^{2l}, \qquad r_k^{\rm i}= \frac{ \Gamma(2k+2+2L) }{2^{2k+1+2L}}\pi;    
\end{equation}
cf. \eqref{SOP Ginibre}. 
Combining this with \eqref{kappaN skewOP} gives rise to the pre-kernel
\begin{align}
\kappa_N^{ \rm i }(z,w)&= \frac{ 2^{2L+1/2}  }{ \pi } \sum_{k=0}^{N-1} \frac{( \sqrt{2}z )^{2k+1}}{(2k+1+2L)!!}    \sum_{l=0}^k \frac{( \sqrt{2} w)^{2l} }{(2l+2L)!!} \nonumber
\\
&\quad - \frac{ 2^{2L+1/2}  }{ \pi }  \sum_{k=0}^{N-1} \frac{( \sqrt{2} w)^{2k+1}}{(2k+1+2L)!!}    \sum_{l=0}^k \frac{( \sqrt{2} z)^{2l} }{(2l+2L)!!} . \label{kappaN ig}
\end{align}
Taking into consideration of the integrable structure, we further define the transformation
\begin{equation}
\widehat{\kappa}_N^{ \rm i }(z,w)=(z w)^{2L}\kappa_N^{ \rm i }(z,w), \qquad \widehat{ \mathcal{K} }_N^{\rm i}(z,w)= e^{ -|z|^2-|w|^2 }  
\begin{bmatrix} 
\widehat{\kappa}_N^{\rm i}(z,w) & \widehat{\kappa}_N^{\rm i}(z,\bar{w})
\smallskip 
\\
\widehat{\kappa}_N^{\rm i}(\bar{z},w) & \widehat{\kappa}_N^{\rm i}(\bar{z},\bar{w}) 
\end{bmatrix}. 
\end{equation}
Then by Proposition~\ref{Prop_Pf S} together with basic properties of the Pfaffian, we have
\begin{equation}
\rho_{(k),N}^{\rm i}(z_1,\dots, z_k) =\prod_{j=1}^{k} (\overline{z}_j-z_j)  {\rm Pf} \, [ \widehat{ \mathcal{K} }_N^{\rm i}(z_j,z_l)]_{j,l=1,\dots,k}.
\end{equation}
The following was found in \cite{BC22}. 

\begin{proposition} \label{Prop_CDI i}
We have 
\begin{align}
\partial_z \widehat{\kappa}_N^{ \rm i }(z,w) =   2 z \, \kappa_N^{\rm i}(z,w)  &+ \frac{2}{\pi}   e^{2zw} \Big( \frac{\Gamma(2L+2N;2 z w)}{ \Gamma(2L+2N) }-\frac{ \Gamma( 2L;2zw ) }{  \Gamma(2L) } \Big)   \nonumber
\\
& -\frac{2^{3/2}}{\pi} \frac{ z^{2N+2L}}{ \Gamma(N+L+1/2) }e^{w^2} \Big( \frac{ \Gamma(N+L;w^2) }{ \Gamma(N+L) }-  \frac{ \Gamma(L;w^2) }{ \Gamma(L) } \Big)   \nonumber
\\
& -\frac{2}{\sqrt{\pi}} \frac{ z^{2L-1}}{ \Gamma(L) } e^{w^2} \Big( \frac{\Gamma(N+L+1/2;w^2) }{ \Gamma(N+L+1/2) }- \frac{ \Gamma(L+1/2;w^2) }{ \Gamma(L+1/2) } \Big).    \label{CDI for kappa i v2}
\end{align}
\end{proposition}

Note that Proposition~\ref{Prop_CDI i} with $L=0$ reduces to Proposition~\ref{Prop_CDI g}.
For $L=0$, the last term in the RHS of \eqref{CDI for kappa i v2} vanishes since $\lim_{L \to 0} 1/\Gamma(L)=0$.
As previously remarked below Proposition~\ref{Prop_CDI e}, one can notice again that the first inhomogeneous term in \eqref{CDI for kappa i v2} coincides with the kernel of the induced GinUE with $N \mapsto 2N$ up to a weight function; see \cite[\S2.4]{BF22a}.

As in \S\ref{Section_correlation GinSE}, the uniform asymptotic expansion \eqref{Tr} can be used to derive various scaling limits. 
It then follows that for $L=\alpha N$ with $\alpha>0$, the universal scaling limits \eqref{kappa S bulk} and \eqref{kappa S edge} appear in the bulk and edge of the spectrum. 
Another interesting regime is the case when $L$ is fixed while $N\to\infty$. 
In this case, the scaling limit at the origin should be treated separately since the potential \eqref{W iG} reveals a conical singularity.
This was investigated in \cite{ABK22}. 

\begin{proposition}
For a fixed $L > -1$, we have    
\begin{equation}
\lim_{N \to \infty} \rho_{(k),N}^{ \rm i} (z_1,\dots, z_k)  =\prod_{j=1}^{k} (\overline{z}_j-z_j)  {\rm Pf} \, [\mathcal K_\infty^{ {\rm s}, L }(z_j,z_l)]_{j,l=1,\dots,k}, 
\end{equation}
where
\begin{equation}
\mathcal{K}_{\infty}^{ {\rm s}, L }(z,w):= e^{-|z|^2-|w|^2} \begin{bmatrix} 
\kappa_\infty^{ {\rm s}, L }(z,w) & \kappa_\infty^{ {\rm s}, L }(z,\bar{w})
\smallskip 
\\
\kappa_\infty^{ {\rm s}, L }(\bar{z},w) & \kappa_\infty^{ {\rm s}, L }(\bar{z},\bar{w}) 
\end{bmatrix}.
\end{equation}
Here, 
	\begin{equation}  \label{kappa insertion}
		\kappa_\infty^{ {\rm s}, L }(z,w)= \frac{2}{\pi} (2zw)^{2L} \int_0^1 s^{2L} \Big(ze^{(1-s^2)z^2}-we^{(1-s^2)w^2}\Big) E_{2,1+2L}((2szw)^2)\,ds,	
	\end{equation}
 where $E_{a,b}(z):=\sum_{k=0}^\infty z^k/\Gamma(ak+b)$ is the two-parametric Mittag-Leffler function. 
 This can be rewritten in terms of the incomplete gamma function as 
 \begin{align}
   \kappa_\infty^{ {\rm s}, L }(z,w)&=\frac{1}{\pi} \int_0^1 (ze^{(1-s^2)z^2}-we^{(1-s^2)w^2}) \nonumber
   \\
   &\qquad \times \Big( e^{2szw} \frac{\gamma(2L;2szw)}{\Gamma(2L)}+(-1)^{-2L} e^{-2szw} \frac{\gamma(2L;-2szw)}{\Gamma(2L)} \Big)\,ds. 
 \end{align}
\end{proposition}

We mention that if $2L$ is a non-negative integer, the pre-kernel $\kappa_\infty^{ {\rm s}, L }$ can also be expressed in terms of error functions. 
(Cf. See also \cite{AB07} for a different way to express the correlation functions as a ratio of certain Pfaffians.)
Namely, if $L$ is a non-negative integer, we have 
\begin{align} 
			\kappa_\infty^{ {\rm s}, L }(z,w)&=\frac{ e^{z^2+w^2}  }{\sqrt{\pi}}  {\rm{erf}}(z-w)	+\frac{1}{\pi}\sum_{ 0 \le l < k \le L-1 }  \frac{ w^{2k} z^{2l+1}-z^{2k} w^{2l+1}   }{k! (1/2)_{l+1} } \nonumber
			\\
			&\quad +\frac{  e^{w^2} }{\sqrt{\pi}} {\rm{erf}}(w) \sum_{k=0}^{L-1} \frac{ z^{2k} }{k!}  -\frac{ e^{z^2} }{\sqrt{\pi}}  {\rm{erf}}(z)   \sum_{k=0}^{L-1}  \frac{ w^{2k} }{k!}. \label{kappa L even}
	\end{align}
Note that if $L=0$, this recovers \eqref{kappa S bulk}, where we have used the convention that the summation with an empty index equals zero.
Here $(a)_n=a(a+1)\cdots (a+n-1)$ is Pochhammer's symbol.
On the other hand if $L-1/2$ is a non-negative integer, we have
	\begin{align} 
	\kappa_\infty^{ {\rm s}, L }(z,w)&=\frac{ e^{z^2+w^2} }{\sqrt{\pi}}   (  {\rm{erf}}(z-w)-{\rm{erf}}(z) +{\rm{erf}}(w) )  +\frac{1}{\pi}\sum_{ 1 \le l < k \le L-1/2 } \frac{  w^{2k-1} z^{2l}-z^{2k-1} w^{2l}  }{l!(1/2)_k} \nonumber
			\\
			&\quad +\frac{ e^{w^2}-1 }{\pi}  \sum_{k=1}^{L-1/2}  \frac{ z^{2k-1}}{(1/2)_k } -\frac{ e^{z^2}-1 }{\pi}  \sum_{k=1}^{L-1/2}  \frac{ w^{2k-1}}{(1/2)_k}. \label{kappa L odd}
	\end{align}
  Setting $w=\bar{z}$ gives the density, in particular reclaiming \eqref{rho bulk Dawson} for $L=0$.
 We mention that other than $L=0$, the limiting $1$-point function is no longer translation invariant along the real axis.
 
Another interesting double scaling limit arises in the so-called almost-circular regime, where the spectrum tends to form a thin annulus of width $O(1/N)$. 
In this case, the scaling limits both at the real axis as well as away from the real axis were obtained in \cite{BC22}. 
For the former case, the limiting correlation functions are Pfaffians, which interpolates the bulk scaling limits of the GinSE and of the antisymmetric Gaussian Hermitian ensemble. 
For the latter case, the limiting correlation functions are determinants, and the kernel is the same as the one appearing in the bulk scaling limit of the weakly non-Hermitian elliptic GinUE \cite{FKS98,ACV18}.

Beyond the eigenvalue statistics, the diagonal and off-diagonal eigenvector overlaps of the induced GinSE were computed in \cite{AFK20}; see also \cite{Du21}. 
Contrary to the determinantal case \cite[\S6.4]{BC22a}, an integrable Pfaffian structure for finite $N$ has not been discovered.

\subsection{Spherical induced GinSE} \label{Section_SI GinSE}

Following \cite{FBKSZ12,MP17}, we first consider an $(N+L) \times N$ rectangular GinSE $X$ and an $N \times N$ Wishart matrix with quaternion entries $A$ (constructed as $A = B^\dagger B$ with $B$ an $n \times N$ ($n \ge N$) rectangular GinSE matrix) to introduce a particular $(N + L) \times N$ random matrix $Y=XA^{-1/2}$ with each entry itself a $2 \times 2$ matrix representation of a quaternion \eqref{1.1}.
In terms of such $Y$ together with a Haar distributed unitary random matrix $U$ with quaternion entries, we define $G = U (Y^\dagger Y)^{1/2}$. 
Then the matrix distribution of $G$ is given by
\begin{equation} \label{pdf of siGinSE}
 \frac{ \det (G G^\dagger )^{ 2L } }{  \det ( \mathbb{I}_N+G G^\dagger  )^{ 2(n+N+L) } }. 
\end{equation}
It was shown in \cite{Fi12,Ma13,MP17} that its eigenvalue PDF follows \eqref{Gibbs symplectic} with 
\begin{equation} \label{W si}
W^{ \rm si }( z,\bar{z} )= (n+L+1)\log(1+|z|^{2})-2L\log |z|.
\end{equation}
We refer to \cite[Appendix~A]{BF22} for a Coulomb gas picture of the spherical induced ensemble. 

With $L,M$ scaled with $N$ in a way that $L/N \to \alpha_1-1 \ge 0$ and $n/N \to \alpha_2 \ge 1$, 
\begin{equation} \label{W si asymp}
\frac{ W^{ \rm si }( z,\bar{z} ) }{ N } \sim (\alpha_1+\alpha_2-1) \log(1+|z|^2) -2(\alpha_1-1) \log |z|. 
\end{equation}
Then by \eqref{droplet annulus} one can specify inner and outer radii of the limiting spectrum, which leads to  
\begin{equation} \label{limiting density si}
\lim_{N \to \infty} \frac{1}{N} \rho_{(1),N}^{\rm si}(z) = {1 \over {\pi}}  \frac{ \alpha_1+\alpha_2-1 }{(1+|z|^2)^2} \Big (  \chi_{|z|<  \sqrt{\alpha_1/ (\alpha_2 - 1)} }- \chi_{|z|<\sqrt{(\alpha_1 -1)/\alpha_2} }  \Big ),
\end{equation}
where the limiting eigenvalue density is obtained by taking the Laplacian of the RHS of \eqref{W si asymp}; cf. \eqref{eq msr form}.  
This limiting distribution is in consistent with the spherical induced GinUE; see \cite[\S2.5]{BF22a}.
Due to the limiting form \eqref{limiting density si}, one again notices that the stereographically projected eigenvalues tend to uniformly occupy a spherical annulus, whence the name spherical.

Beyond the leading order asymptotic of the eigenvalue density, a fluctuation formula can be found \cite[Appendix B]{BF22}.

\begin{proposition}
Let $\Phi = \sum_{j=1}^N \phi(|z_j|)$ be a radially symmetric linear statistic, where $z_j$'s are drawn from the spherical induced ensemble as specified by \eqref{Gibbs symplectic} and \eqref{W si}. 
It is required that $\phi$ be smooth and subject to a growth condition as $|z | \to \infty$.
Then the corresponding characteristic function $\hat{P}_{N,B}(k)$ satisfies 
\begin{equation}
\log \hat{P}_{N,\Phi}(k) = i k \tilde{\mu}_{N,\Phi} - k^2 \tilde{\sigma}_B^2/2 + o(1),
\end{equation} 
 where 
\begin{equation}
\tilde{\mu}_{N,\Phi} =\frac{n + L}{\pi} \int_{ S^{\rm si} } {\phi(z) \over (1 + |z|^2)^2} \, d^2z , \qquad  \tilde{\sigma}_\Phi^2 = {1 \over 8 \pi } \int_{ S^{\rm si} } || \nabla  \phi(|z|) ||^2 \, d^2z.
 \end{equation} 
 Here $S^{\rm{si}}$ is the droplet specified in \eqref{limiting density si}.
\end{proposition}

As a consequence of this proposition, the asymptotic normality of the centred linear statistic $\Phi -  \tilde{\mu}_{N,\Phi}$ follows.
The variance is equal to one half times that for the GinUE in the case of  a radially symmetric test statistic; see \cite[Eq.~(3.21)]{BF22a}.

We now discuss the higher order correlation functions and their scaling limits. 
Due to the radial symmetry of the potential \eqref{W si}, the associated skew orthogonal polynomials can be again constructed by using Proposition~\ref{Prop_SOP for radial}. 
This in turn gives that \cite[Prop.~4]{Fo13a}
\begin{equation} \label{sop spherical}
q_{2k+1}^{ \rm si }(z)=z^{2k+1},\qquad	q_{2k}^{ \rm si }(z)= \frac{ \Gamma(k+L+1) } { \Gamma(k-n+\frac12) } \sum_{l=0}^{k} (-1)^{k-l}  \frac{ \Gamma(l-n+\frac12) }{ \Gamma(l+L+1) }  z^{2l}, 
\end{equation}
with the skew norm 
\begin{equation} \label{skew norm spherical}
r_k^{ \rm si }= \frac{ 2 \Gamma(2k+2L+2) \Gamma(2n-2k) }{ \Gamma(2n+2L+2) }.
\end{equation}
Now it follows from \eqref{kappaN skewOP} and basic functional identities of the gamma function that 
\begin{align}
\kappa_N^{ \rm si }(z,w)& =  \frac{  \Gamma(2n+2L+2) }{ 2^{2L+2n+1} } \sum_{k=0}^{N-1} \sum_{l=0}^{k} \frac{ z^{2k+1} w^{2l} }{  \Gamma(k+L+\frac32) \Gamma(n-k)  \Gamma(n-l+\frac12) \Gamma(l+L+1) }  \nonumber
\\
&\quad -   \frac{  \Gamma(2n+2L+2) }{ 2^{2L+2n+1} } \sum_{k=0}^{N-1} \sum_{l=0}^{k} \frac{ w^{2k+1} z^{2l} }{  \Gamma(k+L+\frac32) \Gamma(n-k)  \Gamma(n-l+\frac12) \Gamma(l+L+1) } ;
\end{align}
see \cite[Lem~3.1]{BF22}.
Next, we define
\begin{equation}
\widetilde{\kappa}_N^{ \rm si }(z,w)=\frac{ (zw)^{2L} }{( (1+z^2) (1+w^2) )^{n+L-\frac12}}\kappa_N^{ \rm si }(z,w). 
\end{equation}
As in \S\ref{Section_spherical GinOE}, the correlation kernel can be effectively analysed using the incomplete beta function $I_x$, which admits the expression 
\begin{equation}
I_x (m,n-m+1) = \sum_{j=m}^n \binom{n}{j} x^j (1-x)^{n-j}.
\end{equation}
The following identity was found in \cite[Prop.~1.1]{BF22}.

\begin{proposition} \label{Prop_CDI si}
Let 
\begin{equation} 
\zeta:=\frac{z w}{1+z w}, \qquad \eta:=\frac{w^2}{1+w^2}.
\end{equation}
Then we have
 \begin{align}
\partial_z \widetilde{\kappa}_N^{ \rm si }(z,w)= I_N(z,w)-II_N(z,w)-III_N(z,w),
 \end{align}
 where 
 \begin{align}
 I_N(z,w) & = \frac{(1+zw)^{2n+2L-1} }{(1+z^2)^{n+L+\frac12}(1+\eta^2)^{n+L-\frac12} }  (2n+2L+1)(n+L) \nonumber
 \\
 &\quad \times \Big( I_{\zeta}(2L,2n)- I_{ \zeta }(2N+2L,2n-2N) \Big),
 \end{align}
 \begin{align}
 II_N(z,w)&= \frac{z^{2N+2L}}{2^{2L+2n} (1+z^2)^{n+L+\frac12} } \frac{ 	 \pi \, \Gamma(2n+2L+2) }{  \Gamma(N+L+\frac12) \Gamma(n-N) \Gamma(n+L+\frac12)  }  \nonumber
 \\
 &\quad \times \Big( I_{\eta}(L,n+\tfrac12)- I_{\eta}(N+L,n-N+\tfrac12) \Big),
 \end{align}
 and 
 \begin{align}
III_N(z,w) & = \frac{z^{2L-1}}{2^{2L+2n} (1+z^2)^{n+L+\frac12} } \frac{ 	 \pi \,  \Gamma(2n+2L+2) }{ \Gamma(n+\frac12) \Gamma(L) \Gamma(n+L+\frac12) }  \nonumber
\\
&\quad \times \Big( I_{\eta }(L+\tfrac12,n)- I_{\eta }(N+L+\tfrac12,n-N) \Big).
 \end{align}
\end{proposition}
From this result, we can again observe a common feature relating the correlation kernels of the symplectic and complex ensembles; namely, up to a weight function, the term $I_N$ agrees with the correlation kernel of the complex spherical induced ensemble; see \cite[\S2.5]{BF22a}.
Also, one can use the known asymptotic behaviours of the incomplete beta functions (that can be found for instance in \cite[\S11.3.3]{Te96}) to derive various scaling limits. 
As a consequence, the bulk and edge universality with the limiting pre-kernels \eqref{kappa S bulk} and \eqref{kappa S edge} were shown in \cite{BF22}. 
Furthermore, in the regime where $L \ge -1$ is fixed, the scaling limit at the origin turns out to be again universal, with the limiting form \eqref{kappa insertion}. 

\begin{remark}
The recent work \cite{HKGG22} on winding number statistics for chiral random matrices encounters the average over GinSE matrices $K_1,K_2$ of the ratio of determinants $$\det(\alpha_1 \mathbb I_N + K_1^{-1} K_2)/ \det( \alpha_2 \mathbb I_N + K_1^{-1} K_2)$$ (as well as higher order products). As noted in this reference, $K_1^{-1} K_2$ is a member of spherical GinSE, and so the eigenvalues have the PDF (\ref{pdf of siGinSE}) with $n=N, L=0$, facilitating the computation of the average.
\end{remark}

Using \eqref{ZN S SOP} and \eqref{skew norm spherical}, one can express the partition function $Z_N^{ \mathbb{H} }( W^{ \rm si } )$ in terms of the Barnes $G$-function as 
\begin{align}
 Z_N^{ \mathbb{H} }( W^{ \rm si } ) 
&=  \Big( \frac{2^{2n+2L+1} }{  \Gamma(2n+2L+2)  } \Big)^N  \nonumber
\\
&\quad \times \frac{ G(N+L+1) }{ G(L+1) } \frac{ G(N+L+3/2) }{ G(L+3/2) } \frac{ G(n+1) }{ G(n-N+1) } \frac{ G(n+3/2) }{ G(n-N+3/2) }. 
\end{align}
Then the well-known asymptotic behaviour of the Barnes $G$-function 
(see \cite[Th.~1]{FL01}) allows the large $N$ expansion of $ Z_N^{ \mathbb{H} }( W^{ \rm si } ) $ to be calculated explicitly. 
In the scaling $L=(\alpha_1-1)N$ and $n=\alpha_2 N-1$, where the associated droplet is an annulus, one can directly check that the general asymptotic result \eqref{ZN symp exp} with the Euler characteristic $\chi=0$ holds. 
On the other hand, for the spherical case when $L=0, n=N$, the limiting eigenvalue support is the whole complex plane with 
density $1/(\pi(1 + |z|^2)^2)$. 
In this case, it follows from straightforward computations that
\begin{align}
\log  Z_N^{ \mathbb{H} }( W^{ \rm si } ) \Big|_{L=0, n=N} & = -N^2 -\frac12 N \log N + \Big( \frac{\log (4\pi^3)}{2}-1 \Big) N -\frac1{12} \log N  \nonumber
\\
&\quad -\frac{13}{24}+ \frac{5\log2}{12} +\zeta'(-1) + \frac{1}{12N}-\frac{61}{1440N^2}+O( \frac{1}{N^3} ).  
\end{align}
This exhibits the universal coefficient $-\chi/24$ of the $\log N$ in \eqref{ZN symp exp}, with the Euler characteristic $\chi=2$ for the sphere. Also, we check that the ${\rm O}(N)$ term is consistent with the general form in \eqref{ZN symp exp}; see \cite[\S4]{FF11} for the complex counterparts.

\subsection{Truncations of Haar real quaternion symplectic matrices}

We now consider the truncations of unitary matrices with real quaternion elements, equivalently, of the unitary symplectic matrices.
Let $A_N$ be the $N \times N$ sub-block of unitary quaternion matrix drawn from
the quaternionic unitary group of size $(n+N) \times (n+N)$. 
It was derived in \cite[\S2.3]{Fo16} that the eigenvalue PDF of $A_N$ follows the law \eqref{Gibbs symplectic} with 
\begin{equation} \label{W truncated s}
W^{\rm{t}}(z,\bar{z})=\begin{cases}
    -(n-1/2)\log(1-|z|^2) &\textup{if }|z|<1,
    \\
    \infty &\textup{otherwise}.
\end{cases}
\end{equation}
This form of potential makes all the eigenvalues completely contained inside the unit disk $|z| < 1$.

From the Coulomb gas viewpoint, the potential \eqref{W truncated s} can be interpreted as a situation under the presence of a hard edge. 
In this situation, the equilibrium problem should be treated depending on the position of the hard edge. 
To be more precise, if the hard edge is built inside the droplet, the mass outside the hard edge becomes lying on the boundary of the droplet, which makes the equilibrium measure no longer absolutely continuous with respect to the area measure. 
Namely, in this case, the equilibrium measure is a combination of two and one-dimensional measure, where the latter is given in terms of the balayage measure. 
On the other hand, if the hard edge is built outside the droplet, the resulting equilibrium measure is not affected by the hard edge. 

In the scaling $n=(1-\alpha)/\alpha N$, the empirical measure converges to the equilibrium measure associated with the potential 
\begin{equation}
Q^{\rm t} (z)= \frac{\alpha-1}{\alpha} \log(1-|z|^2).  
\end{equation}
Then it follows from \eqref{eq msr form} and \eqref{droplet annulus} that 
\begin{equation}
\lim_{N \to \infty} \frac{\rho^{\rm t}_{(1),N}(z)}{N} \Big|_{\alpha=N/(N+n)} = \frac{1-\alpha}{ \pi \alpha } \frac{1}{(1-|z|^2)^2}  \chi_{|z|<\sqrt{\alpha}}; 
\end{equation}
cf. \eqref{3.26A}. 
We mention that 
\begin{equation}
\lim_{\alpha \to 0}Q^{\rm t}(\sqrt{\alpha}z) =|z|^2. 
\end{equation}
This means that after the scaling $z_j \mapsto z_j/\sqrt{\alpha}$, the eigenvalue statistics of truncated ensembles tends to the GinSE as $\alpha \to 0$, equivalently, $n \gg N$.  
On the other hand, in the regime when $n$ is fixed while $N\to \infty$, the limiting global scaled potential has its value $0$ inside the unit disk, and $\infty$ outside the disk. 
The associated equilibrium measure is a uniform distribution on the unit circle. 
This is consistent with the fact that when $n \to 0$, the truncated ensemble corresponds to the circular symplectic ensemble \cite[\S2.6]{Fo10}. 

Turning to the correlation functions, Proposition~\ref{Prop_kappa structure} gives that  \cite[Eq.~(14)]{KL21}
\begin{equation} \label{kappaN t}
\kappa_N^{ \rm t }(z,w)= \frac{B(1/2,n)}{ \pi }  \sum_{ 0\le l \le k \le N-1 } \frac{ z^{2k+1}w^{2l}-z^{2l} w^{2k+1} }{ B(k+3/2,n) B(l+1,n)  }.
\end{equation}
As an analogue of Propositions~\ref{Prop_CDI e}, \ref{Prop_CDI i} and \ref{Prop_CDI si}, we have the following. 
Let us stress that this proposition has not been reported in previous literature.

\begin{proposition} \label{Prop_CDI t}
We have
\begin{align}
\frac{1-z^2}{2n}\partial_z  \kappa_N^{\rm{t}}(z,w) & =\frac{n+1}{n}z \, \kappa_N^{\rm{t}}(z,w)  + \frac{2n+1}{2\pi} \frac{1- I_{zw}(2N,2n+2)}{(1-zw)^{ 2n+2 } }  \nonumber
\\
&\quad -\frac{ z^{2N} }{ \sqrt{\pi} } \frac{ \Gamma(n+N+3/2) }{ \Gamma(n+1/2)  \Gamma(N+1/2)  }   \frac{1-I_{w^2}(N,n+1) }{ (1-w^2)^{n+1} }. \label{CDI for kappa t}
\end{align}
\end{proposition}
\begin{proof}
The proof is similar in spirit to that of Proposition~\ref{Prop_CDI si}, which can be found in \cite{BF22}. 
Let us write
\begin{equation}
G_N^{\rm{t}}(z,w):= \sum_{k=0}^{N-1} \sum_{l=0}^k \frac{ \Gamma(n+k+3/2)\Gamma(n+l+1) }{ \Gamma(k+3/2) \Gamma(l+1) } z^{2k+1} w^{2l}. 
\end{equation}
By differentiating this expression, we have
\begin{align*}
\partial_z  G_N^{\rm{t}}(z,w)
&=  2z\sum_{k=0}^{N-1} \sum_{l=0}^k \frac{ \Gamma(n+k+3/2)\Gamma(n+l+1) }{ \Gamma(k+1/2) \Gamma(l+1) } z^{2k-1} w^{2l}
\\
&=2 \frac{ \Gamma(n+3/2)\Gamma(n+1) }{ \Gamma(1/2)  }  +  2z\sum_{k=0}^{N-2} \sum_{l=0}^{k+1} \frac{ \Gamma(n+k+5/2)\Gamma(n+l+1) }{ \Gamma(k+3/2) \Gamma(l+1) } z^{2k+1} w^{2l}.
\end{align*}
Rearranging the summations, the last term can be rewritten as  
\begin{align*}
&\quad  2z\sum_{k=0}^{N-2} \sum_{l=0}^{k+1} \frac{ \Gamma(n+k+5/2)\Gamma(n+l+1) }{ \Gamma(k+3/2) \Gamma(l+1) } z^{2k+1} w^{2l}
\\
&= 2z\sum_{k=0}^{N-1} \sum_{l=0}^{k} \frac{ \Gamma(n+k+5/2)\Gamma(n+l+1) }{ \Gamma(k+3/2) \Gamma(l+1) } z^{2k+1} w^{2l}
\\
& \quad -  2 \sum_{l=0}^{N-1} \frac{ \Gamma(n+N+3/2)\Gamma(n+l+1) }{ \Gamma(N+1/2) \Gamma(l+1) } z^{2N} w^{2l}+ 2\sum_{k=1}^{N-1}  \frac{ \Gamma(n+k+3/2)\Gamma(n+k+1) }{ \Gamma(k+1/2) \Gamma(k+1) } (zw)^{2k}.
\end{align*}
Here, we have 
\begin{align*}
&\quad \sum_{k=0}^{N-1} \sum_{l=0}^{k} \frac{ \Gamma(n+k+5/2)\Gamma(n+l+1) }{ \Gamma(k+3/2) \Gamma(l+1) } z^{2k+1} w^{2l}
\\
&= \sum_{k=0}^{N-1} \sum_{l=0}^{k} ((n+1)+(k+1/2)) \frac{ \Gamma(n+k+3/2)\Gamma(n+l+1)  }{ \Gamma(k+3/2) \Gamma(l+1) } z^{2k+1} w^{2l}
\\
&= (n+1) G_N^{\rm{t}}(z,w) + \frac{z}{2}  \partial_z G_N^{\rm{t}}(z,w) . 
\end{align*}
Combining all of the above, we obtain
\begin{align}
(1-z^2)\partial_z  G_N^{\rm{t}}(z,w) =2(n+1)z G_N^{\rm{t}}(z,w) &-  2 \sum_{l=0}^{N-1} \frac{ \Gamma(n+N+3/2)\Gamma(n+l+1) }{ \Gamma(N+1/2) \Gamma(l+1) } z^{2N} w^{2l} \nonumber
\\
&\quad + 2\sum_{k=0}^{N-1}  \frac{ \Gamma(n+k+3/2)\Gamma(n+k+1) }{ \Gamma(k+1/2) \Gamma(k+1) } (zw)^{2k}. \label{G zw deri}
\end{align}

We now compute $\partial_z  G_N^{\rm{t}}(w,z)$.
We begin with writing
\begin{align*}
 \partial_z  G_N^{\rm{t}}(w,z)
&= 2z \sum_{k=0}^{N-1} \sum_{l=0}^{k} \frac{ \Gamma(n+k+3/2)\Gamma(n+l+2) }{ \Gamma(k+3/2) \Gamma(l+1) } w^{2k+1} z^{2l}
\\
&\quad - 2 \sum_{k=0}^{N-1} \frac{ \Gamma(n+k+3/2)\Gamma(n+k+2) }{ \Gamma(k+3/2) \Gamma(k+1) } (zw)^{2k+1}. 
\end{align*}
Note here that 
\begin{align*}
&\quad \sum_{k=0}^{N-1} \sum_{l=0}^{k} \frac{ \Gamma(n+k+3/2)\Gamma(n+l+2) }{ \Gamma(k+3/2) \Gamma(l+1) } w^{2k+1} z^{2l}
\\
&= \sum_{k=0}^{N-1} \sum_{l=0}^{k} ((n+1)+l)\frac{ \Gamma(n+k+3/2)\Gamma(n+l+1) }{ \Gamma(k+3/2) \Gamma(l+1) } w^{2k+1} z^{2l}
\\
&= (n+1) G_N^{\rm{t}}(w,z) + \frac{z}{2}  \partial_z G_N^{\rm{t}}(w,z) .
\end{align*}
We have shown that 
\begin{align}\label{G wz deri}
(1-z^2)\partial_z  G_N^{\rm{t}}(w,z) &= 2(n+1)z G_N^{\rm{t}}(w,z) \nonumber
\\
&\quad - 2 \sum_{k=0}^{N-1} \frac{ \Gamma(n+k+3/2)\Gamma(n+k+2) }{ \Gamma(k+3/2) \Gamma(k+1) } (zw)^{2k+1}.
\end{align}
Then by combining \eqref{kappaN t}, \eqref{G zw deri} and \eqref{G wz deri}, we conclude the desired identity. 
\end{proof}

As before, an important point to note here is the form of the second term in \eqref{CDI for kappa t}, which agrees with the kernel of the complex counterpart up to a weight function; see \cite[Eq.~(2.59)]{BF22}. 
Furthermore, as in \S\ref{Section_SI GinSE}, Proposition~\ref{Prop_CDI t} and the knowledge of the uniform asymptotic expansion of the incomplete beta functions can be used to derive scaling limits. 
In the regime of strong non-unitarity when $n=O(N)$, the universal bulk and edge scaling limits \eqref{kappa S bulk} and \eqref{kappa S edge} again appears. 
On the other hand, in the regime of weak non-unitarity when $n$ is fixed while $N \to \infty$, a qualitatively distinct scaling limit arises at the edge of the spectrum. 
In particular, it was shown in \cite[Th.~6.13]{KL21} that for $x>0$ and $y \in \mathbb {R}$,
\begin{equation} \label{limiting density weak t}
\lim_{N \to \infty} \frac{1}{N^2} \rho_{(1),N}^{ \rm t }\Big( 1-\frac{x+iy}{N} \Big) = \frac{2^{4n} y x^{2n+1} }{ \pi \Gamma(2n) } \int_{(0,1)^2} s^{2n+1} t^n e^{-2s(1+t)x} \sin(2s(1-t) y)\,ds\,dt.  
\end{equation}
In \cite{KL21}, instead of Proposition~\ref{Prop_CDI t}, a contour integral representation of $\kappa_N^{ \rm t}$ was used to derive various scaling limits including those away from the real axis.  
We mention too that the scaling limit \eqref{limiting density weak t} also appears in the context of the GinSE with hard edge; see \cite[Th.~2.4]{BES23}.

\subsection{Products of GinSE}
As in \S~\ref{Section_product GinOE}, to describe the products of GinSE in the most general setup, we begin with the $N \times N$ square matrices $\tilde{G}_i$ whose element distribution is proportional to
$$
| \det \tilde{G}_i \tilde{G}_i^\dagger|^{\nu_i} e^{-{\rm Tr} \, \tilde{G}_i \tilde{G}_i^\dagger },
$$
where $\nu_i > 0$ are the differences between matrix dimensions. 
Then the eigenvalue PDF can be computed for a general $M$ and $\nu_j$ (see \cite{Ip13,AI15}). 
It is of the form \eqref{Gibbs symplectic} with 
\begin{equation}
W^{\nu}(z,\bar{z})=-\frac12\log\MeijerG[\bigg]{M}{0}{0}{M}{-}{2\nu_1, \dots, 2\nu_{M-1},2\nu_M}{2^M|z|^2},
\end{equation}
where $G^{M,0}_{0,M}$ is the Meijer $G$-function as previously discussed in Remark~\ref{Remark Meiger G}. 
For the special case when $M=2$ and $\nu_1=\nu_2=0$, this can be written in terms of the Bessel function $K_0$.  
Furthermore, as mentioned in \cite[Remark 2.18.3]{BF22a}, the case $M=2$ permits a generalisation using a non-Hermiticity parameter; see \cite{Ak05}. 

Turning to the Coulomb gas perspective, we first note that the well-known asymptotic behaviour \cite[Eq.~(2.78)]{BF22a} of the Meijer $G$-function implies 
\begin{equation*}
-\frac{1}{2N}\log\MeijerG[\bigg]{M}{0}{0}{M}{-}{2\nu_1, \dots, 2\nu_{M-1},2\nu_M}{(2N)^M|z|^2} \sim  M |z|^{2/M} - \frac{2(\nu_1+\dots+\nu_M)}{MN}\log|z|.
\end{equation*}
Therefore, in the scaling $(\nu_1+\dots+\nu_M)/N \sim M \alpha$ ($\alpha \ge 0$), the general formula \eqref{global density s} can again be applied to the product ensembles, which gives rise to
 \begin{equation}
 \lim_{N \to \infty} N^{M- 1} \rho_{(1),N}( N^{M/2} z) = { |z|^{-2+2/M} \over \pi M}\Big ( \chi_{|z|<(\alpha +1)^{M/2}} - \chi_{|z|<\alpha^{M/2}} \Big ).
\end{equation} 

The associated skew orthogonal polynomials can be constructed by Proposition~\ref{Prop_SOP for radial} with the evaluation
\begin{equation}
 \int_{ \C } |z|^{2k} e^{-2W^\nu(z,\bar{z})} \,d^2 z = \frac{\pi}{ 2^{M(k+1)} } \prod_{l=1}^M \Gamma(2\nu_l+k+1). 
\end{equation} 
Furthermore, by Proposition~\ref{Prop_kappa structure}, we have
\begin{equation}
\kappa_N^{ M,\nu }(z,w)= \frac{2^{M-1-2 \sum_j \nu_j }}{ \pi^{ 2-M/2 } 
 }\sum_{ 0\le l \le k \le N-1 } \frac{ z^{2k+1}w^{2l}-z^{2l} w^{2k+1} }{  \prod_{j=1}^M \Gamma(\nu_j+k+3/2) \Gamma(\nu_j+l+1)  }.  
\end{equation}
As we have illustrated, the main idea to analyse various pre-kernels presented above is to write down proper differential equations and then derive their large $N$ limit. 
This technique can be extended to the present case. 
Contrary to the previous cases, the resulting differential equation is of order $M$.
In case of $M=2$, the associated second order differential equation was found and used in \cite{Ak05} to derive the limiting correlation kernel at the origin; see also \cite{AEP22}.  
A similar situation arises in the study of the Mittag-Leffler ensemble with the potential \eqref{W ML} and it was found in \cite{ABK22} that the associated pre-kernel satisfies a fractional differential equation of order $1/b$. 
On the other hand, as expected from the structure \eqref{radial density}, the analysis for the radial density is considerably simplified; see \cite{Ip13,IK14,AIS14}.  
We also mention that in the same spirit as Remark~\ref{Rem_product GinOE}, the Lyapunov and stability exponents of the products of GinSE are available in the literature \cite{Ka14,Ip15,FZ18}.

\subsection*{Acknowledgements}
	This research is part of the program of study supported
	by the Australian Research Council Discovery Project grant DP210102887.
	SB was partially supported by the National Research Foundation of Korea grant NRF-2019R1A5A1028324, Samsung Science and Technology Foundation grant SSTF-BA1401-51, and KIAS Individual via the Center for Mathematical Challenges at Korea Institute for Advanced Study grant SP083201.

\nopagebreak

\providecommand{\bysame}{\leavevmode\hbox to3em{\hrulefill}\thinspace}
\providecommand{\MR}{\relax\ifhmode\unskip\space\fi MR }
\providecommand{\MRhref}[2]{%
  \href{http://www.ams.org/mathscinet-getitem?mr=#1}{#2}
}
\providecommand{\href}[2]{#2}

\end{document}